\documentclass[letterpaper,11pt]{article}
\usepackage{fullpage}
\usepackage[ansinew]{inputenc}
\usepackage{verbatim}
\usepackage{graphicx}
\usepackage{tabularx}
\usepackage{array}
\usepackage{delarray}
\usepackage{dcolumn}
\usepackage{float}
\usepackage{amsmath}
\usepackage{psfrag}
\usepackage{booktabs}
\usepackage{multirow,hhline}
\usepackage{amstext}
\usepackage{amssymb}
\usepackage{fancyhdr}
\usepackage{setspace}
\usepackage{amsthm}
\usepackage{makeidx}
\usepackage{tikz}
\usepackage{lipsum}
\usepackage{url}
\usepackage{adjustbox}
\usepackage{enumitem}
\usepackage{ragged2e}
\usepackage{booktabs}
\usepackage{rotating}
\usepackage{threeparttable}
\usepackage{pdflscape}

\RequirePackage[round,authoryear,longnamesfirst]{natbib}
\usetikzlibrary{arrows,calc}
\makeindex
\allowdisplaybreaks
\numberwithin{equation}{section}
\newtheorem{lem}{Lemma}[section]
\newtheorem{thm}{Theorem}[section]

\newtheorem{ass}{Assumption}

\renewcommand{\citep}[1]{\citeauthor{#1}, \citeyear{#1}}

\newcommand{\Supp}{\text{Supp}}
\newcommand{\indep}{\perp\!\!\!\perp}

\newcommand{\convP}{\stackrel{p}{\longrightarrow}}
\newcommand{\convD}{\rightsquigarrow}

\newcommand{\N}{\mathcal{N}}
\newcommand{\eps}{\varepsilon}

\renewcommand{\epsilon}{\varepsilon}

\DeclareMathOperator*{\argmin}{arg\,min}
\DeclareMathOperator*{\plim}{plim}

\makeatletter
\newcommand*{\rom}
[1]{\expandafter\@slowromancap\romannumeral #1@}
\makeatother

\title{Bootstrap Inference for Quantile Treatment Effects in Randomized Experiments with Matched Pairs\thanks{\protect\doublespacing We thank the editor and two referees for extensive comments which have led to many improvements. We mention with thanks Kengo Kato, Yu-Chin Hsu, Shuping Shi, and Jun Yu for useful comments and suggestions. We are grateful to David McKenzie for providing the Stata code of the matching algorithm used in the empirical application, and to Esther Duflo and Cynthia Kinnan for providing the empirical data in the early versions of the paper. Jiang acknowledges support from MOE (Ministry of Education in China) Project of Humanities and Social Sciences (Project No.18YJC790063). Liu acknowledges financial support from the Chinese Ministry of Education Project of Humanities and Social Sciences (No. 18YJC790005) and the National Natural Science Foundation of China (No.72003171). Zhang acknowledges financial support from Singapore Ministry of Education Tier 2 grant under grant MOE2018-T2-2-169 and a Lee Kong Chian Fellowship. Phillips acknowledges support from NSF Grant No. SES 18-50860, a Kelly Fellowship at the University of Auckland, and a Lee Kong Chian Fellowship. Any and all errors are our own.\vspace{1.3mm}}
	\\ \vspace{2mm}
}
\author{Liang Jiang\thanks{\protect\doublespacing Fanhai International School of Finance, Fudan University.\ E-mail~address: jiangliang@fudan.edu.cn.\vspace{1.3mm}} \and Xiaobin Liu \thanks{\protect\doublespacing The corresponding author. School of Economics, Academy of Financial Research, Zhejiang University.\ E-mail~address:  liuxiaobin@zju.edu.cn.\vspace{1.3mm}} \and Peter C.B. Phillips \thanks{\protect\doublespacing  Yale University, University of Auckland, University of Southampton, and Singapore Management University. \ E-mail~address: peter.phillips@yale.edu \vspace{1.3mm}} \and Yichong Zhang\thanks{\protect\doublespacing Singapore Management University.\ E-mail~address: yczhang@smu.edu.sg.}
	\date{}
}

\begin{document}
	\maketitle
	\newpage
	\begin{abstract}
		This paper examines methods of inference concerning quantile treatment effects (QTEs) in randomized experiments with matched-pairs designs (MPDs). Standard multiplier bootstrap inference fails to capture the negative dependence of observations within each pair and is therefore conservative. Analytical inference involves estimating multiple functional quantities that require several tuning parameters. Instead, this paper proposes two bootstrap methods that can consistently approximate the limit distribution of the original QTE estimator and lessen the burden of tuning parameter choice. Most especially, the inverse propensity score weighted multiplier bootstrap can be implemented without knowledge of pair identities.

		
		\bigskip
		
	\end{abstract}
	\newpage
\section{Introduction}
Matched-pairs designs (MPDs) have recently seen widespread and increasing use in various randomized experiments conducted by economists. By MPD we mean a randomization scheme that first pairs units based on the closeness of their baseline covariates and then randomly assigns one unit in the pair to be treated. In development economics, researchers routinely pair villages, neighborhoods, microenterprises, or townships in their experiments (\citep{BDGK15}; \citep{crepon2015}; \citep{glewwe2016}; \citep{groh2016}). In labor economics, especially in the field of education, researchers pair schools or students to evaluate the effects of various education interventions (\citep{angrist2009}; \citep{beuermann2015}; \citep{fryer2017management}; \citep{fryer2017}; \citep{bold2018}; \citep{fryer2018}). \cite{B09} surveyed leading experts in development field experiments and reported that 56\% of them explicitly match pairs of observations on baseline characteristics.

Researchers often use randomized experiments to estimate quantile treatment effects (QTEs) as well as average treatment effects (ATEs). Quantile effects can capture heterogeneity in both the sign and magnitude of treatment effects, which may vary according to position within the distribution of outcomes. A common practice in conducting inference on QTEs is to use bootstrap rather than analytical methods because the latter usually require tuning parameters in implementation. However, in MPDs, the treatment statuses are negatively \textit{dependent} within pairs because exactly half of the units are treated; covariates across pairs are \textit{dependent} as well because of pair matching. Neither the standard multiplier bootstrap nor bootstrapping the pairs mimics such dependence structure. This difficulty raises the question of how to conduct bootstrap inference for QTEs in MPDs in a manner that mitigates these shortcomings.

To tackle these shortcomings we propose two bootstrap inference methods: the gradient bootstrap and the inverse propensity score weighted (IPW) multiplier bootstrap. We first show that the gradient bootstrap can consistently approximate the limit distribution of the QTE estimator under MPDs uniformly over a compact set of quantile indexes. \cite{h17} proposed using the gradient bootstrap for the cluster-robust inference in linear quantile regression models. Like \cite{h17}, we rely on the gradient bootstrap to avoid estimating the Hessian matrix that involves the infinite-dimensional nuisance parameters. The gradient bootstrap procedure is therefore free of tuning parameters. On the other hand and differing from \cite{h17}, we construct a specific perturbation of the score based on pair and adjacent pairs of observations, which can capture the dependence structure in the original data.

To implement our gradient bootstrap method, researchers need to know the identities of pairs. Such information may not be available when they are using an experiment that was run by other investigators in the past and the randomization procedure may not have been fully described. For example, publicly available datasets for papers such as \cite{panagopoulos2008} and \cite{butler2010} contain no information on pair identities.\footnote{\protect\doublespacing Both datasets are available in the data archive of the Institute for Social and Policy Studies at Yale University (https://isps.yale.edu/research/data).} \cite{B09} also pointed out that many papers in existing experiments do not describe the randomization procedure in detail.

To address these issues, we next propose an IPW multiplier bootstrap, which can be implemented without the knowledge of pair identities. We show that such a bootstrap can consistently approximate the limit distribution of the QTE estimator under MPDs. There is a cost to not using information about pair identities as the method requires one tuning parameter for the nonparametric estimation of the propensity score. In spite of this additional cost, this multiplier bootstrap method still has an advantage over direct analytic inference because practical implementation of the latter requires more than one tuning parameter.

%


The contributions in the present paper relate to  other recent research. \cite{BRS19} first pointed out that in MPDs the two-sample $t$-test for the null hypothesis that the ATE equals a pre-specified value is conservative. They then proposed adjusting the standard error of the estimator and studied the validity of the permutation test. This paper complements those results by considering the QTEs and by developing new methods of bootstrap inference. Unlike the permutation test, our methods of bootstrap inference do not require studentization, which is cumbersome in the QTE context. In addition, our multiplier bootstrap method complements their results by providing a way to perform inference relating to both ATEs and QTEs when pair identities are unknown. In other work, \cite{B19} investigated the optimality of MPDs in randomized experiments. \cite{ZZ20} considered bootstrap inference under covariate-adaptive randomization. A key difference in our contribution is that in MPDs the number of strata is proportional to the sample size, whereas in covariate-adaptive randomization that number is fixed. In consequence, the present work uses fundamentally different asymptotic arguments and bootstrap methods from those employed by \cite{ZZ20}. The present paper also fits within a growing literature that studies inference in randomized experiments (e.g., \cite{HHK11}, \cite{athey2017}, \cite{abadie2018}, \cite{BCS17}, \cite{T18}, and \cite{BCS18}, among others).

The remainder of the paper is organized as follows. Section \ref{sec:setup} describes the model setup and notation. Section \ref{sec:est} develops the asymptotic properties of our QTE estimator. In Section \ref{sec:bootstrap} we study the naive multiplier bootstrap, the naive multiplier bootstrap of the pairs, the gradient bootstrap, and the IPW multiplier bootstrap. Section \ref{sec:compute} provides computational details and recommendations for practitioners. Section \ref{sec:sim} reports simulation results. Section \ref{sec:app} provides an empirical application of our methods of bootstrap inference to the data in \cite{groh2016}, examining both the ATEs and QTEs of macroinsurance on consumption and profits. Section \ref{sec:concl} concludes. Proofs of all results and additional simulations are in the Online Supplement.

\section{Setup and Notation}
\label{sec:setup}
Denote the potential outcomes for treated and control groups as $Y(1)$ and $Y(0)$, respectively. Treatment status is written as $A$, where $A=1$ is treated and $A=0$ is untreated. The researcher only observes $\{Y_i,X_i,A_i\}_{i=1}^{2n}$ where $Y_i = Y_i(1)A_i + Y_i(0)(1-A_i)$, and $X_i \in \Re^{d_x}$ is a collection of baseline covariates, where $d_x$ is the dimension of $X$. The parameter of interest is the $\tau$th QTE, denoted as
\begin{align*}
\quad q(\tau) = q_1(\tau) - q_0(\tau),
\end{align*}
where $q_1(\tau)$ and $q_0(\tau)$ are the $\tau$th quantiles of $Y(1)$ and $Y(0)$, respectively.
The testing problems of interest involve single, multiple, or even a continuum of quantile indexes, as in the following null hypotheses
\begin{align*}
&\mathcal{H}_0: q(\tau) = \underline{q} \quad \text{versus} \quad q(\tau) \neq \underline{q}, \\
&\mathcal{H}_0: q(\tau_1) - q(\tau_2) = \underline{q} \quad \text{versus} \quad q(\tau_1) - q(\tau_2) \neq \underline{q}, \;\textrm{and} \\
&\mathcal{H}_0: q(\tau) = \underline{q}(\tau)~\forall \tau \in \Upsilon \quad \text{versus} \quad q(\tau) \neq \underline{q}(\tau)~\text{for some}~\tau \in \Upsilon,
\end{align*}
for some pre-specified value $\underline{q}$ or function $\underline{q}(\tau)$, where $\Upsilon$ is some compact subset of $(0,1)$.

The units are grouped into pairs based on the closeness of their baseline covariates, which is now made clear. Pairs of units are denoted
\begin{align*}
(\pi(2j-1),\pi(2j))~\text{for }j \in [n],
\end{align*}
where $[n] = \{1,\cdots,n\}$ and $\pi$ is a permutation of $2n$ units based on $\{X_i\}^{2n}_{i=1}$ as specified in Assumption \ref{ass:assignment1}(iv) below. Within a pair, one unit is randomly assigned to treatment and the other to control. Specifically, we make the following assumption on the data generating process (DGP) and the treatment assignment rule.
\begin{ass}
	\begin{enumerate}[label=(\roman*)]
		\item $\{Y_i(1),Y_i(0),X_i\}_{i=1}^{2n}$ is i.i.d.
		\item $\{Y_i(1),Y_i(0)\}^{2n}_{i=1} \indep \{A_i\}^{2n}_{i=1}|\{X_i\}^{2n}_{i=1}$.
		\item Conditionally on $\{X_i\}^{2n}_{i=1}$, $\{A_{\pi(2j-1)},A_{\pi(2j)}\}_{j \in [n]}$ are i.i.d. and each uniformly distributed over the values in $\{(1,0),(0,1)\}$.
		\item $\frac{1}{n}\sum_{j=1}^n\left\Vert X_{\pi(2j)} - X_{\pi(2j-1)} \right\Vert_2^r \convP 0$ for $r=1,2$.
	\end{enumerate}
	\label{ass:assignment1}
\end{ass}
Assumption \ref{ass:assignment1} is used in \cite{BRS19} to which we refer readers for more discussion. In Assumption \ref{ass:assignment1}(iv), $||\cdot||_2$ denotes Euclidean distance. However, all our results hold if $||\cdot||_2$ is replaced by any distance that is equivalent to it, such as $L_\infty$ distance, $L_1$ distance, and the Mahalanobis distance when all the eigenvalues of the covariance matrix are bounded and bounded away from zero.


\section{Estimation}
\label{sec:est}
Let $\hat{q}_1(\tau)$ and $\hat{q}_0(\tau)$ be the $\tau$th percentiles of outcomes in the treated and control groups, respectively. Then, the $\tau$th QTE estimator we consider is just
\begin{align*}
\hat{q}(\tau) = \hat{q}_1(\tau) - \hat{q}_0(\tau).
\end{align*}
For ease of notation, dependence of $\hat{q}(\tau)$, $\hat{q}_1(\tau)$, $\hat{q}_0(\tau)$ and all the other estimators on $n$ is suppressed throughout the rest of the paper. To facilitate further analysis and motivate our bootstrap procedure, we note that $\hat{q}(\tau)$ can be equivalently computed by direct quantile regression. Let
\begin{align*}
(\hat{\beta}_0(\tau),\hat{\beta}_1(\tau)) = \argmin_b \sum_{i=1}^{2n}\rho_\tau(Y_i - \dot{A}^\top b),
\end{align*}
where $\dot{A}_i = (1,A_i)^\top$ and $\rho_\tau(u) = u(\tau - 1\{u\leq 0\})$. Then, $\hat{q}(\tau) = \hat{\beta}_1(\tau)$ and $\hat{q}_0(\tau) = \hat{\beta}_0(\tau)$.

\begin{ass}
	\label{ass:reg}
	For $a=0,1$, define $F_a(\cdot)$, $F_a(\cdot|x)$, $f_a(\cdot)$, and $f_a(\cdot|x)$ as the CDF of $Y_i(a)$, the conditional CDF of $Y_i(a)$ given $X_i=x$, the PDF of $Y_i(a)$, and the conditional PDF of $Y_i(a)$ given $X_i=x$, respectively. \begin{enumerate}[label=(\roman*)]
		\item $f_a(q_a(\tau))$ is bounded and bounded away from zero uniformly over $\tau \in \Upsilon$, and $f_a(q_a(\tau)|x)$ is uniformly bounded for $(x,\tau) \in \Supp(X) \times \Upsilon$.
		\item There exists a function $C(x)$ such that
		\begin{align*}
		\sup_{\tau \in \Upsilon}|f_a(q_a(\tau) + v|x) - f_a(q_a(\tau)|x)| \leq C(x)|v| \quad \text{and} \quad \mathbb{E}C(X_i) < \infty.
		\end{align*}
		\item Let $\N_0$ be a neighborhood of 0. Then, there exists a constant $C$ such that for any $x,x' \in \Supp(X)$
		\begin{align*}
		\sup_{\tau \in \Upsilon,v\in \N_0}|f_a(q_a(\tau)+v|x') - f_a(q_a(\tau)+v|x)| \leq C||x'-x||_2
		\end{align*}
		and
		$$\sup_{\tau \in \Upsilon,v\in \N_0}|F_a(q_a(\tau)+v|x')-F_a(q_a(\tau)+v|x)| \leq C||x'-x||_2.$$
	\end{enumerate}
\end{ass}

Assumption \ref{ass:reg}(i) is a standard regularity condition widely assumed in quantile estimation. The Lipschitz conditions in Assumptions \ref{ass:reg}(ii) and \ref{ass:reg}(iii) are similar in spirit to those assumed in \citet[Assumption 2.1]{BRS19} and ensure that units that are ``close" in terms of their baseline covariates are suitably comparable. For $a = 0,1$, let $m_{a,\tau}(x,q) = \mathbb{E}(\tau - 1\{Y(a) \leq q\}|X=x)$ and $m_{a,\tau}(x) = m_{a,\tau}(x,q_a(\tau))$.

\begin{thm}
	\label{thm:est}
	Suppose Assumptions \ref{ass:assignment1} and \ref{ass:reg} hold. Then, uniformly over $\tau \in \Upsilon$,
	$$\sqrt{n}(\hat{q}(\tau) - q(\tau)) \convD \mathcal{B}(\tau),$$
	where $\mathcal{B}(\tau)$ is a tight Gaussian process with covariance kernel $\Sigma(\cdot,\cdot)$ such that
	\begin{align*}
	\Sigma(\tau,\tau') = & \frac{\min(\tau, \tau') - \tau\tau' - \mathbb{E}m_{1,\tau}(X)m_{1,\tau'}(X)}{f_1(q_1(\tau))f_1(q_1(\tau'))} +  \frac{\min(\tau, \tau') - \tau\tau' - \mathbb{E}m_{0,\tau}(X)m_{0,\tau'}(X)}{f_0(q_0(\tau))f_0(q_0(\tau'))} \\
	& + \frac{1}{2}\mathbb{E}\left(\frac{m_{1,\tau}(X)}{f_1(q_1(\tau))} - \frac{m_{0,\tau}(X)}{f_0(q_0(\tau))}\right)\left(\frac{m_{1,\tau'}(X)}{f_1(q_1(\tau'))} - \frac{m_{0,\tau'}(X)}{f_0(q_0(\tau'))}\right).
	\end{align*}
\end{thm}

Several remarks are in order. First, the asymptotic variance of $\hat{q}(\tau)$ under MPDs is
\begin{align}
\Sigma(\tau,\tau) = & \frac{\tau - \tau^2 - \mathbb{E}m^2_{1,\tau}(X)}{f^2_1(q_1(\tau))} +  \frac{\tau - \tau^2 - \mathbb{E}m^2_{0,\tau}(X)}{f^2_0(q_0(\tau))} + \frac{1}{2}\mathbb{E}\left(\frac{m_{1,\tau}(X)}{f_1(q_1(\tau))} - \frac{m_{0,\tau}(X)}{f_0(q_0(\tau))}\right)^2.
\label{eq:sigma}
\end{align}
Note further that the asymptotic variance of $\hat{q}(\tau)$ under simple random sampling (SRS)\footnote{By simple random sample, we mean the treatment status is assigned independently with probability $1/2$.} is
\begin{align}
\Sigma^\dagger(\tau,\tau) = & \frac{\tau - \tau^2}{f^2_1(q_1(\tau))} +  \frac{\tau - \tau^2 }{f^2_0(q_0(\tau))}.
\label{eq:sigmadagger}
\end{align}
It is clear that
\begin{align}
\label{eq:simplesigma}
\Sigma^\dagger(\tau,\tau) - \Sigma(\tau,\tau) = \frac{1}{2}\mathbb{E}\left(\frac{m_{1,\tau}(X)}{f_1(q_1(\tau))} + \frac{m_{0,\tau}(X)}{f_0(q_0(\tau))}\right)^2 \geq 0.
\end{align}
Equality in the last expression holds when both $m_{1,\tau}(X)$ and $m_{0,\tau}(X)$ are zero, which implies that $X$ is irrelevant to the $\tau$th quantiles of $Y(0)$ and $Y(1)$.

Second, note that $\hat{q}(\tau)$ has the same asymptotic variance as that for the QTE estimators studied by  \cite{F07} and \cite{DH14} under SRS.

Third, to provide an analytic estimate of the asymptotic variance $\Sigma(\tau,\tau)$ it is necessary at least to estimate the infinite dimensional nuisance parameters $f_1(q_1(\tau))$ and $f_0(q_0(\tau))$, which requires two tuning parameters. Hence, if a researcher is interested in testing a null hypothesis that involves $G$ quantile indexes, $2G$ tuning parameters are needed to estimate $2G$ densities, a cumbersome task in practical work; and to construct a uniform confidence band for the QTE analytically, two tuning parameters are needed at each grid point of the quantile indexes. Moreover, if pair identities are unknown, analytic methods of inference potentially require nonparametric estimation of the quantities $m_{a,\tau}(\cdot)$ for $a = 0,1$ as well. There are other practical difficulties. Nonparametric estimation is sometimes sensitive to the choice of tuning parameters and rule-of-thumb tuning parameter selection may not be appropriate for every data generating process (DGP) or every quantile. Use of cross-validation in selecting the tuning parameters is possible in principle but in practice time-consuming. These practical difficulties of analytic methods of inference provide a strong motivation to investigate bootstrap inference procedures that are much less reliant on tuning parameters.

\section{Bootstrap Inference}
\label{sec:bootstrap}

This section examines four bootstrap inference procedures for the QTEs in MPDs. We first show that both the naive multiplier bootstrap and the naive multiplier bootstrap of the pairs fail to approximate the limit distribution of the QTE estimator derived in Section \ref{sec:est}. We then propose two bootstrap methods that can consistently approximate the limit distribution of the QTE estimator.

\subsection{Naive Multiplier Bootstrap}
\label{sec:weighted}
Consider first the naive multiplier bootstrap estimators of $\hat{\beta}_0(\tau)$ and $\hat{\beta}_1(\tau)$, defining
\begin{align*}
(\hat{\beta}_0^m(\tau),\hat{\beta}_1^m(\tau)) = \argmin_b \sum_{i=1}^{2n}\xi_i\rho_\tau(Y_i - \dot{A}^\top_i b),
\end{align*}
where $\xi_i$ is the bootstrap weight defined in the next assumption.

\begin{ass}
	\label{ass:weight}
	Suppose $\{\xi_i\}_{i=1}^{2n}$ is a sequence of nonnegative i.i.d. random variables with unit expectation and variance and a sub-exponential upper tail.
\end{ass}

In practice, we generate $\xi_{i}$ independently from the standard exponential distribution. Denote $\hat{q}^m(\tau) = \hat{\beta}_1^m(\tau)$ and recall that $\hat{q}(\tau) =  \hat{\beta}_1(\tau)$.
\begin{thm}
	\label{thm:weight}
	If Assumptions \ref{ass:assignment1}--\ref{ass:weight} hold, then conditionally on the data and uniformly over $\tau \in \Upsilon$,
	\begin{align*}
	\sqrt{n}(\hat{q}^m(\tau) - \hat{q}(\tau)) \convD \mathcal{B}^m(\tau),
	\end{align*}
	where $\mathcal{B}^m(\tau)$ is a tight Gaussian process with covariance kernel $\Sigma^\dagger(\cdot,\cdot)$ such that
	\begin{align*}
	\Sigma^\dagger(\tau,\tau') =  \frac{ \min(\tau,\tau') - \tau \tau' }{f_1(q_1(\tau))f_1(q_1(\tau'))} + \frac{ \min(\tau,\tau') - \tau \tau' }{f_0(q_0(\tau))f_0(q_0(\tau'))}.
	\end{align*}
\end{thm}

Three remarks are in order. First, $\sqrt{n}(\hat{q}^m(\tau) - \hat{q}(\tau))$ and $\mathcal{B}^m(\tau)$ are viewed as processes indexed by $\tau \in \Upsilon$ and denoted by $G_n$ and $G$, respectively. Then, following \citet[Chapter 2.9]{VW96}, we say $G_n$ weakly converges to $G$ conditionally on data and uniformly over $\tau \in \Upsilon$ if
\begin{align*}
\sup_{h \in \text{BL}_1}|\mathbb{E}_\xi h(G_n) - \mathbb{E}h(G)| \convP 0,
\end{align*}
where $\text{BL}_1$ is the set of all functions $h:\ell^\infty(\Upsilon) \mapsto [0,1]$ such that $|h(z_1)-h(z_2)| \leq ||z_1-z_2||_\infty$ for every $z_1,z_2 \in \ell^\infty(\Upsilon)$, and $\mathbb{E}_\xi$ denotes expectation with respect to the bootstrap weights $\{\xi\}_{i=1}^n$.\footnote{Asymptotic measurability holds in our setting from \citet[Lemma 1.5.2]{VW96}, which requires asymptotic tightness of the bootstrap process. The latter has been established in the proof of Theorem \ref{thm:weight}. For notational simplicity, we ignore the issue of asymptotic measurability.} The same remark applies to Theorems \ref{thm:pair}, \ref{thm:boot}, and \ref{thm:ipw_w} below.

Second, $\Sigma^\dagger(\tau,\tau')$ is just the covariance kernel of the QTE estimator when simple random sampling (instead of the MPD) is used as the treatment assignment rule. It follows that the naive multiplier bootstrap fails to approximate the limit distribution of $\hat{q}(\tau)$ ($\hat{\beta}_1(\tau)$). The intuition is straightforward. Given the data, the bootstrap weights are i.i.d. and thus unable to mimic the cross-sectional dependence in the original sample.

Third, it is possible to consider the conventional nonparametric bootstrap in which the bootstrap sample is generated from the empirical distribution of the data. If the observations are i.i.d., \citet[Section 3.6]{VW96} showed that the conventional bootstrap is first-order equivalent to a multiplier bootstrap with Poisson(1) weights. However, in the current setting, $\{A_i\}_{i \in [2n]}$ are dependent. It is technically challenging to show rigorously that the above equivalence still holds and this challenge is left for future research.

\subsection{Naive Multiplier Bootstrap of the Pairs }
\label{sec:pair}

Next consider the naive multiplier bootstrap of the pairs which uses the same bootstrap multiplier for the observations within the pair. Let

\begin{align*}
(\hat{\beta}_{0}^p(\tau),\hat{\beta}_{1}^p(\tau)) = \argmin_b \sum_{i=1}^{2n}\xi_i^p\rho_\tau(Y_i - \dot{A}^\top_ib),
\end{align*}
where $\xi_i^p$ is the bootstrap weight defined in the next assumption.

\begin{ass}
	\label{ass:pair_boot}
	Suppose $\{\xi_{\pi(2j-1)}^p\}_{j=1}^{n}$ is a sequence of nonnegative i.i.d. random variables with unit expectation and variance and a sub-exponential upper tail and $\xi_{\pi(2j-1)}^p = \xi_{\pi(2j)}^p$ for $j=1,\cdots,n$.
\end{ass}

Because the units in the same pair share the same multiplier, we call this the naive multiplier bootstrap of the pairs. Denote $\hat{q}^p(\tau) = \hat{\beta}_{1}^p(\tau)$.
\begin{thm}
	\label{thm:pair}
	If Assumptions \ref{ass:assignment1}, \ref{ass:reg}, and \ref{ass:pair_boot} hold, then conditionally on the data and uniformly over $\tau \in \Upsilon$,
	\begin{align*}
	\sqrt{n}(\hat{q}^p(\tau) - \hat{q}(\tau)) \convD \mathcal{B}^p(\tau),
	\end{align*}
	where $\mathcal{B}^p(\tau)$ is a tight Gaussian process with covariance kernel $\Sigma^p(\cdot,\cdot)$ such that
	\begin{align*}
	\Sigma^p(\tau,\tau') = \frac{ \min(\tau,\tau') - \tau \tau' }{f_1(q_1(\tau))f_1(q_1(\tau'))} + \frac{ \min(\tau,\tau') - \tau \tau' }{f_0(q_0(\tau))f_0(q_0(\tau'))} - \frac{\mathbb{E}m_{1,\tau}(X_i)m_{0,\tau'}(X_i)}{f_1(q_1(\tau))f_0(q_0(\tau'))} - \frac{\mathbb{E}m_{1,\tau'}(X_i)m_{0,\tau}(X_i)}{f_0(q_0(\tau))f_1(q_1(\tau'))}.
	\end{align*}
\end{thm}


Three remarks are in order. First, Theorem \ref{thm:pair} implies that bootstrapping pairs of observations alone is unable to mimic the dependence structure in the original sample. In MPDs, covariates across pairs are dependent. However, bootstrapping pairs misses such dependence by assigning independent multipliers to pairs.\footnote{We provide an illustrative example in Section A of the supplement to show the dependence of covariates across pairs. We thank the editor for an inspiring and helpful discussion on the failure of bootstrapping pairs.}

%

Second, the variances of the original estimator $\hat{q}(\tau)$ under MPD and $\hat{q}^p(\tau)$ conditional on data have the following relationship:
\begin{align}
\label{eq:pair_var}
\Sigma^p(\tau,\tau) - \Sigma(\tau,\tau) = \frac{1}{2}\mathbb{E}\left( \frac{m_{1,\tau}(X_i)}{f_1(q_1(\tau))} - \frac{m_{0,\tau}(X_i)}{f_0(q_0(\tau))}\right)^2 \geq 0.
\end{align}

Third, the gradient bootstrap procedure proposed below is based on a similar idea and uses the same weight for the observations within the pair to construct the score $S_{n,1}^*$ defined in \eqref{eq:S1}. But in order to construct a final score that exactly mimics the dependence in the data, an extra score component, $S_{n,2}^*$ defined in \eqref{eq:S2} below, is needed. This component is constructed based on adjacent pairs of observations.

\subsection{Gradient Bootstrap Inference}
\label{sec:boot}
We develop an approximation for the asymptotic distribution of the QTE estimator via the gradient bootstrap. Let $u = \sqrt{n}(b - \beta(\tau))$ be a localizing estimation error parameter. From the derivations in Theorem \ref{thm:est}, we see that
\begin{align}
\sqrt{n}(\hat{\beta}(\tau) - \beta(\tau)) = Q^{-1}(\tau)
\begin{pmatrix}
1 & 1 \\
1 & 0
\end{pmatrix}
S_n(\tau) +o_p(1),
\label{eq:betatau}
\end{align}
where
\begin{align*}
S_n(\tau) = \begin{pmatrix}
\sum_{i=1}^{2n}\frac{A_i}{\sqrt{n}}\left(\tau- 1\{Y_i(1) \leq q_1(\tau)\}\right) \\
\sum_{i=1}^{2n}\frac{(1-A_i)}{\sqrt{n}}\left(\tau- 1\{Y_i(0)\leq q_0(\tau)\}\right)
\end{pmatrix},
\end{align*}
and
\begin{align*}
Q(\tau) = \begin{pmatrix}
f_1(q_1(\tau)) + f_0(q_0(\tau)) & f_1(q_1(\tau)) \\
f_1(q_1(\tau)) & f_1(q_1(\tau))
\end{pmatrix}.
\end{align*}

The gradient bootstrap proposes to perturb the objective function by some random error $S_n^*(\tau)$, which will be specified later. This error in turn perturbs the score function $S_n(\tau)$. The corresponding bootstrap estimator $\hat{\beta}^*(\tau)$ solves the following optimization problem
\begin{align}
\label{eq:gb}
\hat{\beta}^*(\tau) = \argmin_b \sum_{i=1}^{2n}\rho_\tau(Y_i - \dot{A}^\top_ib)- \sqrt{n}b^\top\begin{pmatrix}
1 & 1 \\
1 & 0
\end{pmatrix}S_n^*(\tau).
\end{align}
We can then show that
\begin{align}
\sqrt{n}(\hat{\beta}^*(\tau) - \beta(\tau)) = Q^{-1}(\tau)
\begin{pmatrix}
1 & 1 \\
1 & 0
\end{pmatrix}
[S_n(\tau)+S_n^*(\tau)] + o_p(1).
\label{eq:betatau*}
\end{align}
Taking the difference between \eqref{eq:betatau} and \eqref{eq:betatau*} gives
\begin{align*}
\sqrt{n}(\hat{\beta}^*(\tau) - \hat{\beta}(\tau)) = Q^{-1}(\tau)\begin{pmatrix}
1 & 1 \\
1 & 0
\end{pmatrix} S_n^*(\tau) + o_p(1).
\end{align*}
The second element of $\hat{\beta}^*(\tau)$ in \eqref{eq:gb} is the bootstrap version of the QTE estimator, which is denoted $\hat{q}^*(\tau)$. By solving \eqref{eq:gb} we avoid estimating the Hessian $Q(\tau)$, which involves infinite-dimensional nuisance parameters. Then, for the gradient bootstrap to consistently approximate the limit distribution of the original estimator $\hat{\beta}(\tau)$, we need only construct $S_n^*(\tau)$ in such a way that its weak limit given the data coincides with that of the original score $S_n(\tau)$.

Accordingly, we now show how to specify $S_n^*(\tau)$. Let $\{\eta_j\}_{j=1}^{n}$ and $\{\hat{\eta}_k\}_{k=1}^{\lfloor n/2 \rfloor}$ be two mutually independent i.i.d. sequences of standard normal random variables. Use the indexes $(j,1),(j,0)$ to denote  the indexes in $(\pi(2j-1),\pi(2j))$ with $A=1$ and $A=0$, respectively. For example, if $A_{\pi(2j)} = 1$ and $A_{\pi(2j-1)}=0$, then $(j,1) = \pi(2j)$ and $(j,0) = \pi(2j-1)$. Similarly, use indexes $(k,1),\cdots, (k,4)$ to denote the first index in $(\pi(4k-3),\cdots,\pi(4k))$ with $A=1$, the first index with $A=0$, the second index with $A=1$, and the second index with $A=0$, respectively. Now let
\begin{align*}
S_n^*(\tau)  = \frac{S_{n,1}^*(\tau) + S_{n,2}^*(\tau)}{\sqrt{2}},
\end{align*}
where
\begin{align}
\label{eq:S1}
S_{n,1}^*(\tau) = \frac{1}{\sqrt{n}}\begin{pmatrix}
\sum_{j=1}^{n}\eta_j\left(\tau- 1\{Y_{(j,1)} \leq \hat{q}_1(\tau)\}\right) \\
\sum_{j=1}^{n}\eta_j\left(\tau- 1\{Y_{(j,0)} \leq \hat{q}_0(\tau)\}\right)
\end{pmatrix},
\end{align}
and
\begin{align}
\label{eq:S2}
S_{n,2}^*(\tau) = \frac{1}{\sqrt{n}}\begin{pmatrix}
\sum_{k=1}^{\lfloor n/2 \rfloor}\hat{\eta}_k \left[\left(\tau - 1\{Y_{(k,1)} \leq \hat{q}_1(\tau)\} \right) - \left(\tau - 1\{Y_{(k,3)} \leq \hat{q}_1(\tau)\} \right)\right] \\
\sum_{k=1}^{\lfloor n/2 \rfloor}\hat{\eta}_k\left[\left(\tau - 1\{Y_{(k,2)} \leq \hat{q}_0(\tau)\} \right) - \left(\tau - 1\{Y_{(k,4)} \leq \hat{q}_0(\tau)\} \right)\right] \\
\end{pmatrix}.
\end{align}

The perturbation $S_{n}^*$ consists of two parts: $S_{n,1}^*$ and $S_{n,2}^*$. The first part $S_{n,1}^*$ is constructed by pairs of observations, based on the idea of bootstrapping the pairs. But bootstrapping the pairs alone cannot fully capture the dependence structure in the MPD, as shown in Section \ref{sec:pair}. So $S_{n,1}^*$ is adjusted by adding a term capturing the remaining dependence. This second term, $S_{n,2}^*$, is motivated by the idea of using adjacent pairs to adjust the standard error of the ATE estimator under MPDs in \cite{BRS19}.

In Section \ref{sec:compute} we show how to compute the bootstrap estimator $\hat{\beta}^*(\tau)$ directly from the sub-gradient condition of \eqref{eq:gb}. This method avoids the optimization inherent in \eqref{eq:gb} and computation is fast. The following assumption imposes the condition that baseline covariates in adjacent pairs are also `close'.
\begin{ass}
	\label{ass:pair}
	Suppose that
	$\frac{1}{n}\sum_{k =1}^{\lfloor n/2 \rfloor}\left \Vert X_{(k,l)} - X_{(k,l')}\right \Vert_2^r \convP 0$
	for $r = 1,2$ and $l,l' \in [4]$.
\end{ass}

Assumption \ref{ass:pair} and Assumption \ref{ass:assignment1}(iv) are jointly equivalent to \citet[Assumption 2.4]{BRS19}. We refer readers to \cite{BRS19} for further discussion of this assumption. In particular, \citet[Theorems 4.1 and 4.2]{BRS19} show that it is possible to implement the matching algorithm to re-order pairs so that both Assumption \ref{ass:pair} and Assumption \ref{ass:assignment1}(iv) hold automatically. We provide more detail in Section \ref{subsec:gradient computation}.

%

Define $\hat{q}^*(\tau) = \hat{\beta}_1^*(\tau)$ and recall that $\hat{q}(\tau) =  \hat{\beta}_1(\tau)$. We have the following result.
\begin{thm}
	\label{thm:boot}
	Suppose Assumptions \ref{ass:assignment1}, \ref{ass:reg}, and \ref{ass:pair} hold. Then, conditionally on the data and uniformly over $\tau \in \Upsilon$,
	$ \sqrt{n}(\hat{q}^*(\tau) - \hat{q}(\tau)) \convD \mathcal{B}(\tau)$,
	where $\mathcal{B}(\tau)$ is the same Gaussian process defined in Theorem \ref{thm:est}.
\end{thm}

Two remarks on Theorem \ref{thm:boot} are in order. First, the bootstrap estimator $\hat{q}^*(\tau)$ has the following objectives: (i) to avoid estimating densities; and (ii) to mimic the distribution of the original estimator $\hat{\beta}(\tau)$ under MPDs. Objective (i) relates to the Hessian ($Q$) and (ii) to the score ($S_n$) of the quantile regression. The gradient bootstrap provides a flexible approach to achieve both goals.



Second, to implement the gradient bootstrap, researchers need to know identities of pairs. This information may not be available when the experiment was run by others and the randomization procedure was not fully detailed. In such cases, we propose IPW multiplier bootstrap inference for the QTE, whose validity is established in the next section.

\subsection{IPW Multiplier Bootstrap Inference}
\label{sec:ipw}

In empirical research, researchers may not know the identities of pairs when they are using an experiment that was run by other investigators in the past and the randomization procedure may not have been fully described. For example, publicly available datasets for papers such as \cite{panagopoulos2008} and \cite{butler2010} contain no information on pair identities. \cite{B09} also pointed out that many papers in existing experiments do not describe the randomization procedure in detail.  In this section we establish the validity of IPW multiplier bootstrap inference for the QTE, showing that the procedure can be implemented without the knowledge of pair identities.

We use the sieve method to nonparametrically estimate the propensity score. Let $b(X)$ be the $K$-dimensional sieve basis on $X$ and $\hat{A}_i$ be the estimated propensity score for the $i$th individual. Then,
\begin{align}
\label{eq:ps}
\hat{A}_i = b^\top(X_i)\hat{\theta},
\end{align}
where
$\hat{\theta} = \argmin_\theta \sum_{i=1}^{2n}\xi_i(A_i - b^\top(X_i)\theta)^2$ and  $\xi_i$ is the bootstrap weight defined in Assumption \ref{ass:weight}.

Because the true propensity score is $1/2$, by setting the first component of $b(X)$ to unity, we have $1/2 = b^\top(X)\theta_0$ where $\theta_0 = (0.5,0,\cdots,0)^\top$. The linear probability model for the propensity score is correctly specified. It is possible to use sieve logistic regression to compute the propensity score, as done by \cite{HIR03}, \cite{F07}, and \cite{DH14}. The main benefit of using logistic regression is to guarantee that the estimated propensity score lies between zero and one. However, in MPDs, the estimated propensity score is always very close to 0.5. Therefore, for simplicity, we use a linear sieve regression here.

The IPW multiplier bootstrap estimator can be computed as
\begin{align*}
\hat{q}_{ipw}^w(\tau) = \hat{q}_{ipw,1}^w(\tau) - \hat{q}_{ipw,0}^w(\tau),
\end{align*}
where
\begin{align}
\label{eq:qipw10}
\hat{q}_{ipw,1}^w(\tau) = \argmin_q \sum_{i =1}^{2n} \frac{\xi_iA_i}{\hat{A}_i}\rho_\tau(Y_i - q) \quad
\text{and} \quad \hat{q}_{ipw,0}^w(\tau) = \argmin_q \sum_{i =1}^{2n} \frac{\xi_i(1-A_i)}{1-\hat{A}_i}\rho_\tau(Y_i - q).
\end{align}

\begin{ass}
	\begin{enumerate}[label=(\roman*)]
		\item The support of $X$ is compact. The first component of $b(X)$ is 1.
		\item $\max_{k \in [K]}\mathbb{E}b_k^2(X_i) \leq \overline{C}<\infty$ for some constant $\overline{C}>0$, where $b_k(X_i)$ is the $k$th coordinate of $b(X_i)$. $\sup_{x \in \Supp(X)}||b(x)||_{2} = \zeta(K).$
		\item $K^2 \zeta(K)^2 \log(n) = o(n)$.
		\item With probability approaching one, there exist constants $\underline{C}$ and $\overline{C}$ such that
		\begin{align*}
		0 < \underline{C}\leq \lambda_{\min}\left(\frac{1}{n}\sum_{i=1}^{2n}\xi_ib(X_i)b^\top(X_i)\right) \leq	\lambda_{\max}\left(\frac{1}{n}\sum_{i=1}^{2n}\xi_ib(X_i)b^\top(X_i)\right) \leq \overline{C} < \infty,
		\end{align*}
		where $\lambda_{\min}(\mathcal{M})$ and $\lambda_{\max}(\mathcal{M})$ denote the minimum and maximum eigenvalues of matrix $\mathcal{M}$.
		\item There exist $\gamma_1(\tau) \in \Re^K$ and $\gamma_0(\tau) \in \Re^K$ such that
		\begin{align*}
		B_{a,\tau}(x) = m_{a,\tau}(x) - b^\top(x)\gamma_a(\tau),\;\; a = 0,1,
		\end{align*}
		and $\sup_{a=0,1,\tau\in \Upsilon, x \in \Supp(X)}|B_{a,\tau}(x)| = o(1/\sqrt{n})$.
	\end{enumerate}
	\label{ass:sieve}
\end{ass}

Two remarks are in order. First, requiring $X$ to have a compact support is common in nonparametric sieve estimation.  Second, the quantity $\zeta(K)$ depends on the choice of basis functions. For example, $\zeta(K) = O(K^{1/2})$ for splines and $\zeta(K) = O(K)$ for power series.\footnote{\protect\doublespacing See \cite{c07} for a full discussion of the sieve method.} Taking splines as an example, Assumption \ref{ass:sieve}(iii) requires $K = o(n^{1/3})$. Assumption \ref{ass:sieve}(iv) is standard because $K \ll n$. Assumption \ref{ass:sieve}(v) requires that the approximation error of $m_{a,\tau}(x)$ via a linear sieve function is sufficiently small. For instance, suppose $m_{a,\tau}(x)$ is s-times continuously differentiable in $x$ with all derivatives uniformly bounded by some constant $\overline{C}$, then $\sup_{a=0,1,\tau\in \Upsilon, x \in \Supp(X)}|B_{a,\tau}(x)| = O(K^{-s/d_x})$. Assumptions \ref{ass:sieve}(iii) and \ref{ass:sieve}(v) imply that $K = n^{h}$ for some $h \in (d_x/(2s), 1/3)$, which implicitly requires $s>3d_x/2$. The choice of $K$ reflects the usual bias-variance trade-off and is the only tuning parameter that researchers need to specify when implementing this bootstrap method.

\begin{thm}
	\label{thm:ipw_w}
	Suppose Assumptions \ref{ass:assignment1}--\ref{ass:weight} and \ref{ass:sieve} hold, then conditionally on the data and uniformly over $\gamma \in \Upsilon$,
	$\sqrt{n}(\hat{q}_{ipw}^w(\tau) - \hat{q}(\tau)) \convD \mathcal{B}(\tau)$,
	where $\mathcal{B}(\tau)$ is the same Gaussian process as defined in Theorem \ref{thm:est}.
\end{thm}

To understand the need to nonparametrically estimate the propensity score in the bootstrap sample, note that there are two stages to statistical inference in randomized experiments: the design stage and the analysis stage. In the design stage, researchers can use either simple random sampling (SRS) or matched-pairs design (MPD).\footnote{Simple random sampling means that treatment status is assigned independently with probability $1/2$. Note that SRS and MPD share the same true propensity score $1/2$. But MPD achieves the strong balance that exactly half of the units are treated whereas SRS does not.} In the analysis stage, researchers can choose either the true ($1/2$ in MPDs) or the nonparametrically estimated propensity score to construct the estimator.\footnote{Specifically, we have $\hat{q}_{ipw,1}(\tau) = \argmin_q \sum_{i \in [2n]} \frac{A_i}{\hat{\pi}(X_i)}\rho_\tau(Y_i - q)$ and $\hat{q}_{ipw,0}(\tau) = \argmin_q \sum_{i \in [2n]} \frac{1-A_i}{1-\hat{\pi}(X_i)}\rho_\tau(Y_i - q)$, where $\hat{\pi}(X_i)$ is an estimator of the propensity score. When we use the true score, i.e., $\hat{\pi}(X_i) = 1/2$, we have $\hat{q}_{ipw,a} = \hat{q}_a$ as defined in the paper for $a=0,1$. However, we can also let $\hat{\pi}(X_i)$ be the  nonparametric estimator of the propensity score.} The asymptotic variances of the QTE estimator with true and estimated propensity scores under SRS are $\Sigma^\dagger(\tau,\tau)$ defined in \eqref{eq:sigmadagger} and $\Sigma(\tau,\tau)$ defined in \eqref{eq:sigma}, respectively, which are derived by \cite{H98} and \cite{F07}; the asymptotic variance of the QTE estimator with the true propensity score under MPD is $\Sigma(\tau,\tau)$, which is shown in Theorem \ref{thm:est}. These results are summarized in the following table.
\begin{table}[H]
	\centering
	\caption{Asymptotic Variances}
	\vspace{1ex}
	\begin{tabular}{c|cc}
		& True Score &  Nonparametric Score \\
		\hline
		SRS & $\Sigma^\dagger(\tau,\tau)$ & $\Sigma(\tau,\tau)$ \\
		MPD		& $\Sigma(\tau,\tau)$ & unknown \\
	\end{tabular}%
	\label{tab:est}%
\end{table}%
Note that the asymptotic variances for the QTE estimator under MPD with the true score and that under SRS with the nonparametrically estimated score are the same. If we conduct multiplier bootstrap inference, conditionally on data, the bootstrap sample of observations is independent. Therefore, in order for the multiplier bootstrap estimator to mimic the asymptotic behavior of the original estimator under MPD with the true score, we need to nonparametrically estimate the propensity score in the bootstrap sample.

The benefit of the IPW multiplier bootstrap is that it does not require knowledge of the pair identities. The cost is that we have to nonparametrically estimate the propensity score, which requires one tuning parameter and is subject to the usual curse of dimensionality. Nonetheless, we still prefer this bootstrap method of inference to the analytic approach. Analytic estimation of the standard error of the QTE estimator without the knowledge of pair identities requires nonparametric estimation of $\{m_{a,\tau}(X),f_a(q_a(\tau))\}_{a=0,1}$, which involves four tuning parameters. The number of tuning parameters further increases with the number of quantile indexes involved in the null hypothesis. To construct uniform confidence bands for QTE over $\tau$, we require $4G$ tuning parameters for grid size $G$. By contrast, implementation of the IPW multiplier bootstrap requires estimation of the propensity score only once, and thus, the use of a single tuning parameter.

Inference concerning the \emph{ATE} in MPDs can also be accomplished via a similar IPW multiplier bootstrap procedure. We can show that such a bootstrap can consistently approximate the asymptotic distribution of the ATE estimator under MPDs. This result complements that established by \cite{BRS19} because it provides a way to make inferences about the ATE in MPDs when information on pair identities is unavailable. That pair identity information is required by \cite{BRS19} in computing standard errors for their adjusted $t$-test.


\section{Computation and Guidance for Practitioners}
\label{sec:compute}
\subsection{Computation of the Gradient Bootstrap}
\label{subsec:gradient computation}
In practice, the order of pairs in the dataset is usually arbitrary and does not satisfy Assumption \ref{ass:pair}. To apply the gradient bootstrap, researchers first need to re-order the pairs. For the $j$th pair with units indexed by $(j,1)$ and $(j,0)$ in the treatment and control groups, let $\overline{X}_j = \frac1{2}\{X_{(j,1)} + X_{(j,0)}\}$. Then, let $\overline{\pi}$ be any permutation of $n$ elements that minimizes
\begin{align*}
\frac{1}{n}\sum_{j=1}^n ||\overline{X}_{\overline{\pi}(j)}-\overline{X}_{\overline{\pi}(j-1)}||_2.
\end{align*}
The pairs are re-ordered by indexes $\overline{\pi}(1),\cdots, \overline{\pi}(n)$. With an abuse of notation, we still index the pairs after re-ordering by $1,\cdots,n$. Note that the original QTE estimator $\hat{q}(\tau) = \hat{q}_1(\tau)-\hat{q}_0(\tau)$ is invariant to the re-ordering.

For the bootstrap sample, we directly compute $\hat{\beta}^*(\tau)$ from the sub-gradient condition of \eqref{eq:gb}. Specifically, we compute $\hat{\beta}_0^*(\tau)$ as $Y^0_{(h_0)}$ and $\hat{q}^*(\tau) \equiv \hat{\beta}_1^*(\tau)$ as $Y^1_{(h_1)} - Y^0_{(h_0)}$, where $Y^0_{(h_0)}$ and $Y^1_{(h_1)}$ are the $h_0$th and $h_1$th order statistics of outcomes in the treatment and control groups, respectively,\footnote{\protect\doublespacing We assume $Y^a_{(1)} \leq \cdots \leq Y^a_{(n)}$ for $a = 0,1$.} and $h_0$ and $h_1$ are two integers satisfying
\begin{align}
\label{eq:h}
n\tau + T_{n,a}^*(\tau) + 1 \geq h_a \geq n\tau + T_{n,a}^*(\tau),~a=0,1,
\end{align}
with
\begin{align*}
\begin{pmatrix}
T_{n,1}^*(\tau) \\
T_{n,0}^*(\tau)
\end{pmatrix} = \sqrt{n} S_n^*(\tau) = & \frac{1}{\sqrt{2}}\biggl[\begin{pmatrix}
\sum_{j=1}^{n}\eta_j\left(\tau- 1\{Y_{(j,1)} \leq \hat{q}_1(\tau)\}\right) \\
\sum_{j=1}^{n}\eta_j\left(\tau- 1\{Y_{(j,0)} \leq \hat{q}_0(\tau)\}\right)
\end{pmatrix} \\
& +  \begin{pmatrix}
\sum_{k=1}^{\lfloor n/2 \rfloor}\hat{\eta}_k \left[\left(\tau - 1\{Y_{(k,1)} \leq \hat{q}_1(\tau)\} \right) - \left(\tau - 1\{Y_{(k,3)} \leq \hat{q}_1(\tau)\} \right)\right] \\
\sum_{k=1}^{\lfloor n/2 \rfloor}\hat{\eta}_k\left[\left(\tau - 1\{Y_{(k,2)} \leq \hat{q}_0(\tau)\} \right) - \left(\tau - 1\{Y_{(k,4)} \leq \hat{q}_0(\tau)\} \right)\right] \\
\end{pmatrix} \biggr].
\end{align*}
As the probability of $n\tau + T_{n,a}^*(\tau)$ being an integer is zero, $h_a$ is uniquely defined with probability one.\footnote{The sub-gradient condition of \eqref{eq:gb} is $\hat{q}^*(\tau) = Y_{i_1}-  Y_{i_0}$ such that $i_i,i_0 \in [2n]$ are two indexes, $A_{i_1}= 1$,  $A_{i_0}= 0$,
	\begin{align*}
	& \tau n + T_{n,1}^*(\tau) \geq \sum_{i \in [2n]} A_i 1\{ Y_i < Y_{i_1}\} \geq  \tau n + T_{n,1}^*(\tau) - 1, \quad \text{and} \\
	& \tau n + T_{n,0}^*(\tau) \geq \sum_{i \in [2n]} (1-A_i) 1\{ Y_i < Y_{i_0}\} \geq  \tau n + T_{n,0}^*(\tau) - 1.
	\end{align*}
	By letting  $h_1 = \sum_{i \in [2n]} A_i 1\{ Y_i < Y_{i_1}\}+1$ and $h_0 = \sum_{i \in [2n]} (1-A_i) 1\{ Y_i < Y_{i_0}\}+1$, we have  $Y_{i_1} = Y_{(h_1)}^1$ and $Y_{i_0} = Y_{(h_0)}^0$. }

We summarize the steps in the bootstrap procedure as follows.
\begin{enumerate}
	\item Re-order the pairs.
	\item Compute the original estimator $\hat{q}(\tau) = \hat{q}_1(\tau) - \hat{q}_0(\tau)$.
	\item Let $B$ be the number of bootstrap replications. Let $\mathcal{G}$ be a grid of quantile indexes. For $b \in [B]$, generate $\{\eta_j\}_{j \in [n]}$ and $\{\hat{\eta}_k\}_{k \in \lfloor n/2 \rfloor}$. Compute $\hat{q}^{*b}(\tau) = Y^1_{(h_1)} - Y^0_{(h_0)}$ for $\tau \in \mathcal{G}$, where $h_0$ and $h_1$ are computed in \eqref{eq:h}. Obtain $\{\hat{q}^{*b}(\tau)\}_{\tau \in \mathcal{G}}$.
	\item Repeat the above step for $b \in [B]$ and obtain $B$ bootstrap estimators of the QTE, denoted as
	$\{\hat{q}^{*b}(\tau)\}_{b \in [B],\tau \in \mathcal{G}}$.
\end{enumerate}

\subsection{Computation of the IPW Multiplier Bootstrap}
\label{subsec:IPW computation}
We first provide more details on the sieve bases. Let $b(x) \equiv (b_{1}(x), \cdots, b_{K}(x))^\top$, where $\{b_{k }(\cdot)\}_{k=1}^{K}$ are $K$ basis functions of a linear sieve space $\mathcal{B}$.  Given that all $d_x$ elements of $X$ are continuously distributed, the sieve space $\mathcal{B}$ can be constructed as follows.
\begin{enumerate}
	\item For each element $X^{(l)}$ of $X$, $l=1,\cdots,d_x$, let $\mathcal{B}_l$ be the univariate sieve space of dimension $J_n$. One example of $\mathcal{B}_l$ is the linear span of the $J_n$ dimensional polynomials given by
	$$\mathcal{B}_l = \biggl\{\sum_{k=0}^{J_n}\alpha_k x^k, x \in \Supp(X^{(l)}), \alpha_k \in \Re \biggr\};$$
	Another is the linear span of $r$-order splines with $J_n$ nodes given by
	$$\mathcal{B}_l = \biggl\{\sum_{k=0}^{r-1}\alpha_k x^k + \sum_{j=1}^{J_n}b_j[\max(x-t_j,0)]^{r-1}, x \in \Supp(X^{(l)}), \alpha_k, b_j \in \Re \biggr\},$$
	where the grid $-\infty=t_0 \leq t_1 \leq \cdots \leq t_{J_n} \leq t_{J_n+1} = \infty$ partitions $\Supp(X^{(l)})$ into $J_n+1$ subsets $I_j = [t_j,t_{j+1}) \cap \Supp(X^{(l)})$, $j=1,\cdots,J_n-1$, $I_{0} = (t_{0},t_{1}) \cap \Supp(X^{(l)})$, and $I_{J_{n}} = (t_{J_n},t_{J_n+1}) \cap \Supp(X^{(l)})$.
	\item Let $\mathcal{B}$ be the tensor product of $\{\mathcal{B}_l\}_{l=1}^{d_x}$, which is defined as a linear space spanned by the functions $\prod_{l=1}^{d_x} g_l$, where $g_l \in \mathcal{B}_l$. The dimension of $\mathcal{B}$ is then $K \equiv J_n^{d_x}$.
\end{enumerate}

In practice, we suggest not using all the tensor products. Otherwise the dimension $J_n^{d_x}$ can be too large even for moderate sample sizes. Instead, when the number of pairs is $50$ or $100$, we suggest using splines with one node (usually the median) for each dimension and one interaction term across each pair of dimensions. In the appendix, we also propose a cross-validation method to select the basis functions.

Given the sieve bases, we can estimate the propensity score following \eqref{eq:ps}. We then obtain $\hat{q}_{ipw,1}^w(\tau)$ and $\hat{q}_{ipw,0}^w(\tau)$ by solving the sub-gradient conditions for the two optimizations in \eqref{eq:qipw10}. Specifically, we have $\hat{q}_{ipw,1}^w(\tau) = Y_{h_1'}$ and $\hat{q}_{ipw,0}^w(\tau) = Y_{h_0'}$, where the indexes $h_0'$ and $h_1'$ satisfy $A_{h_a'} = a$, $a=0,1$,
\begin{align}
\label{eq:h1'}
\tau \left(\sum_{i=1}^{2n}\frac{\xi_i A_i}{\hat{A}_i}\right) - \frac{\xi_{h_1'}}{\hat{A}_{h_1'}} \leq \sum_{i =1}^{2n}\frac{\xi_iA_i}{\hat{A}_i} 1\{ Y_i < Y_{h_1'}\} \leq \tau \left(\sum_{i=1}^{2n}\frac{\xi_i A_i}{\hat{A}_i}\right),
\end{align}
and
\begin{align}
\label{eq:h0'}
\tau \left(\sum_{i=1}^{2n}\frac{\xi_i (1-A_i)}{1-\hat{A}_i}\right)  - \frac{\xi_{h_0'}}{1-\hat{A}_{h_0'}} \leq \sum_{i =1}^{2n}\frac{\xi_i(1-A_i)}{1-\hat{A}_i} 1\{ Y_i < Y_{h_0'}\} \leq \tau \left(\sum_{i=1}^{2n}\frac{\xi_i (1-A_i)}{1-\hat{A}_i}\right).
\end{align}
In practical implementation we set $\{\xi_i\}_{i \in [2n]}$ as i.i.d. standard exponential random variables. In this case, all the equalities in \eqref{eq:h1'} and \eqref{eq:h0'} hold with probability zero. Thus, $h_1'$ and  $h_0'$ are uniquely defined with probability one.

The IPW multiplier bootstrap can also be used to infer the ATE when the pairs identities are unknown. The point estimator of ATE under MPD is just the difference in mean estimator: $\hat{\Delta} = \frac{1}{n}\sum_{i =1}^{2n}(A_iY_i - (1-A_i)Y_i)$. In order to mimic its limit distribution, we propose to an IPW multiplier bootstrap for the ATE estimator
\begin{align*}
\hat{\Delta}^w_{ipw} = \frac{1}{\sum_{i=1}^{2n}{\xi_i A_i}/{\hat{A}_i}}\sum_{i =1}^{2n}\frac{\xi_i A_i Y_i}{\hat{A}_i} - \frac{1}{\sum_{i=1}^{2n}{\xi_i (1-A_i)}/{(1-\hat{A}_i)}}\sum_{i =1}^{2n}\frac{\xi_i (1-A_i) Y_i}{1-\hat{A}_i}.
\end{align*}

We summarize the bootstrap procedure as follows.
\begin{enumerate}
	\item Compute the original estimator $\hat{q}(\tau) = \hat{q}_1(\tau) - \hat{q}_0(\tau)$ and $\hat{\Delta}$.
	\item Let $B$ be the number of bootstrap replications. Let $\mathcal{G}$ be a grid of quantile indexes. For $b \in [B]$, generate $\{\xi_i\}_{i \in [2n]}$ as a sequence of i.i.d. exponential random variables. Estimate the propensity score following \eqref{eq:ps}. Compute $\hat{q}_{ipw}^{w,b}(\tau) = Y_{h_1'} - Y_{h_0'}$ for $\tau \in \mathcal{G}$, where $h_0'$ and $h_1'$ are computed as in \eqref{eq:h1'} and \eqref{eq:h0'}, respectively, and $\hat{\Delta}_{ipw}^{w,b}$.
	\item Repeat the above step for $b \in [B]$ and obtain $B$ bootstrap estimators of the QTE and ATE, denoted as
	$\{\hat{q}_{ipw}^{w,b}(\tau)\}_{b \in [B],\tau \in \mathcal{G}}$, $\{\hat{\Delta}_{ipw}^{w,b}\}_{b \in [B]}$, respectively.
\end{enumerate}

For comparison, we also consider the naive multiplier bootstrap and the naive multiplier bootstrap of the pairs in our simulations. The computation of the naive multiplier bootstrap follows a procedure similar to the above with only one difference: the nonparametric estimate $\hat{A}_i$ of the propensity score is replaced by the truth, that is, $1/2$. The computation of the naive multiplier bootstrap of the pairs follows a similar procedure to that of the naive multiplier bootstrap except that the units in the same pair share the same multiplier.

\subsection{Bootstrap Confidence Intervals}
\label{subsec:boostrap+confid_int}
Given the bootstrap estimates, we discuss how to conduct bootstrap inference for the null hypotheses with single, multiple, and a continuum of quantile indexes. We take the gradient bootstrap as an example. If the IPW multiplier bootstrap is used, one can just replace $\{\hat{q}^{*b}(\tau)\}_{b \in [B],\tau \in \mathcal{G}}$ by $\{\hat{q}_{ipw}^{w,b}(\tau)\}_{b \in [B],\tau \in \mathcal{G}}$ in the following cases. The same procedure applies to the bootstrap inference of ATE as well.

\vspace{1mm}
\textbf{Case (1).} We aim to test the single null hypothesis that $\mathcal{H}_0: q(\tau) = \underline{q}$ vs. $q(\tau) \neq \underline{q}$. Let $\mathcal{G} = \{\tau\}$ in the procedures described above. Further denote $\mathcal{Q}(\nu)$ as the $\nu$th empirical quantile of the sequence  $\{\hat{q}^{*b}(\tau)\}_{b \in [B]}$. Let $\alpha \in (0,1)$ be the significance level. We suggest using the bootstrap estimator to construct the standard error of $\hat{q}(\tau)$ as  $\hat{\sigma} = \frac{\mathcal{Q}(0.975)- \mathcal{Q}(0.025)}{C_{0.975} - C_{0.025}}$, where $C_{\mu}$ is the $\mu$th standard normal critical value. Then a valid confidence interval and Wald test using this standard error are
\begin{align*}
CI(\alpha) = (\hat{q}(\tau)-C_{1-\alpha/2}\hat{\sigma}, \hat{q}(\tau)+C_{1-\alpha/2}\hat{\sigma}),
\end{align*}
and $1\{\left|\frac{\hat{q}(\tau) - \underline{q}}{\hat{\sigma}}\right| \geq C_{1-\alpha/2}\}$, respectively.\footnote{It is asymptotically valid to use standard and percentile bootstrap confidence intervals. In our simulations, we found that the confidence interval proposed in the paper has better finite sample performance.}


\vspace{2mm}
\textbf{Case (2).} We aim to test the null hypothesis that $\mathcal{H}_0: q(\tau_1) - q(\tau_2) = \underline{q}$ vs. $q(\tau_1) - q(\tau_2) \neq \underline{q}$. In this case, let $\mathcal{G} = \{\tau_1,\tau_2\}$. Further, let $\mathcal{Q}(\nu)$ denote the $\nu$th empirical quantile of the sequence  $\{\hat{q}^{*b}(\tau_1) - \hat{q}^{*b}(\tau_2)\}_{b \in [B]}$, and let $\alpha \in (0,1)$ be the significance level. We suggest using the bootstrap standard error to construct a valid confidence interval and Wald test as
\begin{align*}
CI(\alpha) = (\hat{q}(\tau_1)-\hat{q}(\tau_2)-C_{1-\alpha/2}\hat{\sigma}, \hat{q}(\tau_1)-\hat{q}(\tau_2)+C_{1-\alpha/2}\hat{\sigma}),
\end{align*}
and $1\{\left|\frac{\hat{q}(\tau_1)-\hat{q}(\tau_2) - \underline{q}}{\hat{\sigma}}\right| \geq C_{1-\alpha/2}\}$, respectively, where $\hat{\sigma} = \frac{\mathcal{Q}(0.975)- \mathcal{Q}(0.025)}{C_{0.975} - C_{0.025}}$.


%

\vspace{2mm}
\textbf{Case (3).} We aim to test the null hypothesis that
$$\mathcal{H}_0: q(\tau) = \underline{q}(\tau)~\forall \tau \in \Upsilon~\text{vs.}~q(\tau) \neq \underline{q}(\tau)~ \exists \tau \in \Upsilon.$$ In theory, we should let $\mathcal{G} = \Upsilon$. In practice, we let $\mathcal{G} = \{\tau_1,\cdots,\tau_G\}$ be a fine grid of $\Upsilon$ where $G$ should be as large as computationally possible. Further, let $\mathcal{Q}_\tau(\nu)$ denote the $\nu$th empirical quantile of the sequence  $\{\hat{q}^{*b}(\tau)\}_{b \in [B]}$ for $\tau \in \mathcal{G}$. Compute the standard error of $\hat{q}(\tau)$ as
\begin{align*}
\hat{\sigma}_\tau = \frac{\mathcal{Q}_\tau(0.975)- \mathcal{Q}_\tau(0.025)}{C_{0.975} - C_{0.025}}.
\end{align*}
The uniform confidence band with an $\alpha$ significance level is constructed as
\begin{align*}
CB(\alpha) = \{ \hat{q}(\tau)-\mathcal{C}_\alpha  \hat{\sigma}_\tau,  \hat{q}(\tau)+\mathcal{C}_\alpha  \hat{\sigma}_\tau: \tau \in \mathcal{G} \},
\end{align*}
where the critical value $\mathcal{C}_\alpha$ is computed as
\begin{align*}
\mathcal{C}_\alpha = \inf\left\{z:\frac{1}{B}\sum_{b=1}^B 1\left\{ \sup_{\tau \in \mathcal{G}}\left| \frac{\hat{q}^{*b}(\tau) - \tilde{q}(\tau)}{\hat{\sigma}_\tau}\right| \leq z \right\}\geq 1-\alpha \right\}
\end{align*}
and $\tilde{q}(\tau)$ is first-order equivalent to $\hat{q}(\tau)$ in the sense that $\sup_{\tau \in \Upsilon}|\tilde{q}(\tau) - \hat{q}(\tau)|=  o_p(1/\sqrt{n})$. We suggest choosing $\tilde{q}(\tau) = \frac1{2}\{\mathcal{Q}_\tau(0.975) + \mathcal{Q}_\tau(0.025)\}$ over other choices such as $\tilde{q}(\tau) =\mathcal{Q}_\tau(0.5)$ and $\tilde{q}(\tau) = \hat{q}(\tau)$ due to its better finite sample performance. We reject $\mathcal{H}_0$ at an $\alpha$ significance level if $\underline{q}(\cdot) \notin CB(\alpha).$

\subsection{Practical Recommendations}
\label{sec:rec}
Our practical recommendations are straightforward. If pair identities are known, we suggest using the gradient bootstrap for inference. Otherwise, we suggest using the IPW multiplier bootstrap with a nonparametrically estimated propensity score for inference.

\section{Simulations}
\label{sec:sim}

In this section, we assess the finite sample performance of the methods discussed in Section \ref{sec:bootstrap} with a Monte Carlo simulation study. In all cases, potential outcomes for $a\in\{0,1\}$ and $1\leq i\leq2n$ are generated as
\begin{equation}
\ensuremath{Y_{i}(a)=\mu_{a}+m_{a}\left(X_{i}\right)+\sigma_{a}\left(X_{i}\right)\epsilon_{a,i}},~a=0,1,\label{eq:simulpart01}
\end{equation}
where $\mu_{a},m_{a}\left(X_{i}\right),\sigma_{a}\left(X_{i}\right)$,
and $\epsilon_{a,i}$ are specified as follows. In each of the specifications below, $n\in\{50,100\}$ and ($X_{i},\epsilon_{0,i},\epsilon_{1,i}$) are i.i.d. The number of replications is 10,000. For bootstrap replications we set $B=5,000$.

\begin{description}
	\item [{Model 1}]
	$X_{i}\sim\text{Unif}[0,1]$; $m_{0}\left(X_{i}\right)=0$; $m_{1}\left(X_{i}\right)=10\left(X_{i}^{2}-\frac{1}{3}\right)$; $\epsilon_{a,i}\sim N(0,1)\text{ for }a=0,1$; $\sigma_{0}\left(X_{i}\right)=\sigma_{0}=1$ and $\sigma_{1}\left(X_{i}\right)=\sigma_{1}$.
	
	\item [{Model 2}] As in Model 1, but with $\sigma_{0}\left(X_{i}\right)=\left(1+X_{i}^{2}\right)$ and $\sigma_{1}\left(X_{i}\right)=\left(1+X_{i}^{2}\right)\sigma_{1}$.
	
	\item [{Model 3}] $X_{i}=\left(\Phi\left(V_{i1}\right),\Phi\left(V_{i2}\right)\right)^\top $,
	where $\Phi(\cdot)$ is the standard normal cumulative distribution function and
	\[
	V_{i}\sim N\left(\left(\begin{array}{l}
	0\\
	0
	\end{array}\right),\left(\begin{array}{ll}
	1 & \rho\\
	\rho & 1
	\end{array}\right)\right),
	\]
	$m_{0}\left(X_{i}\right)=\gamma^\top X_{i}-1$; $m_{1}\left(X_{i}\right)=m_{0}\left(X_{i}\right)+10\left(\Phi^{-1}\left(X_{i1}\right)\Phi^{-1}\left(X_{i2}\right)-\rho\right)$;
	$\epsilon_{a,i}\sim N(0,1)$ for $a=0,1$; $\sigma_{0}\left(X_{i}\right)=\sigma_{0}=1$
	and $\sigma_{1}\left(X_{i}\right)=\sigma_{1}$. We set $\gamma=\left(1,1\right)^\top $, $\sigma_{1}=1$, $\rho=0.2$.
	
	\item [{Model 4}] As in Model 3, but with $\gamma=\left(1,4\right)^\top $,
	$\sigma_{1}=2$, $\rho=0.7$.
	
\end{description}

Pairs are determined similarly to those in \cite{BRS19}. Specifically, if $X_i$ is a scalar, then pairs are determined by sorting $\{X_i\}_{i \in [2n]}$. If $X_i$ is multi-dimensional, then the pairs are determined by the permutation $\pi$ computed using the \emph{R} package \emph{nbpMatching}. We refer interested readers to \citet[Section 4]{BRS19} for more detail.  After forming the pairs, we assign treatment status within each pair through a random draw from the uniform distribution over $\{(0,1),(1,0)\}$.

We examine the performance of various tests for ATEs and QTEs at the nominal level $\alpha=5\%$. For the ATE, we consider the hypothesis that
\begin{align*}
\mathbb{E}(Y(1) - Y(0)) = \text{truth}+\Delta \quad vs. \quad \mathbb{E}(Y(1) - Y(0)) \neq \text{truth}+\Delta.
\end{align*}
For the QTE, we consider the hypotheses that
\begin{align*}
q(\tau) = \text{truth}+\Delta \quad vs. \quad q(\tau) \neq \text{truth}+\Delta,
\end{align*}
for $\tau = 0.25$, $0.5$, and $0.75$,
\begin{align}
\label{eq:diff}
q(0.25)- q(0.75) = \text{truth}+\Delta \quad vs. \quad q(0.25)- q(0.75) \neq \text{truth}+\Delta,
\end{align}
and
\begin{align}
\label{eq:uni}
q(\tau) = \text{truth}+\Delta \quad \forall \tau \in [0.25,0.75] \quad vs. \quad q(\tau) \neq \text{truth}+\Delta \quad \exists \tau \in [0.25,0.75].
\end{align}
To illustrate size and power of the tests, we set $\mathcal{H}_0: \Delta = 0$ and $\mathcal{H}_1: \Delta = 1/2$. The true value for the ATE is $0$, whereas the true values for the QTEs are simulated with a $10,000$ sample size and replications. The computational procedures described in Section \ref{sec:compute} are followed to perform the bootstrap and calculate the test statistics. To test the single null hypothesis involving one or two quantile indexes, we use the Wald tests specified in Section \ref{subsec:boostrap+confid_int}. To test the null hypothesis involving a continuum of quantile indexes, we use the uniform confidence band $CB(\alpha)$ defined in Case (3) in the same section.

\newcolumntype{L}{>{\raggedright\arraybackslash}X}
\newcolumntype{C}{>{\centering\arraybackslash}X}

\begin{table}[ht]
	\caption{The Empirical Size and Power of Tests for ATEs}
	\vspace{1ex}
	\centering{}%
	\begin{tabularx}{1\textwidth} {L CCCC CCCC}
		\hline
		\hline
		\multirow{3}{*}{Model} & \multicolumn{8}{c}{$\mathcal{H}_{0}$: $\Delta=0$} \\
		& \multicolumn{4}{c}{$n=50$} & \multicolumn{4}{c}{$n=100$} \\
		\cmidrule(lr){2-5} \cmidrule(lr){6-9}
		& Naive & Naive Pair & Adj  & IPW  & Naive & Naive Pair & Adj & IPW  \\
		
		\hline
		1 & 1.32 & 1.52 & 5.47  & 5.44 & 1.22 & 1.34 & 5.75 & 6.00 \\
		2 & 1.85 & 2.22 & 5.35 &  5.59 & 1.64  & 1.83 & 5.63  & 5.89 \\
		3 & 1.20 & 1.48 & 4.76 & 4.92 & 0.77 & 0.89 & 4.68 & 5.16 \\
		4 & 2.32 & 2.72 & 6.47 & 6.01 & 1.25 & 1.33 & 5.33 & 4.74 \\
		\hline
		\multirow{3}{*}{Model} & \multicolumn{8}{c}{$\mathcal{H}_{1}$: $\Delta=1/2$}\\
		& \multicolumn{4}{c}{$n=50$} & \multicolumn{4}{c}{$n=100$} \\
		\cmidrule(lr){2-5} \cmidrule(lr){6-9}
		& Naive & Naive Pair & Adj & IPW & Naive & Naive Pair & Adj & IPW \\
		\hline
		1 & 11.80 & 13.47 & 29.10 & 29.44 & 27.67 & 28.57 & 49.79 & 50.46 \\
		2 & 10.43 & 12.06 & 23.26 & 24.24 & 23.72 & 24.99 & 40.42  & 41.68 \\
		3 & 1.31 & 1.73 & 5.66 & 5.91 & 1.92 & 2.07 & 8.13 & 8.74 \\
		4 & 1.08 & 1.45 & 5.16 & 4.35 & 0.93 & 1.05 & 5.65 & 4.89 \\
		\hline
		
	\end{tabularx}
	\justify
	Notes: The table presents the rejection probabilities for tests of ATEs.
	The columns `Naive' and `Adj' correspond to the two-sample $t$-test and the adjusted $t$-test in \cite{BRS19}, respectively; the `Naive pair' column corresponds to the $t$-test using the standard errors estimated by the naive multiplier bootstrap of the pairs; the column `IPW' corresponds to the $t$-test using the standard errors estimated by the IPW multiplier bootstrap ATE estimator.
	\label{tab:sim_ate}
\end{table}

The results for the ATEs appear in Table \ref{tab:sim_ate}. Each row presents a different model and each column reports the rejection probabilities for the various methods. The column `Naive' refers to the two-sample $t$-test and `Adj' refers to the adjusted $t$-test in \cite{BRS19}; the column `Naive pair' corresponds to the $t$-test using the standard errors estimated by the naive multiplier bootstrap of the pairs; the column `IPW' corresponds to the $t$-test using the standard errors estimated by the IPW multiplier bootstrap ATE estimator.

We make several observations on these findings. First, the two-sample $t$-test has rejection probability under $\mathcal{H}_{0}$ far below the nominal level and is the least powerful test among the four. Second, the adjusted $t$-test has rejection probability under $\mathcal{H}_{0}$ close to the nominal level and is not conservative. This result is consistent with those in \cite{BRS19}. Third, the $t$-test using the standard error estimated by bootstrapping the pairs of units alone is conservative under the null and lacks power under the alternative. Fourth, the IPW $t$-test proposed in this paper has performance similar to the adjusted $t$-test.\footnote{\protect\doublespacing Throughout this section, for both ATE and QTE estimation, we use splines to nonparametrically estimate the propensity score in the IPW multiplier bootstrap. If dim($X_i$)=1, we choose the bases $\{1,X,X^2,[\max(X-qx_{0.5},0)]^2 \}$ where $qx_{0.5}$ is the quantile of $X$ at 50\%; if dim($X_i$)=2, we choose the bases $\{1,X_{1},X_{2},\max(X_{1}-qx_{1,0.5},0), \max(X_{2}-qx_{2,0.5},0), X_{1}X_{2}\}$, where for $j=1,2$, $qx_{j,\alpha}$ is the $\alpha$th sample percentile of $X_j$. Results are similar using the cross-validation method proposed in the appendix to choose the sieve bases: for details on the cross-validation method and its simulation results, see Section H  in the supplement.} Under $\mathcal{H}_{0}$, the test has rejection probability close to 5\%; under $\mathcal{H}_{1}$, it is more powerful than the `naive' and `naive pair' methods and has power similar to the adjusted $t$-test. These findings indicate that the IPW $t$-test provides an alternative to the adjusted $t$-test when pair identities are unknown.

The results for QTEs are summarized in Tables \ref{tab:sim_qte} and \ref{tab:sim_uniform}. Each table has four panels (Models 1-4). Each row in the panel displays the rejection probabilities for the tests using the standard errors estimated by various bootstrap methods. Specifically, the rows `Naive', `Naive pair', `Gradient', and `IPW' respectively correspond to the results of the naive multiplier bootstrap, the naive multiplier bootstrap of the pairs, the gradient bootstrap, and the IPW multiplier bootstrap.

Table \ref{tab:sim_qte} reports the empirical size and power of the tests with a single null hypothesis involving one or two quantile indexes. Columns `0.25', `0.50', and `0.75' correspond to tests with quantiles at 25\%, 50\%, and 75\%. Column `Dif' corresponds to the test with null hypothesis \eqref{eq:diff}. As expected given Theorem \ref{thm:weight}, the test with standard errors estimated by two naive methods performs poorly in all cases. It is conservative under $\mathcal{H}_{0}$ and lacks power under $\mathcal{H}_{1}$. In contrast, the test using the standard errors estimated by either the gradient bootstrap or the IPW multiplier bootstrap method has a rejection probability under $\mathcal{H}_{0}$ that is close to the nominal level in almost all specifications. When the number of pairs is 50, the tests in the `Dif' column constructed based on either the gradient or the IPW multiplier bootstrap method are slightly conservative. Sizes approach the nominal level when $n$ increases to $100$.

Table \ref{tab:sim_uniform} reports empirical size and power of the uniform confidence bands for the hypothesis specified in \eqref{eq:uni} with a grid $\mathcal{G} = \{0.25,0.27,\cdots,0.47,0.49,0.5,0.51,0.53,\cdots,0.73,0.75\}$. The test using standard errors estimated by two naive methods has rejection probabilities under $\mathcal{H}_{0}$ far below the nominal level in all specifications. In Models 1-2, the test using standard errors estimated by either the gradient bootstrap or the IPW multiplier bootstrap yields a rejection probability under $\mathcal{H}_{0}$ that is very close to the nominal level even when the number of pairs is as small as 50. Nonetheless, in Models 3-4, the tests constructed based on both methods are conservative when the number of pairs equals 50. When the number of pairs increases to 100, both tests perform much better and have rejection probabilities under $\mathcal{H}_{0}$ that are close to the nominal level. Under $\mathcal{H}_{1}$, the tests based on both the gradient and IPW methods are more powerful than those based on the naive methods.

In summary, the simulation results in Tables \ref{tab:sim_qte} and \ref{tab:sim_uniform} are consistent with the results in Theorems \ref{thm:boot} and \ref{thm:ipw_w}: both the gradient bootstrap and the IPW multiplier bootstrap provide valid pointwise and uniform inference for QTEs under MPDs. The findings also show that when the information on pair identities is unavailable, the IPW multiplier bootstrap continues to provide a sound basis for inference.

\begin{landscape}
	\begin{table}[htbp]
		\caption{The Empirical Size and Power of Tests for QTEs }
		\vspace{1ex}

		\centering{}
		\begin{threeparttable}
			\begin{tabular}{lcccccccccccccccc}
				\hline
				\hline
				\multicolumn{1}{l}{} & \multicolumn{8}{c}{$\mathcal{H}_{0}$: $\Delta=0$} & \multicolumn{8}{c}{$\mathcal{H}_{1}$: $\Delta=1/2$}\\
				\multicolumn{1}{l}{} & \multicolumn{4}{c}{$n=50$} & \multicolumn{4}{c}{$n=100$} & \multicolumn{4}{c}{$n=50$} & \multicolumn{4}{c}{$n=100$}\\
				\cmidrule(lr){2-5} \cmidrule(lr){6-9} \cmidrule(lr){10-13}  \cmidrule(lr){14-17}
				\multicolumn{1}{l}{} & 0.25 & 0.50 & 0.75 & \multicolumn{1}{c}{Dif} & 0.25 & 0.50 & 0.75 & \multicolumn{1}{c}{Dif} & 0.25 & 0.50 & 0.75 & \multicolumn{1}{c}{Dif} & 0.25 & 0.50 & 0.75 & Dif\\
				\hline
				\emph{Model 1} &  &  &  &  &  &  &  &  &  &  &  &  &  &  &  & \\
				Naive & 3.00 & 2.00 & 2.22 & 1.98 & 3.12 & 2.06 & 1.93 & 1.73 & 16.67 & 6.05 & 5.56 & 3.96 & 34.93 & 11.56 & 8.11 & 7.35 \\
				Naive pair & 2.99 & 2.29 & 2.30 & 2.07 & 3.37 & 2.23 & 2.08 & 1.81 & 16.85 & 6.75 & 5.64 & 4.08 & 35.33 & 11.41 & 8.24 & 7.59 \\
				Gradient & 5.13 & 4.82 & 4.92 & 3.66 & 5.07 & 5.62 & 5.30 & 4.04 & 23.76 & 13.03 & 11.27 & 8.18 & 42.92 & 22.91 & 17.30 & 14.57 \\
				IPW & 5.47 & 5.31 & 6.17 & 4.24 & 5.26 & 5.83 & 5.65 & 3.95 & 24.81 & 13.48 & 12.12 & 8.40 & 43.93 & 23.33 & 17.21 & 13.91 \\
				&  &  &  &  &  &  &  &  &  &  &  &  &  &  &  & \\
				\emph{Model 2} &  &  &  &  &  &  &  &  &  &  &  &  &  &  &  & \\
				Naive & 3.08 & 2.32 & 2.55 & 1.96 & 3.64 & 2.53 & 2.08 & 1.87 & 14.82 & 6.54 & 4.71 & 3.68 & 30.29 & 11.50 & 7.46 & 6.88 \\
				Naive pair & 3.05 & 2.50 & 2.65 & 2.03 & 3.87 & 2.77 & 2.23 & 1.95 & 14.67 & 6.81 & 4.96 & 3.65 & 30.97 & 11.80 & 7.85 & 7.27 \\
				Gradient & 4.57 & 4.63 & 4.39 & 3.44 & 5.00 & 5.42 & 5.28 & 3.68 & 19.51 & 12.25 & 8.76 & 6.57 & 35.38 & 20.86 & 14.79 & 12.25 \\
				IPW & 4.93 & 5.12 & 5.78 & 4.45 & 5.17 & 5.73 & 5.88 & 4.00 & 20.29 & 12.90 & 10.40 & 7.35 & 36.38 & 21.53 & 15.14 & 12.53 \\
				&  &  &  &  &  &  &  &  &  &  &  &  &  &  &  & \\
				\emph{Model 3} &  &  &  &  &  &  &  &  &  &  &  &  &  &  &  & \\
				Naive & 2.11 & 1.03 & 2.10 & 0.92 & 1.56 & 1.37 & 1.58 & 0.86 & 4.98 & 2.85 & 1.92 & 0.98 & 6.57 & 7.14 & 1.73 & 1.43 \\
				Naive pair & 2.21 & 1.18 & 2.31 & 1.13 & 1.57 & 1.42 & 1.55 & 0.96 & 4.99 & 3.29 & 2.14 & 1.11 & 6.80 & 7.56 & 1.88 & 1.44 \\
				Gradient & 5.24 & 3.06 & 3.14 & 1.76 & 4.83 & 4.20 & 4.27 & 3.01 & 9.71 & 7.43 & 3.22 & 2.39 & 13.80 & 16.72 & 5.67 & 4.40 \\
				IPW & 4.76 & 3.19 & 5.61 & 2.60 & 4.77 & 3.71 & 4.95 & 3.02 & 8.75 & 7.81 & 5.35 & 3.09 & 13.04 & 15.42 & 6.06 & 4.21 \\
				&  &  &  &  &  &  &  &  &  &  &  &  &  &  &  & \\
				\emph{Model 4} &  &  &  &  &  &  &  &  &  &  &  &  &  &  &  & \\
				Naive & 2.59 & 1.71 & 1.98 & 1.65 & 2.65 & 1.66 & 1.55 & 1.23 & 6.09 & 1.94 & 1.76 & 1.28 & 9.85 & 2.98 & 1.19 & 1.18 \\
				Naive pair & 2.90 & 1.79 & 2.04 & 1.74 & 2.74 & 1.71 & 1.71 & 1.35 & 6.24 & 2.23 & 1.93 & 1.45 & 10.25 & 3.17 & 1.35 & 1.23 \\
				Gradient & 4.75 & 4.00 & 3.33 & 2.82 & 4.70 & 4.74 & 5.06 & 3.88 & 9.37 & 5.76 & 3.35 & 2.87 & 14.67 & 8.88 & 5.27 & 4.25 \\
				IPW & 3.97 & 3.97 & 4.91 & 3.68 & 4.23 & 4.51 & 5.01 & 3.48 & 8.08 & 5.37 & 4.79 & 3.26 & 13.50 & 8.33 & 5.17 & 3.51 \\
				\hline
			\end{tabular}
			
			\begin{tablenotes}[para,flushleft]
				Note:	The table presents the rejection probabilities for tests of QTEs. The columns `0.25', `0.50', and `0.75' correspond to tests with quantiles at 25\%, 50\%, and 75\%, respectively; the column `Dif' corresponds to the test with the null hypothesis specified in \eqref{eq:diff}. The rows `Naive', `Naive pair', `Gradient', and `IPW' correspond to the results of the naive multiplier bootstrap, the naive multiplier bootstrap of the pairs, the gradient bootstrap, and IPW multiplier bootstrap, respectively.
			\end{tablenotes}
			
			\label{tab:sim_qte}
		\end{threeparttable}
		
	\end{table}
	
\end{landscape}

\begin{table}[ht]
	\caption{The Empirical Size and Power of Uniform Inferences for QTEs}
	\vspace{1ex}
	
	\centering{}%
	\begin{tabularx}{1\textwidth}{LCCCC}
		\hline
		\hline
		& \multicolumn{2}{c}{$\mathcal{H}_{0}$: $\Delta=0$} & \multicolumn{2}{c}{$\mathcal{H}_{1}$: $\Delta=1/2$}\\
		\cmidrule(lr){2-3} \cmidrule(lr){4-5}
		& \multicolumn{1}{c}{$n=50$} & \multicolumn{1}{c}{$n=100$} & \multicolumn{1}{c}{$n=50$} & \multicolumn{1}{c}{$n=100$}\\
		\hline
		\emph{Model 1} &  &  &  & \\
		Naive & 1.07 & 1.52 & 7.50 & 18.12 \\
		Naive pair & 1.32 & 1.63 & 7.10 & 18.52 \\
		Gradient & 4.08 & 4.64 & 17.88 & 33.30 \\
		IPW & 4.49 & 4.94 & 16.30 & 32.40 \\
		&  &  &  & \\
		\emph{Model 2} &  &  &  & \\
		Naive & 1.37 & 1.85 & 6.73 & 16.50 \\
		Naive pair & 1.39 & 1.91 & 6.63 & 17.04 \\
		Gradient & 3.66 & 4.57 & 14.30 & 27.64 \\
		IPW & 4.25 & 4.91 & 14.27 & 27.47 \\
		&  &  &  & \\
		\emph{Model 3} &  &  &  & \\
		Naive & 0.63 & 0.63 & 1.43 & 3.50 \\
		Naive pair & 0.60 & 0.69 & 1.54 & 4.02 \\
		Gradient & 1.90 & 3.07 & 5.19 & 13.33 \\
		IPW & 2.19 & 2.99 & 4.25 & 11.34 \\
		&  &  &  & \\
		\emph{Model 4} &  &  &  & \\
		Naive & 0.99 & 1.00 & 1.40 & 3.05 \\
		Naive pair & 0.97 & 1.00 & 1.33 & 3.28 \\
		Gradient & 2.87 & 3.72 & 4.47 & 8.57 \\
		IPW & 2.78 & 3.36 & 3.18 & 6.98 \\
		\hline
	\end{tabularx}
	
	\vspace{-1ex}
	\justify
	Notes: The table presents the rejection probabilities of the uniform confidence bands for the hypothesis specified in \eqref{eq:uni}. The rows `Naive', `Naive pair', `Gradient', and `IPW' correspond to the results of the naive multiplier bootstrap, the naive multiplier bootstrap of the pairs, the gradient bootstrap, and the IPW multiplier bootstrap, respectively.
	\label{tab:sim_uniform}
\end{table}

\section{Empirical Application}
\label{sec:app}

Policy and macroeconomic uncertainty are considered to be two major constraints to firm growth in developing countries (\citep{bloom2014}; \citep{bloom2018}; \citep{worldbank2004}). \cite{groh2016} conducted a randomized experiment with a MPD to explore the treatment effect of providing insurance against macroeconomic and policy shocks to microenterprise owners. In this section, we apply the bootstrap methods developed in this paper to their data and examine both the ATEs and QTEs of macroinsurance on business owners' monthly consumption and their firms' monthly profits.\footnote{\protect\doublespacing Data are available at https://microdata.worldbank.org/index.php/catalog/2063.}

\newcolumntype{B}{>{\hsize=1.3\hsize \raggedright\arraybackslash}X}
\newcolumntype{S}{>{\hsize=.9\hsize \centering\arraybackslash}X}

\begin{table}[ht]
	\centering
	\caption{Summary Statistics}
	\vspace{1ex}
	\begin{tabularx}{1\textwidth}{BSSS}\hline \hline
		& Total & Treatment group  & Control group \\
		\hline
		\textit{Outcome variables} & & & \\
		Consumption & 1946.9(903.0) & 1946.1(900.7) & 1947.7(905.7) \\
		Profit & 1342.4(1470.7) & 1299.9(1444.3) & 1384.9(1496.1) \\
		& & & \\
		\textit{Matching variables} & & & \\
		Owner is female & 0.36(0.48) & 0.36(0.48) & 0.36(0.48) \\
		Expected likelihood of a macro shock & 56.5(32.7) & 56.1(33.0) & 56.8(32.4) \\
		Higher risk aversion & 0.48(0.50) & 0.47(0.50) & 0.49(0.50) \\
		Owner is ambiguity neutral & 0.29(0.46) & 0.29(0.45) & 0.30(0.46) \\
		Sales drop 20\% more & 0.41(0.49) & 0.41(0.49) & 0.41(0.49) \\
		Sales drop between 5\% and 20\% & 0.29(0.45) & 0.29(0.45) & 0.29(0.45) \\
		Considering delaying investments & 0.10(0.30) & 0.10(0.30) & 0.10(0.29) \\
		Expect to renew their loan & 0.89(0.31) & 0.90(0.31) & 0.89(0.31) \\
		Expect to renew a loan of 3000 LE or less & 0.27(0.45) & 0.28(0.45) & 0.27(0.44) \\
		Expect to renew a loan of 3001 to 5000 LE & 0.27(0.44) & 0.27(0.44) & 0.26(0.44) \\
		Profits in Feb 2012 & 1085.4(1174.5) & 1060.4(1149.0) & 1110.4(1199.4) \\
		Profits in Jan 2012 & 1054.0(1161.1) & 1023.4(1114.5) & 1084.6(1205.5) \\
		Missing Feb 2012 profits & 0.04(0.18) & 0.04(0.19) & 0.03(0.18) \\
		Missing Jan 2012 profits & 0.04(0.19) & 0.04(0.20) & 0.03(0.18) \\
		& & & \\
		Observations & 2824 & 1412 & 1412 \\
		\hline
	\end{tabularx} \\
	
	\vspace{-1ex}
	\justify
	Notes: Unit of observation: business owners. The table presents the means and standard deviations (in parentheses) of two outcome variables and all the pair-matching variables.
	\label{tab:sum}
\end{table}

The sample consists of 2824 business owners, who were the clients of Egypt's largest microfinance institution -- Alexandria Business Association (ABA). After an exact match on gender and microfinance branch code within ABA, the business owners were grouped into pairs by using an optimal greedy algorithm to minimize the Mahalanobis distance between the values of additional 13 matching variables (See \citep{groh2016} for the definitions of these 13 variables). This segmentation gives 1412 pairs in the sample; one business owner in each pair was randomly assigned to the treatment group and the other to the control group. In the treatment group, a macroinsurance product was offered. Groh and McKenzie (2016) then examined the impacts of the access to macroinsurance on various outcome variables.

Here we focus on the impacts of macroinsurance on two outcome variables: the business owners' monthly consumption and their firms' monthly profits. Table \ref{tab:sum} gives descriptive statistics (means and standard deviations) of these two outcome variables as well as all the matching variables used by \cite{groh2016} to form the pairs in their experiments.\footnote{\protect\doublespacing We filter out 137 unbalanced observations (less than 5\% of the total observations in \citep{groh2016}) to keep a balanced data in both the pairs and the pairs of the pairs. The summary statistics in Table \ref{tab:sum} are almost exactly the same as those in Table 1 of \cite{groh2016}.}

\begin{table}[ht]
	\centering
	\caption{ATEs of Macroinsurance on Consumption and Profits}
	\vspace{1ex}
	\begin{tabularx}{1\textwidth}{LCCCC}\hline \hline
		& Naive & Naive pair & Adj & IPW \\
		\hline
		Consumption & -1.59(33.98) & -1.59(29.71) & -1.59(29.45) & -1.59(29.45) \\
		Profit & -86.68(56.09) & -86.68(48.41) & -86.68(49.38) & -86.68(45.38) \\
		\hline
	\end{tabularx} \\
	
	\vspace{-1ex}
	\justify
	Notes: The table presents the ATE estimates of the effect of macroinsurance on the monthly consumption and profits. Standard errors are in the parentheses. The columns ``Naive" and ``Adj" correspond to the two-sample $t$-test and the adjusted $t$-test in \cite{BRS19}, respectively. The column `Naive pair' corresponds to the $t$-test using standard errors estimated by the naive multiplier bootstrap of the pairs. The column ``IPW" corresponds to the $t$-test using the standard errors estimated by the IPW multiplier bootstrap.
	\label{tab:emp_ate}
\end{table}

Table \ref{tab:emp_ate} reports the ATE estimators of macroinsurance on the consumption and profits with the standard errors (in parentheses) calculated by four methods. Specifically, the columns `Naive' and `Adj' correspond to the two-sample $t$-test and the adjusted $t$-test in \cite{BRS19}, respectively; the column `Naive pair' corresponds to the $t$-test using standard errors estimated by the naive multiplier bootstrap of the pairs; the column `IPW' corresponds to the $t$-test using standard errors estimated by IPW multiplier bootstrap.\footnote{\protect\doublespacing Throughout this section, to nonparametrically estimate the propensity score in the IPW multiplier bootstrap, we first standardize all the continuous matching variables to have mean zero and variance one. There are only three continuous matching variables; the rest of the matching variables are all dummy variables. We then conduct sieve estimation by choosing the bases $\{1,\max(X_{1}-qx_{1,0.3},0),\max(X_{1}-qx_{1,0.5},0), \max(X_{2}-qx_{2,0.3},0),\max(X_{2}-qx_{2,0.5},0), \max(X_{3}-qx_{3,0.3},0),\max(X_{3}-qx_{3,0.5},0), X_{1}X_{2}, X_{2}X_{3}, X_{1}X_{3}, DV \}$, where $(X_1,X_2,X_3)$ denote  three standardized continuous matching variables, $qx_{j,0.3}$ and $qx_{j,0.5}$ are $0.3$ and $0.5$th quantiles of $X_j$ for $j = 1,2,3$, and DV denotes the dummy matching variables except for the variable ``missing Feb 2012 profits" as it is collinear with the variable ``missing Jan 2012 profits." Results reported in this section are similar to those when the sieve basis functions are selected via cross-validation. For more details on the cross-validation method and the related empirical results, see Section H in the supplement.} The results lead to the following observations. First, consistent with the findings in \cite{groh2016}, the naive two-sample $t$-tests show that expanding access to macroinsurance has no significant average effects on monthly consumption and profits. Second, the standard errors in the adjusted $t$-test are lower than those in the naive $t$-test, which is consistent with the finding in \cite{BRS19}. Compared to the standard errors estimated by the naive multiplier bootstrap of the pairs, they are modestly lower for the ATE estimates of macroinsurance on the consumption and slightly larger for those on the profits. More importantly, the standard errors estimated by the IPW multiplier bootstrap are lower than those estimated by two naive methods. These results corroborate our earlier finding that the IPW multiplier bootstrap is an alternative to the approach adopted in \cite{BRS19}, especially when the information on pair identities is unavailable. In addition, the results for the firms' profits also highlight the importance of accounting for the dependence structure within the pairs when estimating the standard errors of the ATE estimates. The ATE of macroinsurance on profits is statistically insignificant based on the naive $t$-test, but is significant at $10\%$ significance level if adjusted or IPW $t$-test is used instead.

\newcolumntype{B}{>{\hsize=1.44\hsize \raggedright\arraybackslash}X}
\newcolumntype{S}{>{\hsize=.89\hsize \centering\arraybackslash}X}

\begin{table}[ht]
	\centering
	\caption{QTEs of Macroinsurance on Consumption and Profits}
	\vspace{1ex}
	\begin{tabularx}{1\textwidth}{BSSSS}\hline \hline
		& Naive & Naive pair & Gradient & IPW  \\
		\hline
		\textit{Panel A. Consumption} & & & & \\
		25\% & -14.33(27.40) & -14.33(25.85) & -14.33(25.32) & -14.33(25.89) \\
		50\% & -4.50(34.12) & -4.50(32.23) & -4.50(32.50) & -4.50(31.80) \\
		75\% & -22.17(61.18) & -22.17(55.78) & -22.17(55.51) & -22.17(56.08) \\
		& & & & \\
		\textit{Panel B. Profit } & & & & \\
		25\% & -33.33(29.76) & -33.33(29.76) & -33.33(29.12) & -33.33(25.94) \\
		50\% & -66.67(59.52) & -66.67(54.42) & -66.67(53.15) & -66.67(51.02) \\
		75\% & -200.00(99.92) & -200.00(89.29) & -200.00(93.54) & -200.00(85.04) \\
		\hline
	\end{tabularx} \\
	\vspace{-1ex}
	\justify
	Notes: The table presents the QTE estimates of the effect of macroinsurance on the monthly consumption and profits at quantiles 25\%, 50\%, and 75\%. Standard errors are in parentheses. The columns ``Naive,"``Naive pair," ``Gradient," and ``IPW" correspond to the results of the naive multiplier bootstrap, the naive multiplier bootstrap of the pairs, the gradient bootstrap, and the IPW multiplier bootstrap, respectively.
	\label{tab:emp_qte}
\end{table}

Table \ref{tab:emp_qte} presents the QTE estimates at quantile indexes 0.25, 0.5, and 0.75 with the standard errors (in parentheses) estimated by four different methods. Specifically, the columns ``Naive,"``Naive pair," ``Gradient," and ``IPW" correspond to the results of the naive multiplier bootstrap, the naive multiplier bootstrap of the pairs, the gradient bootstrap,\footnote{\protect\doublespacing Using the original pair identities and all the matching variables in \cite{groh2016}, we can re-order the pairs according to the procedure described in Section \ref{subsec:gradient computation}. We thank David McKenzie for providing us the Stata code to implement the matching algorithm.} and the IPW multiplier bootstrap, respectively. These results lead to the following three observations.

First, consistent with the theoretical results in Section \ref{sec:bootstrap}, all the standard errors estimated by the gradient bootstrap or the IPW multiplier bootstrap are lower than those estimated by the naive multiplier bootstrap. For example in Panel A, at the 75th percentile, compared with the naive multiplier bootstrap, the gradient and IPW multiplier bootstraps reduce the standard error by 9.3\% and 8.2\%, respectively.

Second, the standard errors estimated by the gradient or IPW multiplier bootstrap are mostly lower than those estimated by the naive multiplier bootstrap of the pairs as well. The magnitude of reduction is modest, which may be because the two terms, $\frac{m_{1,\tau}(X_i)}{f_1(q_1(\tau))}$ and $\frac{m_{0,\tau}(X_i)}{f_0(q_0(\tau))}$ in \eqref{eq:pair_var}, are almost equal in the current dataset,  and thus, their effects on the standard error estimates are canceled.

Third, there is considerable heterogeneity in the effects of macroinsurance on the firm profits. Specifically, the magnitude of the treatment effects of macroinsurance rises as the quantile indexes increase. For example, in Panel B, the treatment effects on the monthly profits at the 25th percentile and the median are negative but not statistically significantly different from zero, whereas the effects at the 75th percentile are statistically significant. The magnitude of the treatment effect increases by over 100\% from the 25th percentile to the median and by about 200\% from the median to the 75th percentile. These findings may imply that expanding access to macroinsurance has small but negative effects on the firm profits in the lower tail of the distribution, and that these negative effects become stronger for upper-ranked microenterprises.

The third observation in Table \ref{tab:emp_qte} indicates that the heterogeneous effects of macroinsurance on the firms' profits are economically substantial. To assess whether these are statistically significant too, Table \ref{tab:emp_qte_dif} reports statistical tests for the heterogeneity of the QTEs. Specifically, we test the null hypotheses that $q(0.50)- q(0.25) = 0$, $q(0.75)- q(0.50) = 0$, and $q(0.75)- q(0.25) = 0$. We find that only the difference between the 75th and 25th QTEs in Panel B is statistically significant at the 10\% significance level. This finding implies that the statistical evidence of heterogeneous treatment effects of macroinsurance on the firm profits is strong only in comparisons of the microenterprises between the lower and upper tails of the distribution.

\begin{table}[ht]
	\centering
	\caption{Tests for the Difference between Two QTEs of Macroinsurance}
	\vspace{1ex}
	\begin{tabularx}{1\textwidth}{BSSSS}\hline \hline
		& Naive & Naive pair & Gradient & IPW  \\
		\hline
		\textit{Panel A. Consumption} & & & &\\
		50\%-25\% & 9.83(30.34) & 9.83(28.66) & 9.83(29.00) & 9.83(29.87) \\
		75\%-50\% & -17.67(51.51) & -17.67(48.30) & -17.67(49.09) & -17.67(48.58) \\
		75\%-25\% & -7.83(58.42) & -7.83(54.17) & -7.83(55.89) & -7.83(55.87) \\
		& & & & \\
		\textit{Panel B. Profit } & & & &\\
		50\%-25\% & -33.33(51.02) & -33.33(51.02) & -33.33(51.02) & -33.33(49.32) \\
		75\%-50\% & -133.33(85.04) & -133.33(81.63) & -133.33(82.91) & -133.33(85.04) \\
		75\%-25\% & -166.67(88.65) & -166.67(85.04) & -166.67(87.16) & -166.67(89.29) \\
		\hline
	\end{tabularx} \\
	\vspace{-1ex}
	\justify
	Notes: The table presents tests for the difference between two QTEs of macroinsurance on the monthly consumption and profits. Standard errors are in parentheses. The columns ``Naive,"``Naive pair," ``Gradient," and ``IPW" correspond to the results of the naive multiplier bootstrap, the naive multiplier bootstrap of the pairs, the gradient bootstrap, and the  IPW multiplier bootstrap, respectively.
	\label{tab:emp_qte_dif}
\end{table}

\section{Conclusion}
\label{sec:concl}

This paper has studied estimation and inference of QTEs under MPDs and developed new bootstrap methods to improve statistical performance. Derivation of the limit distribution of QTE estimators under MPDs reveals that analytic methods of inference based on asymptotic theory requires estimation of two infinite-dimensional nuisance parameters for every quantile index of interest. A further limitation is that both the naive multiplier bootstrap and the naive multiplier bootstrap of the pairs fail to approximate the limit distribution of the QTE estimator as they do not preserve the dependence structure in the original sample. Instead, we propose a gradient bootstrap approach that can consistently approximate the limit distribution of the original estimator and is free of tuning parameters. Implementation of the gradient bootstrap requires knowledge of pair identities. So when such information is unavailable we propose an IPW multiplier bootstrap and show that it consistently approximates the limit distribution of the original QTE estimator. Simulations provide finite sample evidence of these procedures that support the asymptotic findings. An empirical application of these bootstrap methods to the real dataset in \cite{groh2016} shows considerable evidence of heterogeneity in the effects of macroinsurance on firm profits. In both the simulations and the empirical application, the two recommended bootstrap methods of inference perform well in the sense that they usually provide smaller standard errors and greater inferential accuracy than those obtained by naive bootstrap methods.

Two directions for future research are especially evident from the present findings. First, it would be interesting to study inference of the QTEs when data are independent but not identically distributed. Such an assumption is adopted in the linear quantile regression literature to address issues regarding data heteroskedasticity and clustering (see, for example, \citep{chen2015} and \citep{h17}). Second, it would be useful to incorporate data-driven methods such as cross-validation to the selection of the sieve basis functions when implementing the IPW multiplier bootstrap and then develop a procedure for inference orthogonal to the model selection bias introduced by data-driven methods.

\newpage
\appendix

\section{An Illustrative Example to Show the Failure of Bootstrapping Pairs}
\label{sec:example}
In this section, we provide a simple example to illustrate the covariates across pairs in MPDs are not independent. By assigning independent multipliers to pairs, bootstrapping pairs fails to mimic such dependence of covariates across pairs.


Suppose $Y_i = U_i+X_i$, where $\{U_i,X_i\}_{i \in [2n]}$ is an i.i.d. sequence of random variables such that $U_i$ represents the unobserved heterogeneity and $X_i$ represents the covariate. We further assume, for simplicity, that $\mathbb{E}U_i=\mathbb{E}X_i = 0$ and $U_i \indep X_i.$ $\{Y_{\pi(2j-1)},Y_{\pi(2j)}  \}_{j \in [n]}$ are pairs determined by matching the covariate.

We consider the sample mean estimator $\bar{Y}_n = \frac{1}{2n}\sum_{i \in [2n]}Y_i$. We have
\begin{align*}
\sqrt{n} \bar{Y}_n = \frac{1}{2\sqrt{n}}\sum_{i \in [2n]}Y_i \convD \N(0, \frac{Var(U)}{2}+\frac{Var(X)}{2}).
\end{align*}
\noindent
Note that $\sqrt{n}\bar{Y}_n$ can be written as

\begin{align*}
\sqrt{n}\bar{Y}_n = & \frac{1}{\sqrt{n}}\sum_{j \in [n]}\frac{Y_{\pi(2j-1)}+Y_{\pi(2j)}}{2}\\
= & \frac{1}{\sqrt{n}}\sum_{j \in [n]}\frac{U_{\pi(2j-1)}+U_{\pi(2j)}}{2} +
\frac{1}{\sqrt{n}}\sum_{j \in [n]}\frac{X_{\pi(2j-1)}+X_{\pi(2j)}}{2}.
\end{align*}
\noindent
Given $\{X_i\}_{i \in [2n]}$, $\{U_{\pi(2j-1)}+U_{\pi(2j)}\}_{j \in [n]}$ is an independent sequence, thus

\begin{align*}
\frac{1}{\sqrt{n}}\sum_{j \in [n]}\frac{U_{\pi(2j-1)}+U_{\pi(2j)}}{2}
\convD \N\left(0,\frac{Var(U)}{2}\right).
\end{align*}
\noindent
Because $U \indep X$, we have $\frac{1}{\sqrt{n}}\sum_{j \in [n]}\frac{U_{\pi(2j-1)}+U_{\pi(2j)}}{2} = \frac{1}{2\sqrt{n}}\sum_{i \in [2n]}U_i$ and $\frac{1}{\sqrt{n}}\sum_{j \in [n]}\frac{X_{\pi(2j-1)}+X_{\pi(2j)}}{2}= \frac{1}{2\sqrt{n}}\sum_{i \in [2n]}X_i$ are independent.

Suppose $\{X_{\pi(2j-1)}+X_{\pi(2j)}\}_{j \in [n]}$ are independent across $j$,  $\frac{1}{\sqrt{n}}\sum_{j \in [n]}\frac{X_{\pi(2j-1)}+X_{\pi(2j)}}{2}$ then converges weakly to a normal distribution with zero mean and variance $\sigma^2$ where

\begin{align*}
\sigma^2 = & \plim_{n\rightarrow \infty} \frac{1}{n} \sum_{j \in [n]}\left(\frac{X_{\pi(2j-1)}+X_{\pi(2j)}}{2}\right)^2 \\
= & \plim_{n\rightarrow \infty}\left[\frac{1}{n} \sum_{j \in [n]}\frac{X^2_{\pi(2j-1)} + X^2_{\pi(2j)} + 2X_{\pi(2j-1)} X_{\pi(2j)}}{4}\right] \\
= & \plim_{n\rightarrow \infty}\left[\frac{1}{n} \sum_{j \in [n]}\frac{2X^2_{\pi(2j-1)} + 2X^2_{\pi(2j)} - (X_{\pi(2j-1)} - X_{\pi(2j)})^2}{4} \right]\\
= & \plim_{n\rightarrow \infty}\frac{1}{2n} \sum_{i \in [2n]} X^2_{i} - \plim_{n\rightarrow \infty}\frac{1}{4n}\sum_{j \in [n]}(X_{\pi(2j-1)} - X_{\pi(2j)})^2 \\
= & Var(X) \neq \frac{Var(X)}{2},
\end{align*}
where the last equality holds because of Assumption \ref{ass:assignment1}(iv). This means that $\{X_{\pi(2j-1)}+X_{\pi(2j)}\}_{j \in [n]}$ are not independent across $j$. By assigning independent multipliers to pairs, bootstrapping pairs fails to mimic such dependence of covariates across pairs.\\

\section{Proof of Theorem \ref{thm:est}}
\label{sec:thmest}
Let $u=(u_0, u_1)^\top  \in \Re^2$ and
\begin{align*}
L_n(u,\tau)  = \sum_{i=1}^{2n}\left[\rho_\tau(Y_i - \dot{A}_i^\top \beta(\tau) - \dot{A}_i^\top u/\sqrt{n}) - \rho_\tau(Y_i - \dot{A}_i^\top \beta(\tau))\right].
\end{align*}
Then, by change of variables we have
\begin{align*}
\sqrt{n}(\hat{\beta}(\tau) - \beta(\tau)) = \argmin_u L_n(u,\tau).
\end{align*}
Note that $L_n(u,\tau)$ is convex in $u$ for each $\tau$ and bounded in $\tau$ for each $u$. We divide the proof into three steps. In Step (1), we show that there exists
$$g_n(u,\tau) = - u^\top W_n(\tau) + \frac{u^\top Q(\tau)u}{2}$$
for some $W_n(\tau)$ and $Q(\tau)$ such that for each $u$,
\begin{align*}
\sup_{\tau \in \Upsilon}|L_n(u,\tau) - g_n(u,\tau)| \convP 0,
\end{align*}
and the maximum eigenvalue of $Q(\tau)$ is bounded from above and the minimum eigenvalue of $Q(\tau)$ is bounded away from $0$, uniformly over $\tau \in \Upsilon$. In Step (2), we show $\sup_{\tau \in \Upsilon}||W_n(\tau)||_2= O_p(1)$. Then by \citet[Theorem 2]{K09}, we have
\begin{align*}
\sqrt{n}(\hat{\beta}(\tau) - \beta(\tau)) = [Q(\tau)]^{-1}W_n(\tau) + r_n(\tau),
\end{align*}
where $\sup_{\tau \in \Upsilon}||r_n(\tau)||_2 = o_p(1)$. Last, in Step (3), we establish weak convergence of $[Q(\tau)]^{-1}W_n(\tau)$ uniformly over $\tau \in \Upsilon$. The second element of the limit  process is $\mathcal{B}(\tau)$, as given in Theorem \ref{thm:est}.

\vspace{2mm}
\noindent \textbf{Step (1).} By Knight's identity (\citep{K98}), we have
\begin{align*}
& L_n(u,\tau) \\
= & -\sum_{i=1}^{2n}\frac{u^\top }{\sqrt{n}}\dot{A}_i\left(\tau- 1\{Y_i\leq \dot{A}_i^\top \beta(\tau)\}\right) + \sum_{i=1}^{2n}\int_0^{\frac{\dot{A}_i^\top u}{\sqrt{n}}}\left(1\{Y_i -  \dot{A}_i^\top \beta(\tau)\leq v\} - 1\{Y_i -  \dot{A}_i^\top \beta(\tau)\leq 0\} \right)dv \\
\equiv & -u^\top W_n(\tau) +Q_n(u,\tau),
\end{align*}
where
\begin{align*}
W_n(\tau) = \sum_{i=1}^{2n}\frac{1}{\sqrt{n}}\dot{A}_i\left(\tau- 1\{Y_i\leq \dot{A}_i^\top \beta(\tau)\}\right),
\end{align*}
and
\begin{align}
\label{eq:Qn}
Q_n(u,\tau) = & \sum_{i=1}^{2n}\int_0^{\frac{\dot{A}_i^\prime u}{\sqrt{n}}}\left(1\{Y_i -  \dot{A}_i^\top \beta(\tau)\leq v\} - 1\{Y_i -  \dot{A}_i^\top \beta(\tau)\leq 0\} \right)dv \notag  \\
= & \sum_{i=1}^{2n} A_i\int_0^{\frac{u_0+u_1}{\sqrt{n}}}\left(1\{Y_i(1) - q_1(\tau)\leq v\} - 1\{Y_i(1) - q_1(\tau)\leq 0\}  \right)dv \notag \\
& + \sum_{i=1}^{2n} (1-A_i)\int_0^{\frac{u_0}{\sqrt{n}}}\left(1\{Y_i(0) - q_0(\tau)\leq v\} - 1\{Y_i(0) - q_0(\tau)\leq 0\}  \right)dv \notag \\
\equiv & Q_{n,1}(u,\tau) + Q_{n,0}(u,\tau).
\end{align}
We first consider $Q_{n,1}(u,\tau)$. Let
\begin{align}
\label{eq:H}
H_n(X_i,\tau) = \mathbb{E}\left(\int_0^{\frac{u_0+u_1}{\sqrt{n}}}\left(1\{Y_i(1) -  q_1(\tau)\leq v\} - 1\{Y_i(1) -  q_1(\tau)\leq 0\} \right)dv\biggl|X_i\right).
\end{align}
Then,
\begin{align}
\label{eq:Qn1}
Q_{n,1}(u,\tau) = & \sum_{i=1}^{2n}\frac{H_n(X_i,\tau)}{2} + \sum_{i=1}^{2n} \left(A_i-\frac{1}{2}\right)H_n(X_i,\tau) \notag \\
& + \sum_{i=1}^{2n} A_i \left[\int_0^{\frac{u_0+u_1}{\sqrt{n}}}\left(1\{Y_i(1) -  q_1(\tau)\leq v\} - 1\{Y_i(1) -  q_1(\tau)\leq 0\} \right)dv- H_n(X_i,\tau)\right].
\end{align}
For the first term on the RHS of \eqref{eq:Qn1}, we have, uniformly over $\tau \in \Upsilon$,
\begin{align}
\label{eq:Q11}
\sum_{i=1}^{2n}\frac{H_n(X_i,\tau)}{2} = \frac{1}{4n}\sum_{i=1}^{2n}f_{1}( q_1(\tau)+\tilde{v}|X_i)(u_0+u_1)^2 \convP  \frac{f_{1}(q_1(\tau))(u_0+u_1)^2}{2},
\end{align}
where $\tilde{v}$ is between $0$ and $|u_0+u_1|/\sqrt{n}$ and the convergence holds because
\begin{align*}
\sup_{\tau \in \Upsilon}\frac{1}{2n}\sum_{i=1}^{2n}\left|f_{1}( q_1(\tau)+\tilde{v}|X_i)-f_{1}( q_1(\tau)|X_i)\right| \leq \left(\frac{1}{2n}\sum_{i=1}^{2n}C(X_i)\right) \frac{|u_0+u_1|}{\sqrt{n}} \convP 0
\end{align*}
due to Assumption \ref{ass:reg}, and
\begin{align*}
\sup_{\tau \in \Upsilon}\left|\frac{1}{2n}\sum_{i =1}^{2n}f_1(q_1(\tau)|X_i) - f_1(q_1(\tau))\right| \convP 0.
\end{align*}

Lemma \ref{lem:Q1} shows
\begin{align}
\label{eq:Q12}
& \sup_{\tau \in \Upsilon}\left|\sum_{i=1}^{2n} \left(A_i-\frac{1}{2}\right)H_n(X_i,\tau)\right| =o_p(1),
\end{align}
and
\begin{align}
\label{eq:Q13}
\sup_{\tau \in \Upsilon}&\left|\sum_{i=1}^{2n} A_i \left[\int_0^{\frac{u_0+ u_1}{\sqrt{n}}}\left(1\{Y_i(1) -  q_1(\tau)\leq v\} - 1\{Y_i(1) -  q_1(\tau)\leq 0\} \right)dv- H_n(X_i,\tau)\right]\right| = o_p(1).
\end{align}

Combining \eqref{eq:Qn1}--\eqref{eq:Q13}, we have
\begin{align}
\label{eq:Qn1f}
\sup_{\tau \in \Upsilon}\left|Q_{n,1}(u,\tau) - \frac{f_1(q_1(\tau))(u_0+u_1)^2}{2}\right| = o_p(1).
\end{align}

By a similar argument, we can show that
\begin{align}
\label{eq:Qn0f}
\sup_{\tau \in \Upsilon}\left|Q_{n,0}(u,\tau) - \frac{f_0(q_0(\tau))u_0^2}{2}\right| = o_p(1).
\end{align}

Combining \eqref{eq:Qn1f} and \eqref{eq:Qn0f}, we have
\begin{align*}
Q_n(u,\tau) \convP  \frac{u^\top Q(\tau)u}{2},
\end{align*}
where
\begin{align}
\label{eq:Qqr}
Q(\tau) = \begin{pmatrix}
f_1(q_1(\tau)) + f_0(q_0(\tau)) & f_1(q_1(\tau)) \\
f_1(q_1(\tau)) & f_1(q_1(\tau))
\end{pmatrix}.
\end{align}
Then,
\begin{align*}
\sup_{\tau \in \Upsilon}|L_n(u,\tau) - g_n(u,\tau)| = \sup_{\tau \in \Upsilon}\left|Q_n(u,\tau) - \frac{u^\top Q(\tau)u}{2}\right| = o_p(1).
\end{align*}
Last, because $f_a(q_a(\tau))$ for $a=0,1$ is bounded and bounded away from zero uniformly over $\tau \in \Upsilon$, so are the eigenvalues of $Q(\tau)$ uniformly over $\tau \in \Upsilon$.

\vspace{2mm}
\noindent \textbf{Step (2).}
Let $e_1 = (1,1)^\top $, $e_0 = (1,0)^\top $. Then,
\begin{align}
\label{eq:wn0}
W_n(\tau) = & \sum_{i=1}^{2n}\frac{e_1}{\sqrt{n}}A_i\left(\tau- 1\{Y_i(1) \leq q_1(\tau)\}\right) + \sum_{i=1}^{2n}\frac{e_0}{\sqrt{n}}(1-A_i)\left(\tau- 1\{Y_i(0)\leq q_0(\tau)\}\right) \notag \\
\equiv & e_1 W_{n,1}(\tau) + e_0 W_{n,0}(\tau). \notag \\
\end{align}
Recall $m_{1,\tau}(X_i) = \mathbb{E}(\tau - 1\{Y_{i}(1) \leq q_1(\tau)\}|X_i)$. Denote
$$\eta_{i,1}(\tau) = \tau - 1\{Y_{i}(1) \leq q_1(\tau)\} - m_{1,\tau}(X_i).$$
For $W_{n,1}(\tau)$, we have
\begin{align}
\label{eq:wn}
W_{n,1}(\tau) = \sum_{i=1}^{2n}\frac{A_i}{\sqrt{n}}\eta_{i,1}(\tau) +  \sum_{i=1}^{2n} \frac{1}{2\sqrt{n}}m_{1,\tau}(X_i) + R_1(\tau)
\end{align}
where
\begin{align*}
R_1(\tau) = \sum_{i=1}^{2n} \frac{(A_i - 1/2)}{\sqrt{n}}m_{1,\tau}(X_i).
\end{align*}
By Lemma \ref{lem:R12}, we have
\begin{align*}
\sup_{\tau \in \Upsilon}|R_{1}(\tau)|= o_p(1).
\end{align*}

Next, we focus on the first two terms on the RHS of \eqref{eq:wn}. Note $\{Y_i(1)\}_{i=1}^{2n}$ given $\{X_i\}_{i=1}^{2n}$ is an independent sequence that is also independent of $\{A_i\}_{i=1}^{2n}$. Let $\tilde{Y}_j(1)|\tilde{X}_j$ be distributed according to $Y_{i_j}(1)|X_{i_j}$ where $i_j$ is the $j$-th smallest index in the set $\{i\in [2n]: A_i = 1 \}$ and $\tilde{X}_j = X_{i_j}$. Then, by noticing that $\sum_{i=1}^{2n}A_i = n$, we have
\begin{align}
\label{eq:=d}
\sum_{i=1}^{2n}\frac{A_i}{\sqrt{n}}\eta_{i,1}(\tau)|\{A_i,X_i\}_{i=1}^{2n} \stackrel{d}{=}\sum_{j=1}^{n}\frac{\tilde{\eta}_{j,1}(\tau)}{\sqrt{n}}\biggl| \{\tilde{X}_j\}_{j=1}^n,
\end{align}
where $\tilde{\eta}_{j,1}(\tau) = \tau - 1\{\tilde{Y}_j(1) \leq q_1(\tau)\} - m_{1,\tau}(\tilde{X}_j)$, and given $\{\tilde{X}_j\}_{j=1}^n$, $\{\tilde{\eta}_{j,1}(\tau)\}_{j=1}^n$ is a sequence of independent random variables. Further, denote the conditional distribution of $\tilde{Y}_j(1)$ given $\tilde{X}_j $ as $\mathbb{P}^{(j)}$ and $\Lambda_\tau(x) = F_1(q_1(\tau)|x)(1-F_1(q_1(\tau)|x))$. Then,
\begin{align*}
\frac{1}{n}\sum_{j=1}^n\mathbb{P}^{(j)}(\tilde{\eta}_{j,1}(\tau))^2 = & \; \frac{1}{n}\sum_{j=1}^n\Lambda_\tau(\tilde{X}_j) \\
= & \; \frac{1}{n}\sum_{i=1}^{2n} A_i\Lambda_\tau(X_i) \\
= &\; \frac{1}{2n}\sum_{i=1}^{2n} \Lambda_\tau(X_i) + \frac{1}{2n}\sum_{j=1}^{n} (A_{\pi(2j-1)}-A_{\pi(2j)}) \left[\Lambda_\tau(X_{\pi(2j-1)}) - \Lambda_\tau(X_{\pi(2j)})\right] \\
\convP &\; \mathbb{E}\Lambda_\tau(X_i),
\end{align*}
where the last convergence holds because
\begin{align*}
\frac{1}{2n}\sum_{i=1}^{2n} \Lambda_\tau(X_i) \convP  \mathbb{E}\Lambda_\tau(X_i),
\end{align*}
and
\begin{align*}
\left|\frac{1}{2n}\sum_{j=1}^{n} (A_{\pi(2j-1)}-A_{\pi(2j)}) \left[\Lambda_\tau(X_{\pi(2j-1)}) - \Lambda_\tau(X_{\pi(2j)})\right] \right| \lesssim \frac{1}{2n}\sum_{j=1}^{n}||X_{\pi(2j-1)} - X_{\pi(2j)}||_2 \convP 0.
\end{align*}
In addition, because $\tilde{\eta}_{j,1}(\tau)$ is bounded, the Lyapounov condition holds, i.e.,
\begin{align*}
\frac{1}{n^{3/2}}\sum_{i=1}^n\mathbb{P}^{(j)}|\tilde{\eta}_{j,1}(\tau)|^3 \convP 0.
\end{align*}
Therefore, by the triangular array CLT, for fixed $\tau$, we have
\begin{align*}
\sum_{i=1}^{2n}\frac{A_i}{\sqrt{n}}\eta_{i,1}(\tau)|\{A_i,X_i\}_{i=1}^{2n} \stackrel{d}{=}\sum_{j=1}^{n}\frac{\tilde{\eta}_{j,1}(\tau)}{\sqrt{n}}\biggl| \{\tilde{X}_j\}_{j=1}^n\convD \N(0,\mathbb{E}\Lambda_\tau(X_i)) = O_p(1).
\end{align*}
It is straightforward to extend this result to finite-dimensional convergence by the Cram\'{e}r-Wold device. In particular, the covariance between $\sum_{i=1}^{2n}\frac{A_i}{\sqrt{n}}\eta_{i,1}(\tau)$ and $\sum_{i=1}^{2n}\frac{A_i}{\sqrt{n}}\eta_{i,1}(\tau' )$ conditional on $\{A_i,X_i\}_{i=1}^{2n}$ converges to
\begin{align*}
\min(\tau,\tau' ) - \tau \tau'  - \mathbb{E}m_{1,\tau}(X)m_{1,\tau' }(X).
\end{align*}

Next, we show that the process $\{\sum_{i=1}^{2n}\frac{A_i}{\sqrt{n}}\eta_{i,1}(\tau):\tau \in \Upsilon\}$ is stochastically equicontinuous. Denote $\overline{\mathbb{P}}f = \frac{1}{n}\sum_{j =1}^n \mathbb{P}^{(j)}f$ for a generic function $f$. Let
\begin{align*}
\mathcal{F}_1 = \{ \left[\tau - 1\{Y \leq q_1(\tau)\}\right] -  \left[\tau'  - 1\{Y \leq q_1(\tau' )\}\right] :  \tau,\tau' \in \Upsilon, |\tau - \tau' |\leq \eps\}
\end{align*}
which is a VC-class with a fixed VC-index and has an envelope $F_i = 2$. In addition,
\begin{align*}
\sigma_n^2 = \sup_{f \in \mathcal{F}_1}\overline{\mathbb{P}}f^2 \lesssim \sup_{\tilde{\tau} \in \Upsilon}\frac{1}{n}\sum_{i=1}^n \left[\eps^2+\frac{f_1(q_1(\tilde{\tau})|\tilde{X}_j)\eps}{f_1(q_1(\tilde{\tau}))} \right] \lesssim \eps~a.s.
\end{align*}
Then, by Lemma \ref{lem:max_eq},
\begin{align*}
\mathbb{E}\left[\sup_{\tau, \tau' \in \Upsilon, |\tau-\tau' |\leq \eps}\left|\sum_{j=1}^{n}\frac{\tilde{\eta}_{j,1}(\tau) - \tilde{\eta}_{j,1}(\tau' )}{\sqrt{n}}\right| \biggl| \{\tilde{X}_j\}_{j=1}^n \right] = & \mathbb{E}\left[\Vert\mathbb{P}_n - \overline{\mathbb{P}}\Vert_{\mathcal{F}_1}\biggl| \{\tilde{X}_j\}_{j=1}^n \right] \\
\lesssim & \sqrt{\eps \log(1/\eps)} + \frac{\log(1/\eps)}{\sqrt{n}}~a.s.
\end{align*}
For any $\delta,\eta>0$, we can find an $\eps>0$ such that
\begin{align*}
& \limsup_{n}\mathbb{P}\left(\sup_{\tau, \tau' \in \Upsilon, |\tau-\tau' |\leq \eps}  \left|\sum_{i=1}^{2n}\frac{A_i}{\sqrt{n}}\left(\eta_{i,1}(\tau)-\eta_{i,1}(\tau' ) \right) \right|\geq \delta \right) \\
= & \limsup_{n}\mathbb{E}\mathbb{P}\left(\sup_{\tau, \tau' \in \Upsilon, |\tau-\tau' |\leq \eps}  \left|\sum_{i=1}^{2n}\frac{A_i}{\sqrt{n}}\left(\eta_{i,1}(\tau)-\eta_{i,1}(\tau' ) \right) \right|\geq \delta \biggl| \{A_i,X_i\}_{i=1}^{2n}\right) \\
\leq & \limsup_{n}\mathbb{E} \frac{\mathbb{E}\left[\sup_{\tau, \tau' \in \Upsilon, |\tau-\tau' |\leq \eps}\left|\sum_{j=1}^{n}\frac{\tilde{\eta}_{j,1}(\tau) - \tilde{\eta}_{j,1}(\tau' )}{\sqrt{n}}\right| \biggl| \{\tilde{X}_j\}_{j=1}^n \right]}{\delta} \\
\lesssim & \limsup_{n}\frac{\sqrt{\eps \log(1/\eps)} + \frac{\log(1/\eps)}{\sqrt{n}}}{\delta} \leq \eta,
\end{align*}
where the last inequality holds because $\eps \log(1/\eps) \rightarrow 0$ as $\eps\rightarrow 0$. This implies $\{\sum_{i=1}^{2n}\frac{A_i}{\sqrt{n}}\eta_{i,1}(\tau):\tau \in \Upsilon\}$ is stochastically equicontinuous, and hence $\sup_{\tau \in \Upsilon}|\sum_{i=1}^{2n}\frac{A_i}{\sqrt{n}}\eta_{i,1}(\tau)| = O_p(1)$.

In addition, note $\{X_i\}_{i=1}^{2n}$ are i.i.d. and $\{m_{1,\tau}(x):\tau \in \Upsilon\}$ is Donsker, then $\sup_{\tau \in \Upsilon}|\sum_{i=1}^{2n} \frac{1}{2\sqrt{n}}m_{1,\tau}(X_i)|=O_p(1)$. This leads to the desired result that $\sup_{\tau \in \Upsilon}|W_{n,1}(\tau) |= O_p(1)$. In the same manner, we can show that $\sup_{\tau \in \Upsilon}|W_{n,0}(\tau) |= O_p(1)$, which leads to the desired result that $\sup_{\tau \in \Upsilon}||W_{n}(\tau) ||_2= O_p(1)$.

\vspace{2mm}
\noindent \textbf{Step (3).} Recall $m_{0,\tau}(X_i) = \mathbb{E}(\tau - 1\{Y_{i}(0) \leq q_0(\tau)\}|X_i)$ and let $\eta_{i,0}(\tau) = \tau - 1\{Y_{i}(0) \leq q_0(\tau)\} - m_{0,\tau}(X_i)$. Then, based on the previous two steps, we have
\begin{align}
\label{eq:beta}
\sqrt{n}(\hat{\beta}(\tau) - \beta(\tau)) = Q^{-1}(\tau)\begin{pmatrix}
1 & 1 \\
1 & 0
\end{pmatrix} \begin{pmatrix} \sum_{i=1}^{2n}\frac{A_i}{\sqrt{n}}\eta_{i,1}(\tau) +  \sum_{i=1}^{2n} \frac{1}{2\sqrt{n}}m_{1,\tau}(X_i)\\
\sum_{i=1}^{2n}\frac{1-A_i}{\sqrt{n}}\eta_{i,0}(\tau) +  \sum_{i=1}^{2n} \frac{1}{2\sqrt{n}}m_{0,\tau}(X_i)
\end{pmatrix} + R(\tau),
\end{align}
where $\sup_{\tau \in \Upsilon} ||R(\tau)||_2 = o_p(1)$. In addition, we have already established the stochastic equicontinuity and finite-dimensional convergence of
\begin{align*}
\begin{pmatrix} \sum_{i=1}^{2n}\frac{A_i}{\sqrt{n}}\eta_{i,1}(\tau) +  \sum_{i=1}^{2n} \frac{1}{2\sqrt{n}}m_{1,\tau}(X_i)\\
\sum_{i=1}^{2n}\frac{1-A_i}{\sqrt{n}}\eta_{i,0}(\tau) +  \sum_{i=1}^{2n} \frac{1}{2\sqrt{n}}m_{0,\tau}(X_i)
\end{pmatrix}.
\end{align*}
Thus, in order to derive the weak limit of $\sqrt{n}(\hat{\beta}(\tau) - \beta(\tau))$ uniformly over $\tau \in \Upsilon$, it suffices to consider its covariance kernel. First, note that, by construction, $$\sum_{i=1}^{2n}\frac{A_i}{\sqrt{n}}\eta_{i,1}(\tau) \indep \sum_{i=1}^{2n}\frac{1-A_i}{\sqrt{n}}\eta_{i,0}(\tau' )\biggl| \{A_i,X_i\}_{i \in [2n]}$$ for any $(\tau,\tau' ) \in \Upsilon$. Second, note that  $\sum_{i=1}^{2n}\frac{A_i}{\sqrt{n}}\eta_{i,1}(\tau)$ is asymptotically independent of $ \sum_{i=1}^{2n} \frac{1}{2\sqrt{n}}m_{1,\tau' }(X_i)$. To see this, let $(s,t) \in \Re^2$, then
\begin{align*}
& \mathbb{P}\left(\sum_{i=1}^{2n}\frac{A_i}{\sqrt{n}}\eta_{i,1}(\tau)\leq t,\sum_{i=1}^{2n} \frac{1}{2\sqrt{n}}m_{1,\tau' }(X_i)\leq s \right) \\
= & \mathbb{E}\left\{\mathbb{P}\left(\sum_{i=1}^{2n}\frac{A_i}{\sqrt{n}}\eta_{i,1}(\tau)\leq t \biggl|\{A_i,X_i\}_{i=1}^{2n}\right)1\left\{\sum_{i=1}^{2n} \frac{1}{2\sqrt{n}}m_{1,\tau' }(X_i)\leq s  \right\}\right\} \\
= & \Phi(t/\sqrt{\mathbb{E}\Lambda_\tau(X_i)})\mathbb{P}\left(\sum_{i=1}^{2n} \frac{1}{2\sqrt{n}}m_{1,\tau' }(X_i)\leq s \right) \\
& + \mathbb{E}\biggl\{\left[\mathbb{P}\left(\sum_{i=1}^{2n}\frac{A_i}{\sqrt{n}}\eta_{i,1}(\tau)\leq t\biggl|\{A_i,X_i\}_{i=1}^{2n}\right) - \Phi(t/\sqrt{\mathbb{E}\Lambda_\tau(X_i)})\right]1\left\{\sum_{i=1}^{2n} \frac{1}{2\sqrt{n}}m_{1,\tau' }(X_i)\leq s   \right\}\biggr\} \\
\rightarrow & \Phi(t/\sqrt{\mathbb{E}\Lambda_\tau(X_i)})\Phi(s/\sqrt{\mathbb{E}m_{1,\tau}^2(X_i)/2}) \\
= & \lim_{n \rightarrow \infty}\mathbb{P}\left( \sum_{i=1}^{2n}\frac{A_i}{\sqrt{n}}\eta_{i,1}(\tau)\leq t\right)\mathbb{P}\left( \sum_{i=1}^{2n} \frac{1}{2\sqrt{n}}m_{1,\tau' }(X_i)\leq s \right),
\end{align*}
where $\Phi(\cdot)$ is the standard normal CDF and the second last line holds due to the fact that
\begin{align*}
\mathbb{P}\left(\sum_{i=1}^{2n}\frac{A_i}{\sqrt{n}}\eta_{i,1}(\tau)\leq t\biggl|\{A_i,X_i\}_{i=1}^{2n}\right) - \Phi(t/\sqrt{\mathbb{E}\Lambda_\tau(X_i)}) \convP 0.
\end{align*}

We can extend the independence result to multiple $\tau$ and $\tau' $, implying that the two stochastic processes
\begin{align*}
\left\{\sum_{i=1}^{2n}\frac{A_i}{\sqrt{n}}\eta_{i,1}(\tau): \tau \in \Upsilon \right\} \quad \text{and} \quad \left\{\sum_{i=1}^{2n} \frac{1}{2\sqrt{n}}m_{1,\tau}(X_i): \tau \in \Upsilon \right\}
\end{align*}
are asymptotically independent. For the same reason, we can show
\begin{align*}
\left\{\left(\sum_{i=1}^{2n}\frac{A_i}{\sqrt{n}}\eta_{i,1}(\tau),\sum_{i=1}^{2n}\frac{1-A_i}{\sqrt{n}}\eta_{i,0}(\tau)\right): \tau \in \Upsilon \right\} \quad \text{and} \quad \left\{\left(\sum_{i=1}^{2n} \frac{1}{2\sqrt{n}}m_{1,\tau}(X_i),\sum_{i=1}^{2n} \frac{1}{2\sqrt{n}}m_{0,\tau}(X_i)\right): \tau \in \Upsilon \right\}
\end{align*}
are asymptotically independent. Last, it is tedious but straightforward to show that, uniformly over $\tau \in \Upsilon$,
\begin{align*}
\left(\sum_{i=1}^{2n}\frac{A_i}{\sqrt{n}}\eta_{i,1}(\tau),\sum_{i=1}^{2n}\frac{1-A_i}{\sqrt{n}}\eta_{i,0}(\tau)\right)\biggl| \{A_i,X_i\}_{i \in [2n]} \convD \tilde{\mathcal{B}}_1(\tau),
\end{align*}
and
\begin{align*}
\left(\sum_{i=1}^{2n} \frac{1}{2\sqrt{n}}m_{1,\tau}(X_i),\sum_{i=1}^{2n} \frac{1}{2\sqrt{n}}m_{0,\tau}(X_i)\right)\convD \tilde{\mathcal{B}}_2(\tau),
\end{align*}
where $\tilde{\mathcal{B}}_1(\tau)$ and $\tilde{\mathcal{B}}_2(\tau)$ are two Gaussian processes with covariance kernels
\begin{align}
\label{eq:sigmatilde1}
\tilde{\Sigma}_1(\tau,\tau' ) =
\begin{pmatrix}
\mathbb{E}\left[\min(\tau,\tau' ) - \tau \tau'  - \mathbb{E}m_{1,\tau}(X)m_{1,\tau' }(X)  \right] & 0 \\
0 & \mathbb{E}\left[\min(\tau,\tau' ) - \tau \tau'  - \mathbb{E}m_{0,\tau}(X)m_{0,\tau' }(X)  \right]
\end{pmatrix}
\end{align}
and
\begin{align}
\label{eq:sigmatilde2}
\tilde{\Sigma}_2(\tau,\tau' ) =
\frac{1}{2}\begin{pmatrix}
\mathbb{E}m_{1,\tau}(X)m_{1,\tau' }(X) & \mathbb{E}m_{1,\tau}(X)m_{0,\tau' }(X) \\
\mathbb{E}m_{1,\tau' }(X)m_{0,\tau}(X) & \mathbb{E}m_{0,\tau}(X)m_{0,\tau' }(X)
\end{pmatrix}, \quad \text{respectively.}
\end{align}
This implies
$\sqrt{n}(\hat{\beta}(\tau) - \beta(\tau)) \convD \tilde{\mathcal{B}}(\tau)$, where $\tilde{\mathcal{B}}(\tau)$ is a tight Gaussian process with covariance kernel
\begin{align*}
\tilde{\Sigma}(\tau,\tau' ) = Q^{-1}(\tau)\begin{pmatrix}
1 & 1 \\
1 & 0
\end{pmatrix} \left(\tilde{\Sigma}_1(\tau,\tau' )+\tilde{\Sigma}_2(\tau,\tau' ) \right) \left[Q^{-1}(\tau' )\begin{pmatrix}
1 & 1 \\
1 & 0
\end{pmatrix} \right]^\top .
\end{align*}
Focusing on the second element of $\hat{\beta}(\tau)$, we have
\begin{align*}
\sqrt{n}(\hat{q}(\tau) - q(\tau)) \convD \mathcal{B}(\tau),
\end{align*}
where $\mathcal{B}(\tau)$ is a tight Gaussian process with covariance kernel
\begin{align*}
\Sigma(\tau,\tau' ) = & \frac{\min(\tau, \tau' ) - \tau\tau'  - \mathbb{E}m_{1,\tau}(X)m_{1,\tau' }(X)}{f_1(q_1(\tau))f_1(q_1(\tau' ))} +  \frac{\min(\tau, \tau' ) - \tau\tau'  - \mathbb{E}m_{0,\tau}(X)m_{0,\tau' }(X)}{f_0(q_0(\tau))f_0(q_0(\tau' ))} \\
& + \frac{1}{2}\mathbb{E}\left(\frac{m_{1,\tau}(X)}{f_1(q_1(\tau))} - \frac{m_{0,\tau}(X)}{f_0(q_0(\tau))}\right)\left(\frac{m_{1,\tau' }(X)}{f_1(q_1(\tau' ))} - \frac{m_{0,\tau' }(X)}{f_0(q_0(\tau' ))}\right).
\end{align*}

\section{Proof of Theorem \ref{thm:weight}}
\label{sec:weight_pf}
Let $u=(u_0, u_1)^\top  \in \Re^2$ and
\begin{align*}
L^m_n(u,\tau)  = \sum_{i=1}^{2n}\xi_i\left[\rho_\tau(Y_i - \dot{A}_i^\top \beta(\tau) - \dot{A}_i^\top u/\sqrt{n}) - \rho_\tau(Y_i - \dot{A}_i^\top \beta(\tau))\right],
\end{align*}
where $\{\xi\}_{i \in [2n]}$ is defined in Assumption \ref{ass:weight}. Then, by change of variables we have
\begin{align*}
\sqrt{n}(\hat{\beta}^m(\tau) - \beta(\tau)) = \argmin_u L^m_n(u,\tau).
\end{align*}
Notice that $L^m_n(u,\tau)$ is convex in $u$ for each $\tau$ and bounded in $\tau$ for each $u$. In the following, we divide the proof into three steps. In Step (1), we show that there exists
$$g^m_n(u,\tau) = - u^\top W^m_n(\tau) + \frac{u^\top Q(\tau)u}{2}$$
for some $W^m_n(\tau)$ such that for each $u$,
\begin{align*}
\sup_{\tau \in \Upsilon}|L^m_n(u,\tau) - g^m_n(u,\tau)| \convP 0
\end{align*}
and $Q(\tau)$ is defined in the proof of Theorem \ref{thm:est}. In Step (2), we show $\sup_{\tau \in \Upsilon}||W^m_n(\tau)||_2= O_p(1)$. Then by \citet[Theorem 2]{K09}, we have
\begin{align*}
\sqrt{n}(\hat{\beta}^m(\tau) - \beta(\tau)) = [Q(\tau)]^{-1}W^m_n(\tau) + r_n(\tau),
\end{align*}
where $\sup_{\tau \in \Upsilon}||r_n(\tau)||_2 = o_p(1)$. Last, in Step (3), we establish the weak convergence of
\begin{align*}
\sqrt{n}(\hat{\beta}^m(\tau) - \hat{\beta}(\tau))
\end{align*}
conditionally on data. The proof of the first two steps are the same as those in Section \ref{sec:pair_pf} with $\xi_i$ replaced by $\xi_i^m$, and thus, is omitted for brevity. In the following, we focus on the third step.

\vspace{1mm}
\noindent \textbf{Step (3).} Based on the previous two steps, we have
\begin{align}
\label{eq:betaw}
\sqrt{n}(\hat{\beta}^m(\tau) - \beta(\tau)) = Q^{-1}(\tau)\begin{pmatrix}
1 & 1 \\
1 & 0
\end{pmatrix} \begin{pmatrix} \sum_{i=1}^{2n}\frac{A_i}{\sqrt{n}}\eta^m_{i,1}(\tau) +  \sum_{i=1}^{2n} \frac{1}{2\sqrt{n}}m_{1,\tau}(X_i)\\
\sum_{i=1}^{2n}\frac{1-A_i}{\sqrt{n}}\eta^m_{i,0}(\tau) +  \sum_{i=1}^{2n} \frac{1}{2\sqrt{n}}m_{0,\tau}(X_i)
\end{pmatrix} + R^m(\tau),
\end{align}
where $\eta_{i,a}^m(\tau) = \xi_i(\tau - 1\{Y_i(a) \leq q_a(\tau)\}) - m_{a,\tau}(X_i)$, $\sup_{\tau \in \Upsilon} ||R^m(\tau)||_2 = o_p(1)$, and $\sqrt{n}(\hat{\beta}^m(\tau) - \beta(\tau))$ is stochastically equicontinuous.
Taking the difference between \eqref{eq:beta} and \eqref{eq:betaw}, we have
\begin{align}
\label{eq:betaw2}
\sqrt{n}(\hat{\beta}^m(\tau) - \hat{\beta}(\tau)) = Q^{-1}(\tau)\begin{pmatrix}
1 & 1 \\
1 & 0
\end{pmatrix} \begin{pmatrix} \sum_{i=1}^{2n}\frac{A_i}{\sqrt{n}}(\xi_i - 1)(\tau - 1\{Y_i(1) \leq q_1(\tau)\}) \\
\sum_{i=1}^{2n}\frac{1-A_i}{\sqrt{n}}(\xi_i - 1)(\tau - 1\{Y_i(0) \leq q_0(\tau)\})
\end{pmatrix} + R^*(\tau),
\end{align}
where $\sup_{\tau \in \Upsilon}||R^*(\tau)||_2 = o_p(1)$. In addition, because both $\sqrt{n}(\hat{\beta}^m(\tau) - \beta(\tau))$ and $\sqrt{n}(\hat{\beta}(\tau) - \beta(\tau))$ are stochastically equicontinuous, so be $\sqrt{n}(\hat{\beta}^m(\tau) - \hat{\beta}(\tau))$. Then by the Markov inequality, $\sqrt{n}(\hat{\beta}^m(\tau) - \hat{\beta}(\tau))$ is stochastically equicontinuous conditionally on data as well. In order to derive the limit  distribution of  $\sqrt{n}(\hat{\beta}^m(\tau) - \hat{\beta}(\tau))$ conditionally on data, we need only compute the covariance kernel. Note that
\small
\begin{align*}
& \mathbb{E}\left[\begin{pmatrix} \sum_{i=1}^{2n}\frac{A_i}{\sqrt{n}}(\xi_i - 1)(\tau - 1\{Y_i(1) \leq q_1(\tau)\}) \\
\sum_{i=1}^{2n}\frac{1-A_i}{\sqrt{n}}(\xi_i - 1)(\tau - 1\{Y_i(0) \leq q_0(\tau)\})
\end{pmatrix}\begin{pmatrix} \sum_{i=1}^{2n}\frac{A_i}{\sqrt{n}}(\xi_i - 1)(\tau'  - 1\{Y_i(1) \leq q_1(\tau' )\}) \\
\sum_{i=1}^{2n}\frac{1-A_i}{\sqrt{n}}(\xi_i - 1)(\tau'  - 1\{Y_i(0) \leq q_0(\tau' )\})
\end{pmatrix}^\top \biggl|Data\right] \\
= & \frac{1}{n}\sum_{i=1}^{2n}\begin{pmatrix}
A_i(\tau - 1\{Y_i(1) \leq q_1(\tau)\})(\tau'  - 1\{Y_i(1) \leq q_1(\tau' )\}) & 0 \\
0 & (1-A_i)(\tau - 1\{Y_i(0) \leq q_0(\tau)\})(\tau'  - 1\{Y_i(0) \leq q_0(\tau' )\})
\end{pmatrix}.
\end{align*}
\normalsize
For the $(1,1)$ entry, we have
\begin{align*}
& \frac{1}{n}\sum_{i=1}^{2n}A_i(\tau - 1\{Y_i(1) \leq q_1(\tau)\})(\tau'  - 1\{Y_i(1) \leq q_1(\tau' )\}) \\
&=  \frac{1}{n}\sum_{i =1}^{2n}A_i \eta_{1,i}(\tau)\eta_{1,i}(\tau' ) + \frac{1}{n}\sum_{i =1}^{2n}A_i \eta_{1,i}(\tau) m_{1,\tau' }(X_i)+\frac{1}{n}\sum_{i =1}^{2n}A_i \eta_{1,i}(\tau' ) m_{1,\tau}(X_i)+ \frac{1}{n}\sum_{i =1}^{2n}A_i m_{1,\tau}(X_i)m_{1,\tau' }(X_i).
\end{align*}
Note that
\begin{align}
\label{eq:etaeta}
\frac{1}{n}\sum_{i =1}^{2n}A_i \eta_{1,i}(\tau)\eta_{1,i}(\tau' ) \stackrel{d}{=} & \frac{1}{n}\sum_{j =1}^{n} \tilde{\eta}_{1,j}(\tau)\tilde{\eta}_{1,j}(\tau' ) \notag \\
\convP &\; \lim_n \frac{1}{n}\sum_{j =1}^{n}(F_1(q_1(\min(\tau, \tau' ))|\tilde{X}_j) - F_1(q_1(\tau)|\tilde{X}_j)F_1(q_1(\tau' )|\tilde{X}_j))  \notag \\
= &\; \min(\tau, \tau' ) - \mathbb{E}F_1(q_1(\tau)|X_i)F_1(q_1(\tau' )|X_i).
\end{align}
Lemma \ref{lem:w1} shows
\begin{align*}
\frac{1}{n}\sum_{i =1}^{2n}A_i \eta_{1,i}(\tau) m_{1,\tau' }(X_i) \convP 0
\end{align*}
and
\begin{align*}
\frac{1}{n}\sum_{i =1}^{2n}A_i \eta_{1,i}(\tau' ) m_{1,\tau}(X_i) \convP 0.
\end{align*}
Lemma \ref{lem:4} implies
\begin{align*}
\frac{1}{n}\sum_{i =1}^{2n}A_i m_{1,\tau}(X_i)m_{1,\tau' }(X_i) = & \frac{1}{2n}\sum_{i =1}^{2n} m_{1,\tau}(X_i)m_{1,\tau' }(X_i) + \frac{1}{n}\sum_{i =1}^{2n}(A_i -\frac{1}{2})m_{1,\tau}(X_i)m_{1,\tau' }(X_i)\\ \convP & \ \mathbb{E}m_{1,\tau}(X_i)m_{1,\tau' }(X_i).
\end{align*}
This means
\begin{align*}
\frac{1}{n}\sum_{i=1}^{2n}A_i(\tau - 1\{Y_i(1) \leq q_1(\tau)\})(\tau'  - 1\{Y_i(1) \leq q_1(\tau' )\}) \convP \min(\tau, \tau' ) - \tau\tau' .
\end{align*}
For the same reason,
\begin{align*}
\frac{1}{n}\sum_{i=1}^{2n}(1-A_i)(\tau - 1\{Y_i(0) \leq q_0(\tau)\})(\tau'  - 1\{Y_i(0) \leq q_0(\tau' )\} \convP \min(\tau, \tau' ) - \tau\tau' .
\end{align*}
Then, for the second element $\hat{\beta}_1^m(\tau)$ of $\hat{\beta}^m(\tau)$, conditional on the data, we have
\begin{align*}
\sqrt{n}(\hat{\beta}^m_1(\tau) - \hat{\beta}_1(\tau)) \convD \mathcal{B}^m(\tau),
\end{align*}
where $\mathcal{B}^m(\tau)$ is a Gaussian process with covariance kernel
\begin{align*}
\Sigma^\dagger(\tau,\tau' ) = \frac{ \min(\tau,\tau' ) - \tau \tau'  }{f_1(q_1(\tau))f_1(q_1(\tau' ))} + \frac{ \min(\tau,\tau' ) - \tau \tau'  }{f_0(q_0(\tau))f_0(q_0(\tau' ))}.
\end{align*}

\section{Proof of Theorem \ref{thm:pair}}
\label{sec:pair_pf}
Let $u=(u_0, u_1)^\top  \in \Re^2$ and
\begin{align*}
L^p_n(u,\tau)  = \sum_{i=1}^{2n}\xi_i^p\left[\rho_\tau(Y_i - \dot{A}_i^\top \beta(\tau) - \dot{A}_i^\top u/\sqrt{n}) - \rho_\tau(Y_i - \dot{A}_i^\top \beta(\tau))\right],
\end{align*}
where $\xi_i^p$ is defined in Assumption \ref{ass:pair_boot}. Then, by change of variables we have
\begin{align*}
\sqrt{n}(\hat{\beta}^p(\tau) - \beta(\tau)) = \argmin_u L^p_n(u,\tau).
\end{align*}
Notice that $L^p_n(u,\tau)$ is convex in $u$ for each $\tau$ and bounded in $\tau$ for each $u$. In the following, we divide the proof into three steps. In Step (1), we show that there exists
$$g^p_n(u,\tau) = - u^\top W^p_n(\tau) + \frac{u^\top Q(\tau)u}{2}$$
for some $W_n^p(\tau)$ such that for each $u$,
\begin{align*}
\sup_{\tau \in \Upsilon}|L^p_n(u,\tau) - g^p_n(u,\tau)| \convP 0
\end{align*}
and $Q(\tau)$ is defined in the proof of Theorem \ref{thm:est}. In Step (2), we show $\sup_{\tau \in \Upsilon}||W^p_n(\tau)||_2=O_p(1)$. Then by \citet[Theorem 2]{K09}, we have
\begin{align*}
\sqrt{n}(\hat{\beta}^p(\tau) - \beta(\tau)) = [Q(\tau)]^{-1}W^p_n(\tau) + r_n(\tau),
\end{align*}
where $\sup_{\tau \in \Upsilon}||r_n(\tau)||_2 = o_p(1)$. Last, in Step (3), we establish the weak convergence of
\begin{align*}
\sqrt{n}(\hat{\beta}^p(\tau) - \hat{\beta}(\tau))
\end{align*}
conditionally on data.
%

\vspace{2mm}
\noindent \textbf{Step (1).} We have
\begin{align*}
& L^p_n(u,\tau) =  -u^\top W^p_n(\tau) +Q^p_n(u,\tau),
\end{align*}
where
\begin{align*}
W^p_n(\tau) = \sum_{i=1}^{2n}\frac{\xi_i^p}{\sqrt{n}}\dot{A}_i\left(\tau- 1\{Y_i\leq \dot{A}_i^\top \beta(\tau)\}\right)
\end{align*}
and
\begin{align}
\label{eq:Qnm}
Q^p_n(u,\tau) = & \sum_{i=1}^{2n}\xi_i^p\int_0^{\frac{\dot{A}_i^\prime u}{\sqrt{n}}}\left(1\{Y_i -  \dot{A}_i^\top \beta(\tau)\leq v\} - 1\{Y_i -  \dot{A}_i^\top \beta(\tau)\leq 0\} \right)dv \notag  \\
= & \sum_{i=1}^{2n} \xi_i^pA_i\int_0^{\frac{u_0+u_1}{\sqrt{n}}}\left(1\{Y_i(1) - q_1(\tau)\leq v\} - 1\{Y_i(1) - q_1(\tau)\leq 0\}  \right)dv \notag \\
& + \sum_{i=1}^{2n} \xi_i^p(1-A_i)\int_0^{\frac{u_0}{\sqrt{n}}}\left(1\{Y_i(0) - q_0(\tau)\leq v\} - 1\{Y_i(0) - q_0(\tau)\leq 0\}  \right)dv \notag \\
\equiv & Q^p_{n,1}(u,\tau) + Q^p_{n,0}(u,\tau).
\end{align}
We first consider $Q^p_{n,1}(u,\tau)$. Note
\begin{align}
\label{eq:Hm}
H_n(X_i,\tau) = & \ \mathbb{E}\xi_i^p\left(\int_0^{\frac{u_0+u_1}{\sqrt{n}}}\left(1\{Y_i(1) -  q_1(\tau)\leq v\} - 1\{Y_i(1) -  q_1(\tau)\leq 0\} \right)dv \biggl|X_i\right) \notag \\
= & \ \mathbb{E}\left(\int_0^{\frac{u_0+u_1}{\sqrt{n}}}\left(1\{Y_i(1) -  q_1(\tau)\leq v\} - 1\{Y_i(1) -  q_1(\tau)\leq 0\} \right)dv\biggl|X_i\right).
\end{align}
Then,
\begin{align}
\label{eq:Qn1m}
Q^p_{n,1}(u,\tau) = & \sum_{i=1}^{2n}\frac{H_n(X_i,\tau)}{2} + \sum_{i=1}^{2n} \left(A_i-\frac{1}{2}\right)H_n(X_i,\tau) \notag \\
& + \sum_{i=1}^{2n} A_i \left[\xi_i^p\int_0^{\frac{u_0+u_1}{\sqrt{n}}}\left(1\{Y_i(1) -  q_1(\tau)\leq v\} - 1\{Y_i(1) -  q_1(\tau)\leq 0\} \right)dv- H_n(X_i,\tau)\right].
\end{align}
By \eqref{eq:Q11}, we have, uniformly over $\tau \in \Upsilon$,
\begin{align*}
\sum_{i=1}^{2n}\frac{H_n(X_i,\tau)}{2}  \convP  \frac{f_{1}(q_1(\tau))(u_0+u_1)^2}{2}.
\end{align*}
In addition, \eqref{eq:Q12} implies
\begin{align*}
\sup_{\tau \in \Upsilon}\left|\sum_{i=1}^{2n} \left(A_i-\frac{1}{2}\right)H_n(X_i,\tau)\right| =o_p(1).
\end{align*}
Last, Lemma \ref{lem:Q1m} implies
\begin{align*}
\sup_{\tau \in \Upsilon}&\left|\sum_{i=1}^{2n} A_i\left[ \xi_i^p\int_0^{\frac{u_0+ u_1}{\sqrt{n}}}\left(1\{Y_i(1) -  q_1(\tau)\leq v\} - 1\{Y_i(1) -  q_1(\tau)\leq 0\} \right)dv- H_n(X_i,\tau)\right]\right| = o_p(1).
\end{align*}
Combining the above results, we have
\begin{align}
\label{eq:Qn1fm}
\sup_{\tau \in \Upsilon}\left|Q^p_{n,1}(u,\tau) - \frac{f_1(q_1(\tau))(u_0+u_1)^2}{2}\right| = o_p(1).
\end{align}

By a similar argument, we can show that
\begin{align}
\label{eq:Qn0fm}
\sup_{\tau \in \Upsilon}\left|Q^p_{n,0}(u,\tau) - \frac{f_0(q_0(\tau))u_0^2}{2}\right| = o_p(1).
\end{align}
Combining \eqref{eq:Qn1fm} and \eqref{eq:Qn0fm}, we have
\begin{align*}
\sup_{\tau \in \Upsilon}|L^p_n(u,\tau) - g^p_n(u,\tau)| = \sup_{\tau \in \Upsilon}\left|Q^p_n(u,\tau) - \frac{u^\top Q(\tau)u}{2}\right| = o_p(1).
\end{align*}

\vspace{2mm}
\noindent \textbf{Step (2).} We have
\begin{align}
\label{eq:wn0m}
W^p_n(\tau) = & \sum_{i=1}^{2n}\frac{e_1}{\sqrt{n}}\xi_i^pA_i\left(\tau- 1\{Y_i(1) \leq q_1(\tau)\}\right) + \sum_{i=1}^{2n}\frac{e_0}{\sqrt{n}}(1-A_i)\xi_i^p\left(\tau- 1\{Y_i(0)\leq q_0(\tau)\}\right) \notag \\
\equiv & e_1 W^p_{n,1}(\tau) + e_0 W^p_{n,0}(\tau). \notag \\
\end{align}
Recall $m_{1,\tau}(X_i) = \mathbb{E}(\tau - 1\{Y_{i}(1) \leq q_1(\tau)\}|X_i)$, $e_1 = (1,1)^\top $, and $e_0 = (1,0)^\top $, and denote
$$\eta^p_{i,1}(\tau) = \xi_i^p(\tau - 1\{Y_{i}(1) \leq q_1(\tau)\}) - m_{1,\tau}(X_i).$$
Then, for $W^p_{n,1}(\tau)$, we have
\begin{align}
\label{eq:wnm}
W^p_{n,1}(\tau) = \sum_{i=1}^{2n}\frac{A_i}{\sqrt{n}}\eta_{i,1}^p(\tau) +  \sum_{i=1}^{2n} \frac{1}{2\sqrt{n}}m_{1,\tau}(X_i) + R_1(\tau),
\end{align}
where by Lemma \ref{lem:R12},
\begin{align*}
\sup_{\tau \in \Upsilon}|R_{1}(\tau)| = \sup_{\tau \in \Upsilon}\left|\sum_{i=1}^{2n} \frac{(A_i - 1/2)}{\sqrt{n}}m_{1,\tau}(X_i)\right| = o_p(1).
\end{align*}

The second term on the RHS of \eqref{eq:wnm} is stochastically equicontinuous and uniformly bounded in probability. Next, we focus on the first term. Similar to the argument in Step (2) in Section \ref{sec:thmest}, we have
\begin{align}
\label{eq:=dm}
\sum_{i=1}^{2n}\frac{A_i}{\sqrt{n}}\eta^p_{i,1}(\tau)|\{A_i,X_i\}_{i=1}^{2n} \stackrel{d}{=}\sum_{j=1}^{n}\frac{\tilde{\eta}^p_{j,1}(\tau)}{\sqrt{n}}\biggl| \{\tilde{X}_j\}_{j=1}^n,
\end{align}
where $\tilde{\eta}^p_{j,1}(\tau) = \tilde{\xi}_j(\tau - 1\{\tilde{Y}_j(1) \leq q_1(\tau)\}) - m_{1,\tau}(\tilde{X}_j)$, $(\tilde{Y}_j(1),\tilde{X}_j)$ are as defined before, $\tilde{\xi}_j = \xi_{i_j}^p$, $i_j$ is the $j$-th smallest index in the set $\{i\in [2n]: A_i = 1 \}$, and given $\{\tilde{X}_j\}_{j=1}^n$, $\{\tilde{\eta}^p_{j,1}(\tau)\}_{j=1}^n$ is a sequence of independent random variables.\footnote{Although units in the same pair share the same bootstrap weight, the weights across pairs are still independent.} Further, denote the conditional distribution of $(\tilde{\xi}_j,\tilde{Y}_j(1))$ given $\tilde{X}_j $ as $\mathbb{P}^{(j)}$. Then,
\begin{align*}
\frac{1}{n}\sum_{j=1}^n\mathbb{P}^{(j)}(\tilde{\eta}^p_{j,1}(\tau))^2 = & \frac{1}{n}\sum_{j=1}^n\left\{ \mathbb{E}\left[(\tilde{\xi}_j)^2(\tau - 1\{\tilde{Y}_j(1) \leq q_1(\tau)\})^2|\tilde{X}_j\right] - m_{1,\tau}^2(\tilde{X}_j)\right\} \leq \overline{C}<\infty,
\end{align*}
for some constant $\overline{C}>0.$ This implies that pointwise in $\tau \in \Upsilon$,
\begin{align*}
\sum_{i=1}^{2n}\frac{A_i}{\sqrt{n}}\eta^p_{i,1}(\tau)|\{A_i,X_i\}_{i=1}^{2n} \stackrel{d}{=}\sum_{j=1}^{n}\frac{\tilde{\eta}^p_{j,1}(\tau)}{\sqrt{n}}\biggl| \{\tilde{X}_j\}_{j=1}^n = O_p(1).
\end{align*}

In addition, let
\begin{align*}
\mathcal{F}_2 = \{ \xi \left[\tau - 1\{Y \leq q_1(\tau)\}\right] -  \xi\left[\tau'  - 1\{Y \leq q_1(\tau' )\}\right] :  \tau,\tau' \in \Upsilon, |\tau - \tau' |\leq \eps\}
\end{align*}
which is a VC-class with a fixed VC-index and has an envelope $F_j = 2\tilde{\xi}_j$. In addition,  $||\max_{j \in [n]}F_j||_{\mathbb{P},2} \leq C \log(n)$ and
\begin{align*}
\sigma_n^2 = \sup_{f \in \mathcal{F}_2}\overline{\mathbb{P}}f^2 \lesssim \sup_{\tilde{\tau} \in \Upsilon}\frac{1}{n}\sum_{i=1}^n \left[\eps^2+\frac{f_1(q_1(\tilde{\tau})|\tilde{X}_j)\eps}{f_1(q_1(\tilde{\tau}))} \right] \lesssim \eps~a.s.
\end{align*}
Then, by Lemma \ref{lem:max_eq},
\begin{align*}
\mathbb{E}\left[\sup_{\tau, \tau' \in \Upsilon, |\tau-\tau' |\leq \eps}\left|\sum_{j=1}^{n}\frac{\tilde{\eta}^p_{j,1}(\tau) - \tilde{\eta}^p_{j,1}(\tau' )}{\sqrt{n}}\right| \biggl| \{\tilde{X}_j\}_{j=1}^n \right] = & \mathbb{E}\left[\Vert\mathbb{P}_n - \overline{\mathbb{P}}\Vert_{\mathcal{F}_2}\biggl| \{\tilde{X}_j\}_{j=1}^n \right] \\
\lesssim & \sqrt{\eps \log(1/\eps)} + \frac{\log(1/\eps) \log(n)}{\sqrt{n}}~a.s.
\end{align*}
The RHS of the above display vanishes as $n \rightarrow \infty$ followed by $\eps \rightarrow 0$, which implies
\begin{align}
\label{eq:Bstar}
\sum_{i=1}^{2n}\frac{A_i}{\sqrt{n}}\eta^p_{i,1}(\tau)|\{A_i,X_i\}_{i=1}^{2n} \stackrel{d}{=}\sum_{j=1}^{n}\frac{\tilde{\eta}^p_{j,1}(\tau)}{\sqrt{n}}\biggl| \{\tilde{X}_j\}_{j=1}^n
\end{align}
is stochastically equicontinuous. Therefore, $\sup_{\tau \in \Upsilon}\left|\sum_{i=1}^{2n}\frac{A_i}{\sqrt{n}}\eta^p_{i,1}(\tau)\right|= O_p(1)$, and thus, $\sup_{\tau \in \Upsilon}|W_{n,1}^p(\tau)|= O_p(1)$. Similarly, we can show $\sup_{\tau \in \Upsilon}|W_{n,0}^p(\tau)|= O_p(1)$.

\vspace{2mm}
\noindent \textbf{Step (3).} Based on the previous two steps, we have
\begin{align}
\label{eq:betam}
\sqrt{n}(\hat{\beta}^p(\tau) - \beta(\tau)) = Q^{-1}(\tau)\begin{pmatrix}
1 & 1 \\
1 & 0
\end{pmatrix} \begin{pmatrix} \sum_{i=1}^{2n}\frac{A_i}{\sqrt{n}}\eta^p_{i,1}(\tau) +  \sum_{i=1}^{2n} \frac{1}{2\sqrt{n}}m_{1,\tau}(X_i)\\
\sum_{i=1}^{2n}\frac{1-A_i}{\sqrt{n}}\eta^p_{i,0}(\tau) +  \sum_{i=1}^{2n} \frac{1}{2\sqrt{n}}m_{0,\tau}(X_i)
\end{pmatrix} + R^p(\tau)
\end{align}
where $\sup_{\tau \in \Upsilon} ||R^p(\tau)||_2 = o_p(1)$ and $\sqrt{n}(\hat{\beta}^p(\tau) - \beta(\tau))$ is stochastically equicontinuous.
Taking the difference between \eqref{eq:beta} and \eqref{eq:betam}, we have
\begin{align}
\label{eq:betam2}
\sqrt{n}(\hat{\beta}^p(\tau) - \hat{\beta}(\tau)) = Q^{-1}(\tau)\begin{pmatrix}
1 & 1 \\
1 & 0
\end{pmatrix} \begin{pmatrix} \sum_{i=1}^{2n}\frac{A_i}{\sqrt{n}}(\xi_i^p - 1)(\tau - 1\{Y_i(1) \leq q_1(\tau)\}) \\
\sum_{i=1}^{2n}\frac{1-A_i}{\sqrt{n}}(\xi_i^p - 1)(\tau - 1\{Y_i(0) \leq q_0(\tau)\})
\end{pmatrix} + R^*(\tau),
\end{align}
where $\sup_{\tau \in \Upsilon} ||R^*(\tau)||_2 = o_p(1)$. In addition, because both $\sqrt{n}(\hat{\beta}^p(\tau) - \beta(\tau))$ and $\sqrt{n}(\hat{\beta}(\tau) - \beta(\tau))$ are stochastically equicontinuous, so be  $\sqrt{n}(\hat{\beta}^p(\tau) - \hat{\beta}(\tau))$. Then by Markov's inequality, $\sqrt{n}(\hat{\beta}^p(\tau) - \hat{\beta}(\tau))$ is stochastically equicontinuous conditionally on data as well. In order to derive the limit  distribution of  $\sqrt{n}(\hat{\beta}^p(\tau) - \hat{\beta}(\tau))$ conditionally on data, we need only compute the covariance kernel. Note that
\footnotesize
\begin{align*}
& \mathbb{E}\left[\begin{pmatrix} \sum_{i=1}^{2n}\frac{A_i}{\sqrt{n}}(\xi_i^p - 1)(\tau - 1\{Y_i(1) \leq q_1(\tau)\}) \\
\sum_{i=1}^{2n}\frac{1-A_i}{\sqrt{n}}(\xi_i^p - 1)(\tau - 1\{Y_i(0) \leq q_0(\tau)\})
\end{pmatrix}\begin{pmatrix} \sum_{i=1}^{2n}\frac{A_i}{\sqrt{n}}(\xi_i^p - 1)(\tau'  - 1\{Y_i(1) \leq q_1(\tau' )\}) \\
\sum_{i=1}^{2n}\frac{1-A_i}{\sqrt{n}}(\xi_i^p - 1)(\tau'  - 1\{Y_i(0) \leq q_0(\tau' )\})
\end{pmatrix}^\top \biggl|Data\right] \\
= & \mathbb{E}\left[\begin{pmatrix} \sum_{j=1}^{n}\frac{1}{\sqrt{n}}(\xi_{(j,1)}^p - 1)(\tau - 1\{Y_{(j,1)}(1) \leq q_1(\tau)\}) \\
\sum_{j=1}^{n}\frac{1}{\sqrt{n}}(\xi_{(j,0)}^p - 1)(\tau - 1\{Y_{(j,0)}(0) \leq q_0(\tau)\})
\end{pmatrix}
\begin{pmatrix} \sum_{j=1}^{n}\frac{1}{\sqrt{n}}(\xi_{(j,1)}^p - 1)(\tau'  - 1\{Y_{(j,1)}(1) \leq q_1(\tau' )\}) \\
\sum_{j=1}^{n}\frac{1}{\sqrt{n}}(\xi_{(j,0)}^p - 1)(\tau'  - 1\{Y_{(j,1)}(0) \leq q_0(\tau' )\})
\end{pmatrix}^\top \biggl|Data\right] \\
= & \frac{1}{n}\sum_{j=1}^{n}\begin{pmatrix}
(\tau - 1\{Y_{(j,1)}(1) \leq q_1(\tau)\})(\tau'  - 1\{Y_{(j,1)}(1) \leq q_1(\tau' )\}) & (\tau - 1\{Y_{(j,1)}(1) \leq q_1(\tau)\})(\tau'  - 1\{Y_{(j,0)}(0) \leq q_0(\tau' )\}) \\
(\tau - 1\{Y_{(j,0)}(0) \leq q_0(\tau)\})(\tau'  - 1\{Y_{(j,1)}(1) \leq q_1(\tau' )\}) & (\tau - 1\{Y_{(j,0)}(0) \leq q_0(\tau)\})(\tau'  - 1\{Y_{(j,0)}(0) \leq q_0(\tau' )\})
\end{pmatrix},
\end{align*}
\normalsize
where the indexes notation $(j,1)$ and $(j,0)$ are defined in Section \ref{sec:boot}, and the second equality holds because $\xi_{(j,1)}^p = \xi_{(j,0)}^p$. Following the same argument in the proof of Theorem \ref{thm:boot},\footnote{In fact, in the proof of Theorem \ref{thm:boot}, we establish a more general result allowing the true quantiles $q_1(\tau)$ and $q_0(\tau)$ to be replaced by their estimates.} we can show that the limit of the above display is
\begin{align*}
\begin{pmatrix}
\min(\tau,\tau' ) - \tau \tau'  & \mathbb{E}m_{1,\tau}(X_i)m_{0,\tau' }(X_i)\\
\mathbb{E}m_{1,\tau' }(X_i)m_{0,\tau}(X_i) & \min(\tau,\tau' ) - \tau \tau'
\end{pmatrix}.
\end{align*}

\normalsize

Then, for the second element $\hat{\beta}_1^p(\tau)$ of $\hat{\beta}^p(\tau)$, conditional on the data, we have
\begin{align*}
\sqrt{n}(\hat{\beta}^p_1(\tau) - \hat{\beta}_1(\tau)) \convD \mathcal{B}^p(\tau),
\end{align*}
where $\mathcal{B}^p(\tau)$ is a tight Gaussian process with covariance kernel
\begin{align*}
\Sigma^p(\tau,\tau' ) = \frac{ \min(\tau,\tau' ) - \tau \tau'  }{f_1(q_1(\tau))f_1(q_1(\tau' ))} + \frac{ \min(\tau,\tau' ) - \tau \tau'  }{f_0(q_0(\tau))f_0(q_0(\tau' ))} - \frac{\mathbb{E}m_{1,\tau}(X_i)m_{0,\tau' }(X_i)}{f_1(q_1(\tau))f_0(q_0(\tau' ))} - \frac{\mathbb{E}m_{1,\tau' }(X_i)m_{0,\tau}(X_i)}{f_0(q_0(\tau))f_1(q_1(\tau' ))}.
\end{align*}

\section{Proof of Theorem \ref{thm:boot}}
\label{sec:boot_pf}
Let $u \in \Re^2$ and
\begin{align*}
L_n^*(u,\tau)  = \sum_{i=1}^{2n}\left[\rho_\tau(Y_i - \dot{A}_i^\top \beta(\tau) - \dot{A}_i^\top u/\sqrt{n}) -  \rho_\tau(Y_i - \dot{A}_i^\top \beta(\tau))\right] - u^\top  \begin{pmatrix}
1 & 1 \\
1 & 0
\end{pmatrix}S_n^*(\tau).
\end{align*}
Then,
\begin{align*}
\sqrt{n}\left(\hat{\beta}^*(\tau) - \beta(\tau) \right) = \argmin_u L_n^*(u,\tau).
\end{align*}
By the same argument as in the proof of Theorem \ref{thm:est}, we have
\begin{align*}
L_n^*(u,\tau) = - u^\top  W_n(\tau) + Q_n(u,\tau) - u^\top  \begin{pmatrix}
1 & 1 \\
1 & 0
\end{pmatrix} S_n^*(\tau) = -u^\top \begin{pmatrix}
1 & 1 \\
1 & 0
\end{pmatrix}\left(S_n(\tau) + S_n^*(\tau)\right) +  Q_n(u,\tau).
\end{align*}
Further note that $S_n^*(\tau) = \frac{1}{\sqrt{2}}\left(S_{n,1}^*(\tau)+S_{n,2}^*(\tau) \right)$. In the following, we divide the proof into three steps. In Step (1), we derive the weak limit of $S_{n,1}^*(\tau)$ given data. In Step (2), we derive the weak limit of $S_{n,2}^*(\tau)$. In Step (3), we derive the desired result of this theorem.

\vspace{2mm}
\noindent \textbf{Step (1).} Given the data, $S_{n,1}^*(\tau)$ is a Gaussian process with covariance kernel
\begin{align*}
\tilde{\Sigma}_1^*(\tau,\tau' ) = \begin{pmatrix}
\tilde{\Sigma}^*_{1,1,1}(\tau,\tau' )  & \tilde{\Sigma}^*_{1,1,2}(\tau,\tau' ) \\
\tilde{\Sigma}^*_{1,2,1}(\tau,\tau' )  & \tilde{\Sigma}^*_{1,2,2}(\tau,\tau' ) \\
\end{pmatrix}
\end{align*}
where
\begin{align*}
\tilde{\Sigma}^*_{1,1,1}(\tau,\tau' )  = \frac{1}{n}\sum_{j =1}^{n}\left(\tau - 1\{Y_{(j,1)} \leq \hat{q}_1(\tau)\}\right)\left(\tau'  - 1\{Y_{(j,1)} \leq \hat{q}_1(\tau' )\}\right),
\end{align*}
\begin{align*}
\tilde{\Sigma}^*_{1,1,2}(\tau,\tau' )  = \frac{1}{n}\sum_{j =1}^{n}\left(\tau - 1\{Y_{(j,1)} \leq \hat{q}_1(\tau)\}\right)\left(\tau'  - 1\{Y_{(j,0)} \leq \hat{q}_0(\tau' )\}\right),
\end{align*}
\begin{align*}
\tilde{\Sigma}^*_{1,2,1}(\tau,\tau' )  = \frac{1}{n}\sum_{j =1}^{n}\left(\tau - 1\{Y_{(j,0)} \leq \hat{q}_0(\tau)\}\right)\left(\tau'  - 1\{Y_{(j,1)} \leq \hat{q}_1(\tau' )\}\right),
\end{align*}
and
\begin{align*}
\tilde{\Sigma}^*_{1,2,2}(\tau,\tau' )  = \frac{1}{n}\sum_{j =1}^{n}\left(\tau - 1\{Y_{(j,0)} \leq \hat{q}_0(\tau)\}\right)\left(\tau'  - 1\{Y_{(j,0)} \leq \hat{q}_0(\tau' )\}\right).
\end{align*}
Next, we derive the limit of $\tilde{\Sigma}_1^*(\tau,\tau' )$ uniformly over $\tau,\tau'  \in \Upsilon$. Recall $m_{1,\tau}(X_i,q) = \mathbb{E}\left(\tau - 1\{Y_i(1) \leq q\} |X_i\right)$ and define $\eta_{1,i}(q,\tau) = (\tau - 1\{Y_i(1) \leq q\}) - m_{1,\tau}(X_i,q)$. Then
\begin{align}
\label{eq:sigma111}
\tilde{\Sigma}_{1,1,1}(\tau,\tau' ) = & \frac{1}{n}\sum_{j =1}^n\eta_{1,(j,1)}(\hat{q}_1(\tau),\tau)\eta_{1,(j,1)}(\hat{q}_1(\tau' ),\tau' ) +  \frac{1}{n}\sum_{j =1}^n\eta_{1,(j,1)}(\hat{q}_1(\tau),\tau)m_{1,\tau' }(X_{(j,1)},\hat{q}_1(\tau' )) \notag \\
& + \frac{1}{n}\sum_{j =1}^n\eta_{1,(j,1)}(\hat{q}_1(\tau' ),\tau' )m_{1,\tau}(X_{(j,1)},\hat{q}_1(\tau)) + \frac{1}{n}\sum_{j =1}^n m_{1,\tau}(X_{(j,1)},\hat{q}_1(\tau))m_{1,\tau' }(X_{(j,1)},\hat{q}_1(\tau' )) \notag \\
= & I(\tau,\tau' )+II(\tau,\tau' ) +III(\tau,\tau' )+IV(\tau,\tau' ),
\end{align}
where we use the fact that $Y_{(j,1)} = Y_{(j,1)}(1)$ and $Y_{(j,0)} = Y_{(j,0)}(0)$. Given $\{A_i,X_i\}_{i=1}^{2n}$, $\{Y_{(j,1)}(1)\}_{j=1}^n$ is a sequence of independent random variables with probability measure $\Pi_{j=1}^n\mathbb{P}^{(j)}$, where $\mathbb{P}^{(j)}$ is the conditional probability of $Y(1)$ given $X$ evaluated at $X = X_{(j,1)}$. Therefore,
\begin{align}
\label{eq:I}
I(\tau,\tau' ) = \overline{\mathbb{P}}\eta_{1,(j,1)}(\hat{q}_1(\tau),\tau)\eta_{1,(j,1)}(\hat{q}_1(\tau' ),\tau' ) + \left(\mathbb{P}_n - \overline{\mathbb{P}}\right)\eta_{1,(j,1)}(\hat{q}_1(\tau),\tau)\eta_{1,(j,1)}(\hat{q}_1(\tau' ),\tau' ),
\end{align}
where $\overline{\mathbb{P}}\eta_{1,(j,1)}(\hat{q}_1(\tau),\tau)\eta_{1,(j,1)}(\hat{q}_1(\tau' ),\tau' )$ is interpreted as $\overline{\mathbb{P}}\eta_{1,(j,1)}(q,\tau)\eta_{1,(j,1)}(q',\tau' )|_{q = \hat{q}_1(\tau),q'=\hat{q}(\tau' )}$. In addition, by Theorem \ref{thm:est}, for any $\eps>0$, it is possible to find a sufficiently large constant $L$ such that
\begin{align}
\label{eq:qhat}
\mathbb{P}(\sup_{\tau \in \Upsilon}|\hat{q}(\tau) - q(\tau)|\leq L/\sqrt{n}) \geq 1-\eps.
\end{align}
Therefore, we have,
\begin{align}
\label{eq:I1}
& \overline{\mathbb{P}}\eta_{1,(j,1)}(\hat{q}_1(\tau),\tau)\eta_{1,(j,1)}(\hat{q}_1(\tau' ),\tau' ) \notag \\
= & \frac{1}{n}\sum_{j=1}^n\left[ F_1(\min(\hat{q}_1(\tau),\hat{q}_1(\tau' ))|X_{(j,1)}) - F_1(\hat{q}_1(\tau)|X_{(j,1)})F_1(\hat{q}_1(\tau' )|X_{(j,1)})\right] \notag \\
= & \frac{1}{n}\sum_{j=1}^n\left[ F_1(\min(q_1(\tau),q_1(\tau' ))|X_{(j,1)}) - F_1(q_1(\tau)|X_{(j,1)})F_1(q_1(\tau' )|X_{(j,1)})\right] + R_I(\tau,\tau' ) \notag \\
= & \frac{1}{n}\sum_{i=1}^{2n}A_i\left[ F_1(\min(q_1(\tau),q_1(\tau' ))|X_i) - F_1(q_1(\tau)|X_i)F_1(q_1(\tau' )|X_i)\right] +R_I(\tau,\tau' ) \notag \\
= & \frac{1}{2n}\sum_{i=1}^{2n}\left[ F_1(\min(q_1(\tau),q_1(\tau' ))|X_i) - F_1(q_1(\tau)|X_i)F_1(q_1(\tau' )|X_i)\right] \notag \\
+  & \frac{1}{n}\sum_{i=1}^{2n}\left(A_i - \frac{1}{2}\right)\left[ F_1(\min(q_1(\tau),q_1(\tau' ))|X_i) - F_1(q_1(\tau)|X_i)F_1(q_1(\tau' )|X_i)\right] +  R_I(\tau,\tau' ),
\end{align}
where $\sup_{\tau,\tau'  \in \Upsilon}|R_I(\tau,\tau' )| \convP 0$ due to \eqref{eq:qhat} and Lipschitz continuity of $F_1(\cdot|X)$.

By the standard uniform convergence theorem (\citet[Theorem 2.4.1]{VW96}), uniformly over $\tau, \tau' \in \Upsilon$,
\begin{align*}
\frac{1}{2n}\sum_{i=1}^{2n}\left[ F_1(\min(q_1(\tau),q_1(\tau' ))|X_i) - F_1(q_1(\tau)|X_i)F_1(q_1(\tau' )|X_i)\right] \convP \min(\tau,\tau' ) - \tau\tau'  - \mathbb{E}m_{1,\tau}(X)m_{1,\tau' }(X).
\end{align*}
By the same argument as in Lemma \ref{lem:R12},
\begin{align*}
\sup_{\tau, \tau' \in \Upsilon}\left| \frac{1}{n}\sum_{i=1}^{2n}\left(A_i - \frac{1}{2}\right)\left[ F_1(\min(q_1(\tau),q_1(\tau' ))|X_i) - F_1(q_1(\tau)|X_i)F_1(q_1(\tau' )|X_i)\right]\right| \convP 0.
\end{align*}
Therefore, uniformly over $\tau,\tau'  \in \Upsilon$,
\begin{align*}
\overline{\mathbb{P}}\eta_{1,(j,1)}(\hat{q}_1(\tau),\tau)\eta_{1,(j,1)}(\hat{q}_1(\tau' ),\tau' ) \convP \min(\tau,\tau' ) - \tau\tau'  - \mathbb{E}m_{1,\tau}(X)m_{1,\tau' }(X).
\end{align*}

To deal with the second term in \eqref{eq:I}, first denote
\begin{align*}
\mathcal{F}_3 = \{\left(\tau - 1\{Y \leq q_1(\tau)+v\}\right)\left(\tau'  - 1\{Y \leq q_1(\tau' ) +v'\}\right): \tau,\tau'  \in \Upsilon, |v|,|v'| \leq L/\sqrt{n}\}.
\end{align*}
Note  $\mathcal{F}_3$ has an envelope $F=1$ and is nested by a VC-class of functions with a fixed VC-index. Then, by Lemma \ref{lem:max_eq},
\begin{align*}
\mathbb{E}\Vert\mathbb{P}_n - \overline{\mathbb{P}} \Vert_{\mathcal{F}_3} \lesssim 1/\sqrt{n}.
\end{align*}
This implies, with probability greater than $1 -\eps$, that
\begin{align}
\label{eq:I2}
\sup_{\tau, \tau' \in \Upsilon}|\left(\mathbb{P}_n - \overline{\mathbb{P}}\right)\eta_{1,(j,1)}(\hat{q}_1(\tau),\tau)\eta_{1,(j,1)}(\hat{q}_1(\tau' ),\tau' )| \convP 0.
\end{align}
Since $\eps$ in \eqref{eq:qhat} is arbitrary, we have, uniformly over $\tau, \tau' \in \Upsilon$,
\begin{align}
\label{eq:I0}
I(\tau,\tau' ) \convP \min(\tau,\tau' ) - \tau\tau'  - \mathbb{E}m_{1,\tau}(X)m_{1,\tau' }(X).
\end{align}
By Lemma \ref{lem:23}, we have
\begin{align*}
\sup_{\tau, \tau' \in \Upsilon}|II(\tau,\tau' )| = o_p(1) \quad \text{and} \quad \sup_{\tau, \tau' \in \Upsilon}|III(\tau,\tau' )| = o_p(1).
\end{align*}

For $IV(\tau,\tau' )$, we note that
\begin{align}
\label{eq:IV}
IV(\tau,\tau' ) = & \frac{1}{n}\sum_{j =1}^n m_{1,\tau}(X_{(j,1)})m_{1,\tau' }(X_{(j,1)}) + R_{IV}(\tau,\tau' ) \notag \\
= & \frac{1}{n}\sum_{i =1}^{2n} A_i m_{1,\tau}(X_i)m_{1,\tau' }(X_i) + R_{IV}(\tau,\tau' ) \notag \\
= & \frac{1}{2n}\sum_{i =1}^{2n} m_{1,\tau}(X_i)m_{1,\tau' }(X_i) + \frac{1}{n}\sum_{i =1}^{2n} \left(A_i-\frac{1}{2}\right) m_{1,\tau}(X_i)m_{1,\tau' }(X_i) + R_{IV}(\tau,\tau' ).
\end{align}
By the standard uniform convergence theorem (\citet[Theorem 2.4.1]{VW96}), uniformly over $\tau, \tau' \in \Upsilon$,
\begin{align*}
\frac{1}{2n}\sum_{i =1}^{2n} m_{1,\tau}(X_i)m_{1,\tau' }(X_i) \convP \mathbb{E} m_{1,\tau}(X)m_{1,\tau' }(X).
\end{align*}
Lemma \ref{lem:4} further shows that
\begin{align*}
\sup_{\tau, \tau' \in \Upsilon}|R_{IV}(\tau,\tau' )| = o_p(1)   \quad \text{and} \quad \sup_{\tau, \tau' \in \Upsilon}\left|\frac{1}{n}\sum_{i =1}^{2n} \left(A_i-\frac{1}{2}\right) m_{1,\tau}(X_i)m_{1,\tau' }(X_i)\right| =o_p(1).
\end{align*}
Combining the above results, we have, uniformly over $\tau, \tau' \in \Upsilon$,
\begin{align*}
\tilde{\Sigma}^*_{1,1,1}(\tau,\tau' ) \convP \min(\tau,\tau' ) - \tau\tau' .
\end{align*}

Now we turn to $\tilde{\Sigma}^*_{1,1,2}(\tau,\tau' )$. Recall $m_{0,\tau}(X_i,q) = \mathbb{E}\left(\tau - 1\{Y_i(0) \leq q\} |X_i\right)$ and define $\eta_{0,i}(q,\tau) = (\tau - 1\{Y_i(0) \leq q\}) - m_{0,\tau}(X_i,q)$. Then,
\begin{align*}
\tilde{\Sigma}^*_{1,1,2}(\tau,\tau' ) = & \frac{1}{n}\sum_{j =1}^n\eta_{1,(j,1)}(\hat{q}_1(\tau),\tau)\eta_{0,(j,0)}(\hat{q}_0(\tau' ),\tau' ) +  \frac{1}{n}\sum_{j =1}^n\eta_{1,(j,1)}(\hat{q}_1(\tau),\tau)m_{0,\tau' }(X_{(j,0)},\hat{q}_0(\tau' )) \\
& + \frac{1}{n}\sum_{j =1}^n\eta_{0,(j,0)}(\hat{q}_0(\tau' ),\tau' )m_{1,\tau}(X_{(j,1)},\hat{q}_1(\tau)) + \frac{1}{n}\sum_{j =1}^n m_{1,\tau}(X_{(j,1)},\hat{q}_1(\tau))m_{0,\tau' }(X_{(j,0)},\hat{q}_0(\tau' )) \\
= & \widetilde{I}(\tau,\tau' )+\widetilde{II}(\tau,\tau' ) +\widetilde{III}(\tau,\tau' )+\widetilde{IV}(\tau,\tau' ).
\end{align*}

We derive the uniform limit for each term on the RHS of the above display. First, note that
\begin{align}
\label{eq:Itilde}
\widetilde{I}(\tau,\tau' ) = \overline{\mathbb{P}}\eta_{1,(j,1)}(\hat{q}_1(\tau),\tau)\eta_{0,(j,0)}(\hat{q}_0(\tau' ),\tau' ) + (\mathbb{P}_n - \overline{\mathbb{P}})\eta_{1,(j,1)}(\hat{q}_1(\tau),\tau)\eta_{0,(j,0)}(\hat{q}_0(\tau' ),\tau' ).
\end{align}
Similar to \eqref{eq:I1}, we have
\begin{align*}
\sup_{\tau, \tau' \in \Upsilon}\left| \overline{\mathbb{P}}\eta_{1,(j,1)}(\hat{q}_1(\tau),\tau)\eta_{0,(j,0)}(\hat{q}_0(\tau' ),\tau' ) - \overline{\mathbb{P}}\eta_{1,(j,1)}(q_1(\tau),\tau)\eta_{0,(j,0)}(q_0(\tau' ),\tau' )\right| \convP 0.
\end{align*}
Furthermore, because $(j,1) \neq (j,0)$, conditionally on $\{A_i,X_i\}_{i=1}^{2n}$, $\eta_{1,(j,1)}(q_1(\tau),\tau) \indep \eta_{1,(j,0)}(q_0(\tau),\tau)$,
\begin{align*}
\overline{\mathbb{P}}\eta_{1,(j,1)}(q_1(\tau),\tau)\eta_{0,(j,0)}(q_0(\tau' ),\tau' ) = 0.
\end{align*}
Similar to \eqref{eq:I2}, we have
\begin{align*}
\sup_{\tau, \tau' \in \Upsilon}\left|(\mathbb{P}_n - \overline{\mathbb{P}})\eta_{1,(j,1)}(\hat{q}_1(\tau),\tau)\eta_{0,(j,0)}(\hat{q}_0(\tau' ),\tau' )\right| \convP 0.
\end{align*}
This implies that, uniformly over $\tau,\tau' \in \Upsilon$,
$\; \widetilde{I}(\tau,\tau' ) \convP 0.$
By the same argument as in the proof of Lemma \ref{lem:23}, we can show that
\begin{align*}
\sup_{\tau, \tau' \in \Upsilon}\left|\widetilde{II}(\tau,\tau' )\right| \convP 0 \quad \text{and} \quad \sup_{\tau, \tau' \in \Upsilon}\left|\widetilde{III}(\tau,\tau' )\right| \convP 0.
\end{align*}

Last, by the same argument in the proof of Lemma \ref{lem:4}, we can show that, uniformly over $\tau,\tau' \in \Upsilon$,
\begin{align*} \widetilde{IV}(\tau,\tau' ) = &\; \frac{1}{n}\sum_{j=1}^n m_{1,\tau}(X_{(j,1)})m_{0,\tau' }(X_{(j,0)}) + o_p(1) \\
= &\; \frac{1}{n}\sum_{j=1}^n m_{1,\tau}(X_{(j,1)})m_{0,\tau' }(X_{(j,1)}) + \frac{1}{n}\sum_{j=1}^n m_{1,\tau}(X_{(j,1)})[m_{0,\tau' }(X_{(j,0)}) - m_{0,\tau' }(X_{(j,1)})]+o_p(1) \\
\convP &\; \mathbb{E}m_{1,\tau}(X)m_{0,\tau' }(X),
\end{align*}
where the $o_p(1)$ holds uniformly over $\tau,\tau'  \in \Upsilon$, and the last line holds because $m_{1,\tau}(x)$ is bounded and $m_{0,\tau}(x)$ is Lipschitz.

Combining the above results, we have uniformly over $\tau,\tau' \in \Upsilon$,
\begin{align*}
\tilde{\Sigma}^*_{1,1,2}(\tau,\tau' ) \convP\mathbb{E}m_{1,\tau}(X)m_{0,\tau' }(X).
\end{align*}
The limits of $\tilde{\Sigma}^*_{1,2,1}$ and $\tilde{\Sigma}^*_{1,2,2}$ can be derived similarly. To sum up, we have established that, uniformly over $\tau,\tau' \in \Upsilon$,
\begin{align*}
\tilde{\Sigma}_1^*(\tau,\tau' ) \convP \begin{pmatrix}
\min(\tau,\tau' ) - \tau\tau'  & \mathbb{E}m_{1,\tau}(X_i)m_{0,\tau' }(X_i) \\
\mathbb{E}m_{0,\tau}(X_i)m_{1,\tau' }(X_i) & \min(\tau,\tau' ) - \tau\tau'
\end{pmatrix}.
\end{align*}
Lemma \ref{lem:tight1} shows $S_{n,1}^*(\tau)$ is stochastically equicontinuous and $\sup_{\tau \in \Upsilon}||S_{n,1}^*(\tau)||_2= O_p(1)$. This concludes the proof of this step.

\vspace{2mm}
\noindent \textbf{Step (2).} Given the data, $S_{n,2}^*(\tau)$ is a Gaussian process with covariance kernel
\begin{align*}
\tilde{\Sigma}_2^*(\tau,\tau' ) = \begin{pmatrix}
\tilde{\Sigma}^*_{2,1,1}(\tau,\tau' )  & \tilde{\Sigma}^*_{2,1,2}(\tau,\tau' ) \\
\tilde{\Sigma}^*_{2,2,1}(\tau,\tau' )  & \tilde{\Sigma}^*_{2,2,2}(\tau,\tau' ) \\
\end{pmatrix},
\end{align*}
where
\begin{align*}
\tilde{\Sigma}^*_{2,1,1}(\tau,\tau' )  = & \frac{1}{n}\sum_{k =1}^{\lfloor n/2 \rfloor}\left[\left(\tau - 1\{Y_{(k,1)} \leq \hat{q}_1(\tau)\} \right) - \left(\tau - 1\{Y_{(k,3)} \leq \hat{q}_1(\tau)\} \right)\right] \\
& \times \left[\left(\tau'  - 1\{Y_{(k,1)} \leq \hat{q}_1(\tau' )\} \right) - \left(\tau'  - 1\{Y_{(k,3)} \leq \hat{q}_1(\tau' )\} \right)\right],
\end{align*}
\begin{align*}
\tilde{\Sigma}^*_{2,1,2}(\tau,\tau' )  = & \frac{1}{n}\sum_{k =1}^{\lfloor n/2 \rfloor}\left[\left(\tau - 1\{Y_{(k,1)} \leq \hat{q}_1(\tau)\} \right) - \left(\tau - 1\{Y_{(k,3)} \leq \hat{q}_1(\tau)\} \right)\right] \\
& \times \left[\left(\tau'  - 1\{Y_{(k,2)} \leq \hat{q}_0(\tau' )\} \right) - \left(\tau'  - 1\{Y_{(k,4)} \leq \hat{q}_0(\tau' )\} \right)\right],
\end{align*}
\begin{align*}
\tilde{\Sigma}^*_{2,2,1}(\tau,\tau' )  = & \frac{1}{n}\sum_{k =1}^{\lfloor n/2 \rfloor}\left[\left(\tau - 1\{Y_{(k,2)} \leq \hat{q}_0(\tau)\} \right) - \left(\tau - 1\{Y_{(k,4)} \leq \hat{q}_0(\tau)\} \right)\right] \\
& \times \left[\left(\tau'  - 1\{Y_{(k,1)} \leq \hat{q}_1(\tau' )\} \right) - \left(\tau'  - 1\{Y_{(k,3)} \leq \hat{q}_1(\tau' )\} \right)\right],
\end{align*}
and
\begin{align*}
\tilde{\Sigma}^*_{2,2,2}(\tau,\tau' )  = & \frac{1}{n}\sum_{k =1}^{\lfloor n/2 \rfloor}\left[\left(\tau - 1\{Y_{(k,2)} \leq \hat{q}_0(\tau)\} \right) - \left(\tau - 1\{Y_{(k,4)} \leq \hat{q}_0(\tau)\} \right)\right] \\
& \times \left[\left(\tau'  - 1\{Y_{(k,2)} \leq \hat{q}_0(\tau' )\} \right) - \left(\tau'  - 1\{Y_{(k,4)} \leq \hat{q}_0(\tau' )\} \right)\right].
\end{align*}
In the following, we derive the limit of $\tilde{\Sigma}_2^*(\tau,\tau' )$. For $\tilde{\Sigma}^*_{2,1,1}(\tau,\tau' )$, we have
\begin{align*}
& \tilde{\Sigma}^*_{2,1,1}(\tau,\tau' ) \\
= &\; \frac{1}{n}\sum_{k =1}^{\lfloor n/2 \rfloor}\left[\eta_{1,(k,1)}(\hat{q}_1(\tau),\tau) - \eta_{1,(k,3)}(\hat{q}_1(\tau),\tau)  \right]\left[\eta_{1,(k,1)}(\hat{q}_1(\tau' ),\tau' ) - \eta_{1,(k,3)}(\hat{q}_1(\tau' ),\tau' )\right] \\
& + \frac{1}{n}\sum_{k =1}^{\lfloor n/2 \rfloor}\left[\eta_{1,(k,1)}(\hat{q}_1(\tau),\tau) - \eta_{1,(k,3)}(\hat{q}_1(\tau),\tau)  \right]\left[m_{1,\tau' }(X_{(k,1)},\hat{q}_1(\tau' )) - m_{1,\tau' }(X_{(k,3)},\hat{q}_1(\tau' ))\right] \\
&  + \frac{1}{n}\sum_{k =1}^{\lfloor n/2 \rfloor}\left[m_{1,\tau}(X_{(k,1)},\hat{q}_1(\tau)) - m_{1,\tau}(X_{(k,3)},\hat{q}_1(\tau))  \right]\left[\eta_{1,(k,1)}(\hat{q}_1(\tau' ),\tau' ) - \eta_{1,(k,3)}(\hat{q}_1(\tau' ),\tau' )\right] \\
& + \frac{1}{n}\sum_{k =1}^{\lfloor n/2 \rfloor}\left[m_{1,\tau}(X_{(k,1)},\hat{q}_1(\tau)) - m_{1,\tau}(X_{(k,3)},\hat{q}_1(\tau))  \right]\left[m_{1,\tau' }(X_{(k,1)},\hat{q}_1(\tau' )) - m_{1,\tau' }(X_{(k,3)},\hat{q}_1(\tau' ))\right] \\
\equiv&\; \widehat{I}(\tau,\tau' ) + \widehat{II}(\tau,\tau' ) + \widehat{III}(\tau,\tau' ) + \widehat{IV}(\tau,\tau' ).
\end{align*}
Also note that
\begin{align*}
& \widehat{I}(\tau,\tau' ) \\
= &\; \frac{1}{n}\sum_{k =1}^{\lfloor n/2 \rfloor} \left[\eta_{1,(k,1)}(\hat{q}_1(\tau),\tau)\eta_{1,(k,1)}(\hat{q}_1(\tau' ),\tau' )  + \eta_{1,(k,3)}(\hat{q}_1(\tau),\tau)\eta_{1,(k,3)}(\hat{q}_1(\tau' ),\tau' )\right] \\
& - \frac{1}{n}\sum_{k =1}^{\lfloor n/2 \rfloor}\eta_{1,(k,1)}(\hat{q}_1(\tau),\tau)\eta_{1,(k,3)}(\hat{q}_1(\tau' ),\tau' ) - \frac{1}{n}\sum_{k =1}^{\lfloor n/2 \rfloor}\eta_{1,(k,1)}(\hat{q}_1(\tau' ),\tau' )\eta_{1,(k,3)}(\hat{q}_1(\tau),\tau) \\
= &\; \frac{1}{n}\sum_{j =1}^n\eta_{1,(j,1)}(\hat{q}_1(\tau),\tau)\eta_{1,(j,1)}(\hat{q}_1(\tau' ),\tau' ) \\
& - \frac{1}{n}\sum_{k =1}^{\lfloor n/2 \rfloor}\eta_{1,(k,1)}(\hat{q}_1(\tau),\tau)\eta_{1,(k,3)}(\hat{q}_1(\tau' ),\tau' ) - \frac{1}{n}\sum_{k =1}^{\lfloor n/2 \rfloor}\eta_{1,(k,1)}(\hat{q}_1(\tau' ),\tau' )\eta_{1,(k,3)}(\hat{q}_1(\tau),\tau).
\end{align*}
The first term on the RHS of the above display is just $I(\tau,\tau' )$ defined in Step (1), whose limit is established in \eqref{eq:I0}. For the second and third terms, we note that $(k,1) \neq (k,3)$, which implies, given $\{X_i,A_i\}_{i=1}^{2n}$, $(\eta_{1,(k,1)}(q_1(\tau),\tau),\eta_{1,(k,1)}(q_1(\tau' ),\tau' )) \indep(\eta_{1,(k,3)}(q_1(\tau),\tau),\eta_{1,(k,3)}(q_1(\tau' ),\tau' ))$. Then, by the same argument in \eqref{eq:Itilde} and the discussion below, we have
\begin{align*}
\sup_{\tau, \tau' \in \Upsilon}\left| \frac{1}{n}\sum_{k =1}^{\lfloor n/2 \rfloor}\eta_{1,(k,1)}(\hat{q}_1(\tau),\tau)\eta_{1,(k,3)}(\hat{q}_1(\tau' ),\tau' ) \right| \convP 0,
\end{align*}
and
\begin{align*}
\sup_{\tau, \tau' \in \Upsilon}\left|\frac{1}{n}\sum_{k =1}^{\lfloor n/2 \rfloor}\eta_{1,(k,1)}(\hat{q}_1(\tau' ),\tau' )\eta_{1,(k,3)}(\hat{q}_1(\tau),\tau) \right| \convP 0.
\end{align*}
This implies that, uniformly over $\tau,\tau'  \in \Upsilon$,
\begin{align*}
\widehat{I}(\tau,\tau' ) \convP  \min(\tau,\tau' ) - \tau\tau'  - \mathbb{E}m_{1,\tau}(X)m_{1,\tau' }(X).
\end{align*}

By the same argument as in the proof of Lemma \ref{lem:23}, we have
\begin{align*}
\sup_{\tau, \tau' \in \Upsilon}\left|\widehat{II}(\tau,\tau' )  \right| \convP 0 \quad \text{and} \quad \sup_{\tau, \tau' \in \Upsilon}\left|\widehat{III}(\tau,\tau' )  \right| \convP 0.
\end{align*}
For $\widehat{IV}(\tau,\tau' )$, we note $m_{1,\tau}(x,q)$ is Lipschitz in $x$ by Assumption \ref{ass:reg}. Therefore, by Assumption \ref{ass:pair}, we have
\begin{align*}
\sup_{\tau, \tau' \in \Upsilon}\left|\widehat{IV}(\tau,\tau' )\right| \lesssim \frac{1}{n}\sum_{k =1}^{\lfloor n/2 \rfloor}||X_{(k,1)} - X_{(k,3)}||_2^2 \convP 0.
\end{align*}
Combining the above results, we show that, uniformly over $\tau,\tau'  \in \Upsilon$,
\begin{align*}
\tilde{\Sigma}^*_{2,1,1}(\tau,\tau' ) \convP \min(\tau,\tau' ) - \tau\tau'  - \mathbb{E}m_{1,\tau}(X)m_{1,\tau' }(X).
\end{align*}

For $\tilde{\Sigma}^*_{2,1,2}(\tau,\tau' )$, we have
\begin{align*}
& \tilde{\Sigma}^*_{2,1,1}(\tau,\tau' ) \\
= &\; \frac{1}{n}\sum_{k =1}^{\lfloor n/2 \rfloor}\left[\eta_{1,(k,1)}(\hat{q}_1(\tau),\tau) - \eta_{1,(k,3)}(\hat{q}_1(\tau),\tau)  \right]\left[\eta_{0,(k,2)}(\hat{q}_0(\tau' ),\tau' ) - \eta_{0,(k,4)}(\hat{q}_0(\tau' ),\tau' )\right] \\
& + \frac{1}{n}\sum_{k =1}^{\lfloor n/2 \rfloor}\left[\eta_{1,(k,1)}(\hat{q}_1(\tau),\tau) - \eta_{1,(k,3)}(\hat{q}_1(\tau),\tau)  \right]\left[m_{0,\tau' }(X_{(k,2)},\hat{q}_0(\tau' )) - m_{0,\tau' }(X_{(k,4)},\hat{q}_0(\tau' ))\right] \\
&  + \frac{1}{n}\sum_{k =1}^{\lfloor n/2 \rfloor}\left[m_{1,\tau}(X_{(k,1)},\hat{q}_1(\tau)) - m_{1,\tau}(X_{(k,3)},\hat{q}_1(\tau))  \right]\left[\eta_{0,(k,2)}(\hat{q}_0(\tau' ),\tau' ) - \eta_{0,(k,4)}(\hat{q}_0(\tau' ),\tau' )\right] \\
& + \frac{1}{n}\sum_{k =1}^{\lfloor n/2 \rfloor}\left[m_{1,\tau}(X_{(k,1)},\hat{q}_1(\tau)) - m_{1,\tau}(X_{(k,3)},\hat{q}_1(\tau))  \right]\left[m_{0,\tau' }(X_{(k,2)},\hat{q}_0(\tau' )) - m_{0,\tau' }(X_{(k,4)},\hat{q}_0(\tau' ))\right] \\
\equiv&\; \overline{I}(\tau,\tau' ) + \overline{II}(\tau,\tau' ) + \overline{III}(\tau,\tau' ) + \overline{IV}(\tau,\tau' ).
\end{align*}
Because $(k,1),\cdots,(k,4)$ are distinctive, the elements
\begin{align*}
\left(\eta_{1,(k,1)}(q,\tau), \eta_{1,(k,3)}(q,\tau), \eta_{0,(k,2)}(q',\tau), \eta_{0,(k,4)}(q',\tau)\right)
\end{align*}
are mutually independent conditionally on $\{X_i,A_i\}_{i=1}^{2n}$. Then, by the same arguments as in \eqref{eq:I1} and \eqref{eq:I2}, we have
\begin{align*}
\sup_{\tau, \tau' \in \Upsilon}|\overline{I}(\tau,\tau' )| \convP 0.
\end{align*}
By the same argument as in the proof of Lemma \ref{lem:23}, we also have
\begin{align*}
\sup_{\tau, \tau' \in \Upsilon}|\overline{II}(\tau,\tau' )| \convP 0 \quad \text{and} \quad \sup_{\tau, \tau' \in \Upsilon}|\overline{III}(\tau,\tau' )| \convP 0.
\end{align*}
Last, by Assumption \ref{ass:pair}, we have
\begin{align*}
\sup_{\tau, \tau' \in \Upsilon}|\overline{IV}(\tau,\tau' )| \lesssim & \frac{1}{n}\sum_{k =1}^{\lfloor n/2 \rfloor}||X_{(k,1)} - X_{(k,3)}||_2||X_{(k,2)} - X_{(k,4)}||_2 \\
\lesssim & \frac{1}{n}\sum_{k =1}^{\lfloor n/2 \rfloor}||X_{(k,1)} - X_{(k,3)}||_2^2 + \frac{1}{n}\sum_{k =1}^{\lfloor n/2 \rfloor}||X_{(k,2)} - X_{(k,4)}||_2^2 \convP 0.
\end{align*}
Combining the above results, we have
\begin{align*}
\sup_{\tau, \tau' \in \Upsilon}|\tilde{\Sigma}^*_{2,1,2}(\tau,\tau' )| \convP 0.
\end{align*}

We can derive the limits of $\tilde{\Sigma}^*_{2,2,1}(\tau,\tau' )$ and $\tilde{\Sigma}^*_{2,2,2}(\tau,\tau' )$ in the same manner. To sum up, uniformly over $\tau,\tau' \in \Upsilon$, we have
\begin{align*}
\tilde{\Sigma}^*_2 \convP \begin{pmatrix}
\min(\tau,\tau' ) - \tau\tau'  - \mathbb{E}m_{1,\tau}(X_i)m_{1,\tau' }(X_i) & 0 \\
0 & \min(\tau,\tau' ) - \tau\tau'  - \mathbb{E}m_{0,\tau}(X_i)m_{0,\tau' }(X_i)
\end{pmatrix}.
\end{align*}
The stochastic equicontinuity and uniform boundedness in probability of $S_{n,2}^*(\tau)$ can be established similarly to $S_{n,1}^*(\tau)$.

\vspace{2mm}
\noindent \textbf{Step (3).}
Because both $S_n(\tau)$ and $S_n^*(\tau)$ are stochastically equicontinuous and
$$\sup_{\tau \in \Upsilon}\left(||S_n(\tau)||_2 + ||S_n^*(\tau)||_2 \right)= O_p(1),$$ we can apply \citet[Theorem 2]{K09} and have
\begin{align}
\label{eq:betastar}
\sqrt{n}(\hat{\beta}^*(\tau) - \beta(\tau) )= Q^{-1}(\tau)\begin{pmatrix}
1 & 1 \\
1 & 0
\end{pmatrix}\left(S_n(\tau) + S_n^*(\tau)\right) + R^*(\tau),
\end{align}
where $\sup_{\tau \in \Upsilon} ||R^*(\tau)||_2 = o_p(1)$. Taking the difference between \eqref{eq:betastar} and \eqref{eq:beta}, we have
\begin{align*}
\sqrt{n}(\hat{\beta}^*(\tau) - \hat{\beta}(\tau) ) = Q^{-1}(\tau)\begin{pmatrix}
1 & 1 \\
1 & 0
\end{pmatrix}S_n^*(\tau) + \tilde{R}^*(\tau),
\end{align*}
where $\sup_{\tau \in \Upsilon} ||\tilde{R}^*(\tau)||_2 = o_p(1)$. In addition, given the data, $S_{n,1}^*(\tau)$ and $S_{n,2}^*(\tau)$ are independent. Steps (1) and (2) show that uniformly over $\tau \in \Upsilon$ and conditionally on data, $S_n^*(\tau) = \frac{S_{n,1}^*(\tau) + S_{n,2}^*(\tau)}{\sqrt{2}}$ converges to a Gaussian process with covariance kernel
\begin{align*}
& \frac{1}{2}\left[ \tilde{\Sigma}_1(\tau,\tau' ) +\tilde{\Sigma}_2(\tau,\tau' )\right],
\end{align*}
where $\tilde{\Sigma}_1(\tau,\tau' ) $ and $\tilde{\Sigma}_2(\tau,\tau' )$ are defined in \eqref{eq:sigmatilde1} and \eqref{eq:sigmatilde2}, respectively. The weak limit of $S_n^*(\tau)$ given data coincides with the weak limit of $S_n(\tau)$. This implies, given the data, that
\begin{align*}
\sqrt{n}(\hat{q}^*(\tau) -\hat{q}(\tau)) \convD \mathcal{B}(\tau),
\end{align*}
where $\mathcal{B}(\tau)$ is the Gaussian process defined in Theorem \ref{thm:est}. This concludes the proof.

\section{Proof of Theorem \ref{thm:ipw_w}}
\label{sec:thmipww}
We first focus on $\hat{q}_{ipw,1}^w(\tau)$. Let $u \in \Re$ and
\begin{align*}
\tilde{L}^w_n(u,\tau)  = \sum_{i=1}^{2n}\frac{\xi_iA_i}{2\hat{A}_i}\left[\rho_\tau(Y_i - q_1(\tau) - u/\sqrt{n}) - \rho_\tau(Y_i - q_1(\tau))\right].
\end{align*}
Then, by change of variables, we have
\begin{align*}
\sqrt{n}(\hat{q}_{ipw,1}^w(\tau) - q_1(\tau)) = \argmin_u \tilde{L}^w_n(u,\tau).
\end{align*}
Notice that $\tilde{L}^w_n(u,\tau)$ is convex in $u$ for each $\tau$ and bounded in $\tau$ for each $u$. In the following, we divide the proof into three steps. In Step (1), we show that there exists
$$\tilde{g}^w_n(u,\tau) = - u\widetilde{W}^w_{n,1}(\tau) + \frac{f_1(q_1(\tau))u^2}{2}$$
for some $ \widetilde{W}^w_{n,1}(\tau)$ such that for each $u$,
\begin{align*}
\sup_{\tau \in \Upsilon}|\tilde{L}^w_n(u,\tau) - \tilde{g}^w_n(u,\tau)| \convP 0.
\end{align*}
In Step (2), we show $\sup_{\tau \in \Upsilon}|\widetilde{W}^w_{n,1}(\tau)|= O_p(1)$. Then by \citet[Theorem 2]{K09}, we have
\begin{align*}
\sqrt{n}(\hat{q}_{ipw,1}^w(\tau) - q_1(\tau))= [f_1(q_1(\tau))]^{-1}\widetilde{W}^w_{n,1}(\tau) + \tilde{r}_{n,1}(\tau),
\end{align*}
where $\sup_{\tau \in \Upsilon}|\tilde{r}_{n,1}(\tau)| = o_p(1)$. For the same reason, we can show
\begin{align*}
\sqrt{n}(\hat{q}_{ipw,0}^w(\tau) - q_0(\tau))= [f_0(q_0(\tau))]^{-1}\widetilde{W}^w_{n,0}(\tau) + \tilde{r}_{n,0}(\tau),
\end{align*}
for some $\widetilde{W}^w_{n,0}(\tau)$ to be specified later and $\sup_{\tau \in \Upsilon}|\tilde{r}_{n,0}(\tau)| = o_p(1)$. Last, in Step (3), we establish the weak convergence of
\begin{align*}
\sqrt{n}(\hat{q}_{ipw}^w(\tau) - \hat{q}(\tau))
\end{align*}
conditionally on data.

\vspace{2mm}
\noindent \textbf{Step (1).} Similar to Step (1) in the previous section, we have
\begin{align*}
& \tilde{L}^w_n(u,\tau) =  -\widetilde{W}^w_{n,1}(\tau)u +\tilde{Q}^w_n(u,\tau),
\end{align*}
where
\begin{align*}
\widetilde{W}^w_{n,1}(\tau) = \sum_{i=1}^{2n}\frac{\xi_i A_i}{2\sqrt{n}\hat{A}_i}\left(\tau- 1\{Y_i(1)\leq q_1(\tau)\}\right)
\end{align*}
and
\begin{align}
\label{eq:Qnipw}
\tilde{Q}^w_n(u,\tau) = & \sum_{i=1}^{2n}\frac{\xi_iA_i}{2\hat{A}_i}\int_0^{\frac{u}{\sqrt{n}}}\left(1\{Y_i(1) -  q_1(\tau)\leq v\} - 1\{Y_i(1) -  q_1(\tau)\leq 0\} \right)dv \notag  \\
= & \sum_{i=1}^{2n}\xi_iA_i\int_0^{\frac{u}{\sqrt{n}}}\left(1\{Y_i(1) -  q_1(\tau)\leq v\} - 1\{Y_i(1) -  q_1(\tau)\leq 0\} \right)dv \notag \\
& + \sum_{i=1}^{2n}\frac{\xi_iA_i(1/2 - \hat{A}_i)}{\hat{A}_i}\int_0^{\frac{u}{\sqrt{n}}}\left(1\{Y_i(1) -  q_1(\tau)\leq v\} - 1\{Y_i(1) -  q_1(\tau)\leq 0\} \right)dv \notag \\
\equiv & \tilde{Q}^w_{n,1}(u,\tau) + \tilde{Q}^w_{n,2}(u,\tau).
\end{align}
Similar to $Q^p_{n,1}(u,\tau)$ in Section \ref{sec:pair_pf}, we have
\begin{align}
\label{eq:Qn1fipw}
\sup_{\tau \in \Upsilon}\left|\tilde{Q}^w_{n,1}(u,\tau) - \frac{f_1(q_1(\tau))u^2}{2}\right| = o_p(1).
\end{align}
For $\tilde{Q}^w_{n,2}(u,\tau)$, we have, with probability approaching one,
\begin{align}
\label{eq:q20}
|\tilde{Q}^w_{n,2}(u,\tau)| \leq & \max_{i \in [2n]}|\hat{A}_i - 1/2| \sum_{i =1}^{2n }\frac{\xi_i}{1/2 - \max_{i \in [2n]}|\hat{A}_i - 1/2|} 1\{  |Y_i(1)-q_1(\tau)|\leq u/\sqrt{n} \} \frac{|u|}{\sqrt{n}} \notag \\
\leq & \max_{i \in [2n]}|\hat{A}_i - 1/2| \sum_{i =1}^{2n} 4\xi_i 1\{  |Y_i(1)-q_1(\tau)|\leq u/\sqrt{n} \} \frac{|u|}{\sqrt{n}},
\end{align}
where the second inequality follows the fact that, w.p.a.1, $|\hat{A}_i - 1/2| \leq 1/4$ as proved in Lemma \ref{lem:sieve}. Because $\{\xi_i,Y_i(1)\}_{i \in [2n]}$ are i.i.d., by the usual maximal inequality, we can show that
\begin{align}
\label{eq:q21}
\sup_{\tau \in \Upsilon}\left|\sum_{i =1}^{2n} 4\xi_i 1\{ |Y_i(1)-q_1(\tau)|\leq u/\sqrt{n} \} \frac{|u|}{\sqrt{n}} - \mathbb{E}\sum_{i =1}^{2n} 4\xi_i 1\{  |Y_i(1)-q_1(\tau)|\leq u/\sqrt{n} \} \frac{|u|}{\sqrt{n} }\right|= o_p(1).
\end{align}
In addition,
\begin{align}
\label{eq:q2E}
\mathbb{E}\sum_{i =1}^{2n} 4\xi_i 1\{  |Y_i(1)-q_1(\tau)|\leq u/\sqrt{n} \} \frac{|u|}{\sqrt{n}} \lesssim \sqrt{n}|u|\left(F_1(q_1(\tau) + \frac{|u|}{\sqrt{n}}) - F_1(q_1(\tau) - \frac{|u|}{\sqrt{n}})\right) \lesssim u^2.
\end{align}
Combining \eqref{eq:q20}--\eqref{eq:q2E} with the fact that $ \max_{i \in [2n]}|\hat{A}_i - 1/2| = o_p(1)$ as proved in Lemma \ref{lem:sieve}, we have
\begin{align*}
\sup_{\tau \in \Upsilon}|\tilde{Q}^w_{n,2}(u,\tau)| = o_p(1).
\end{align*}
This concludes the proof of Step (1).

\vspace{2mm}
\noindent \textbf{Step (2).} We have
\begin{align}
\label{eq:w10}
\widetilde{W}^w_{n,1}(\tau) = & \sum_{i=1}^{2n}\frac{\xi_i A_i}{\sqrt{n}}\left(\tau- 1\{Y_i(1)\leq q_1(\tau)\}\right) - \sum_{i=1}^{2n}\frac{2\xi_i A_i (\hat{A}_i-1/2)}{\sqrt{n}}\left(\tau- 1\{Y_i(1)\leq q_1(\tau)\}\right) \notag \\
& + \sum_{i=1}^{2n}\frac{2\xi_i A_i (1/2 - \hat{A}_i)^2}{\sqrt{n}\hat{A}_i}\left(\tau- 1\{Y_i(1)\leq q_1(\tau)\}\right) \notag \\
\equiv & \widetilde{W}^w_{n,1,1}(\tau)- \widetilde{W}^w_{n,1,2}(\tau)+\widetilde{W}^w_{n,1,3}(\tau).
\end{align}
First, $\widetilde{W}^w_{n,1,1}(\tau)$ is uniformly bounded in probability following the exact same argument as in Step (2) of Section \ref{sec:pair_pf}. Second, we have
\begin{align*}
\widetilde{W}^w_{n,1,2}(\tau) = &\sum_{i=1}^{2n}\frac{\xi_i m_{1,\tau}(X_i) (\hat{A}_i-1/2)}{\sqrt{n}} + \sum_{i=1}^{2n}\frac{2\xi_i (A_i - 1/2)m_{1,\tau}(X_i) (\hat{A}_i-1/2)}{\sqrt{n}} + \sum_{i=1}^{2n}\frac{2\xi_iA_i \eta_{1,i}(\tau) (\hat{A}_i-1/2)}{\sqrt{n}} \\
\equiv & \; I(\tau) + II(\tau) + III(\tau).
\end{align*}
Lemma \ref{lem:sieve2} shows
\begin{align*}
\sup_{\tau \in \Upsilon}\left|I(\tau) - \sum_{i=1}^{2n}\frac{\xi_i m_{1,\tau}(X_i)(A_i - 1/2)}{\sqrt{n}}\right| = o_p(1),
\end{align*}
\begin{align*}
\sup_{\tau \in \Upsilon}\left|II(\tau) \right| = o_p(1), \quad \text{and} \quad \sup_{\tau \in \Upsilon}\left|III(\tau) \right| = o_p(1).
\end{align*}
Combining the above results, we have
\begin{align}
\label{eq:W12}
\sup_{\tau \in \Upsilon}\left|\widetilde{W}^w_{n,1,2}(\tau)-\sum_{i=1}^{2n}\frac{\xi_i m_{1,\tau}(X_i)(A_i - 1/2)}{\sqrt{n}}\right| = o_p(1).
\end{align}

Last, we have, w.p.a.1,
\begin{align}
\label{eq:w13}
\sup_{\tau \in \Upsilon}|\widetilde{W}^w_{n,1,3}(\tau)| \leq & \sum_{i=1}^{2n}\frac{2\xi_i}{\sqrt{n}(1/2 - \max_{i \in [2n]}|1/2-\hat{A}_i|)}(1/2 - \hat{A}_i)^2 \notag \\
\lesssim & \frac{4}{\sqrt{n}}\sum_{i=1}\xi_i(1/2 - \hat{A}_i)^2 = o_p(1),
\end{align}
where the first inequality holds because $\sup_{\tau \in \Upsilon}|\tau -1\{Y_i(1) \leq q_1(\tau)\}|\leq 1$, the second inequality holds because $\max_i|1/2-\hat{A}_i|\leq 1/4$ w.p.a.1 as proved in Lemma \ref{lem:sieve}, and the last inequality holds due to Lemma \ref{lem:sieve}.

Combining \eqref{eq:w10}--\eqref{eq:w13}, we have
\begin{align*}
\widetilde{W}^w_{n,1}(\tau) = \sum_{i=1}^{2n}\frac{\xi_i A_i \eta_{1,i}(\tau)}{\sqrt{n}} +\sum_{i=1}^{2n}\frac{\xi_i m_{1,\tau}(X_i)}{2\sqrt{n}} + o_p(1),
\end{align*}
where the $o_p(1)$ term holds uniformly over $\tau \in \Upsilon$. By \eqref{eq:Bstar} and the argument above, we can show $\sum_{i=1}^{2n}\frac{\xi_i A_i \eta_{1,i}(\tau)}{\sqrt{n}}$ as a stochastic process over $\tau \in \Upsilon$ is stochastically equicontinuous and $\sup_{\tau \in \Upsilon}|\sum_{i=1}^{2n}\frac{\xi_i A_i \eta_{1,i}(\tau)}{\sqrt{n}}| = O_p(1)$. Furthermore, $\{\xi_i,X_i\}_{i \in [2n]}$ is a sequence of i.i.d. random variables. Then, by the usual maximal inequality, we can show $\sum_{i=1}^{2n}\frac{\xi_i m_{1,\tau}(X_i)}{2\sqrt{n}}$ as a stochastic process over $\tau \in \Upsilon$ is stochastically equicontinuous and $\sup_{\tau \in \Upsilon}|\sum_{i=1}^{2n}\frac{\xi_i m_{1,\tau}(X_i)}{2\sqrt{n}}| = O_p(1)$. This implies, $\widetilde{W}^w_{n,1}(\tau)$ as a stochastic process over $\tau \in \Upsilon$ is stochastically equicontinuous and $\sup_{\tau \in \Upsilon}|\widetilde{W}^w_{n,1}(\tau)| = O_p(1)$, and thus, is stochastically equicontinuous conditionally on data by the Markov inequality. Therefore, we have
\begin{align*}
\sqrt{n}(\hat{q}_{ipw,1}^w(\tau) - q_1(\tau)) = \frac{1}{f_1(q_1(\tau))}\left(\sum_{i=1}^{2n}\frac{\xi_i A_i \eta_{1,i}(\tau)}{\sqrt{n}} +\sum_{i=1}^{2n}\frac{\xi_i m_{1,\tau}(X_i)}{2\sqrt{n}}\right) + \tilde{r}_{n,1}(\tau),
\end{align*}
where $\sup_{\tau \in \Upsilon}|\tilde{r}_{n,1}(\tau)| = o_p(1)$. Similarly, we can show that
\begin{align*}
\sqrt{n}(\hat{q}_{ipw,0}^w(\tau) - q_0(\tau)) = \frac{1}{f_0(q_0(\tau))}\left(\sum_{i=1}^{2n}\frac{\xi_i (1-A_i) \eta_{0,i}(\tau)}{\sqrt{n}} +\sum_{i=1}^{2n}\frac{\xi_i m_{0,\tau}(X_i)}{2\sqrt{n}}\right) + \tilde{r}_{n,0}(\tau),
\end{align*}
where $\sup_{\tau \in \Upsilon}|\tilde{r}_{n,1}(\tau)| = o_p(1)$.

\vspace{2mm}
\noindent \textbf{Step (3).} In the proof of Theorem \ref{thm:est}, we established that
\begin{align*}
& \sqrt{n}(\hat{q}(\tau) - q(\tau)) \\
= & \frac{1}{f_1(q_1(\tau))}\left(\sum_{i=1}^{2n}\frac{ A_i \eta_{1,i}(\tau)}{\sqrt{n}} +\sum_{i=1}^{2n}\frac{ m_{1,\tau}(X_i)}{2\sqrt{n}}\right) \\
&-  \frac{1}{f_0(q_0(\tau))}\left(\sum_{i=1}^{2n}\frac{ (1-A_i) \eta_{0,i}(\tau)}{\sqrt{n}} +\sum_{i=1}^{2n}\frac{ m_{0,\tau}(X_i)}{2\sqrt{n}}\right) + r_b(\tau),
\end{align*}
where $\sup_{\tau \in \Upsilon}|r_{b}(\tau)| = o_p(1)$. Then, we have
\begin{align*}
\sqrt{n}(\hat{q}_{ipw}^w(\tau) -\hat{q}(\tau)) = & \frac{1}{f_1(q_1(\tau))}\left(\sum_{i=1}^{2n}\frac{(\xi_i-1) A_i \eta_{1,i}(\tau)}{\sqrt{n}}\right) -\frac{1}{f_0(q_0(\tau))}\left(\sum_{i=1}^{2n}\frac{(\xi_i-1) (1-A_i) \eta_{0,i}(\tau)}{\sqrt{n}}\right)\\
& +\sum_{i=1}^{2n}\frac{(\xi_i-1) }{2\sqrt{n}} \left(\frac{m_{1,\tau}(X_i)}{f_1(q_1(\tau))}-\frac{m_{0,\tau}(X_i)}{f_0(q_0(\tau))}\right) + \tilde{r}_b(r),
\end{align*}
where $\sup_{\tau \in \Upsilon}|\tilde{r}_{b}(\tau)| = o_p(1)$. The conditional stochastic equicontinuity of the first three terms on the RHS of the above display has been established in Step (2). The finite-dimensional convergence of $\{\sqrt{n}(\hat{q}_{ipw}^w(\tau_l) -\hat{q}(\tau_l))\}_{l=1}^L$ given data can be shown using the Cram\'{e}r-Wold device and the Lyapunov CLT. Next, we determine the covariance kernel of  $ \sqrt{n}(\hat{q}_{ipw}^w(\tau) -\hat{q}(\tau))$ given data. Specifically, the covariance kernel is the limit of the display below:
\begin{align}
\label{eq:cov}
& \frac{1}{f_1(q_1(\tau))f_1(q_1(\tau' ))}\sum_{i=1}^{2n}\frac{A_i \eta_{1,i}(\tau)\eta_{1,i}(\tau' )}{n} + \frac{1}{f_0(q_0(\tau))f_0(q_0(\tau' ))}\sum_{i=1}^{2n}\frac{(1-A_i) \eta_{0,i}(\tau)\eta_{0,i}(\tau' )}{n} \notag \\
& + \sum_{i=1}^{2n}\frac{1}{4n} \left(\frac{m_{1,\tau}(X_i)}{f_1(q_1(\tau))}-\frac{m_{0,\tau}(X_i)}{f_0(q_0(\tau))}\right)\left(\frac{m_{1,\tau' }(X_i)}{f_1(q_1(\tau' ))}-\frac{m_{0,\tau' }(X_i)}{f_0(q_0(\tau' ))}\right) \notag \\
& + \frac{1}{2n}\sum_{i=1}^{2n} \frac{(1-A_i)\eta_{0,i}(\tau)}{f_0(q_0(\tau))}\left(\frac{m_{1,\tau' }(X_i)}{f_1(q_1(\tau' ))}-\frac{m_{0,\tau' }(X_i)}{f_0(q_0(\tau' ))}\right) + \frac{1}{2n}\sum_{i=1}^{2n} \frac{A_i\eta_{1,i}(\tau)}{f_1(q_1(\tau))}\left(\frac{m_{1,\tau' }(X_i)}{f_1(q_1(\tau' ))}-\frac{m_{0,\tau' }(X_i)}{f_0(q_0(\tau' ))}\right) \notag \\
& + \frac{1}{2n}\sum_{i=1}^{2n} \frac{(1-A_i)\eta_{0,i}(\tau' )}{f_0(q_0(\tau' ))}\left(\frac{m_{1,\tau}(X_i)}{f_1(q_1(\tau))}-\frac{m_{0,\tau}(X_i)}{f_0(q_0(\tau))}\right) + \frac{1}{2n}\sum_{i=1}^{2n} \frac{A_i\eta_{1,i}(\tau' )}{f_1(q_1(\tau' ))}\left(\frac{m_{1,\tau}(X_i)}{f_1(q_1(\tau))}-\frac{m_{0,\tau}(X_i)}{f_0(q_0(\tau))}\right).
\end{align}
Note that \eqref{eq:etaeta} implies
\begin{align*}
\frac{1}{f_1(q_1(\tau))f_1(q_1(\tau' ))}\sum_{i=1}^{2n}\frac{A_i \eta_{1,i}(\tau)\eta_{1,i}(\tau' )}{n} \convP & \ \frac{\min(\tau, \tau' ) - \mathbb{E}F_1(q_1(\tau)|X_i)F_1(q_1(\tau' )|X_i)}{f_1(q_1(\tau))f_1(q_1(\tau' ))} \\
= & \ \frac{\min(\tau, \tau' ) -  \tau\tau'  -\mathbb{E}m_{1,\tau}(X_i)m_{1,\tau' }(X_i)}{f_1(q_1(\tau))f_1(q_1(\tau' ))}.
\end{align*}
Similarly,
\begin{align*}
\frac{1}{f_0(q_0(\tau))f_0(q_0(\tau' ))}\sum_{i=1}^{2n}\frac{(1-A_i) \eta_{0,i}(\tau)\eta_{0,i}(\tau' )}{n} \convP \ \frac{\min(\tau, \tau' ) -  \tau\tau'  -\mathbb{E}m_{0,\tau}(X_i)m_{0,\tau' }(X_i)}{f_0(q_0(\tau))f_0(q_0(\tau' ))}.
\end{align*}
By the law of large numbers,
\begin{align*}
& \sum_{i=1}^{2n}\frac{1}{4n} \left(\frac{m_{1,\tau}(X_i)}{f_1(q_1(\tau))}-\frac{m_{0,\tau}(X_i)}{f_0(q_0(\tau))}\right)\left(\frac{m_{1,\tau' }(X_i)}{f_1(q_1(\tau' ))}-\frac{m_{0,\tau' }(X_i)}{f_0(q_0(\tau' ))}\right) \\
\convP & \ \frac{1}{2}\mathbb{E}\left(\frac{m_{1,\tau}(X_i)}{f_1(q_1(\tau))}-\frac{m_{0,\tau}(X_i)}{f_0(q_0(\tau))}\right)\left(\frac{m_{1,\tau' }(X_i)}{f_1(q_1(\tau' ))}-\frac{m_{0,\tau' }(X_i)}{f_0(q_0(\tau' ))}\right).
\end{align*}
Last, by Lemma \ref{lem:w1}, the last four terms on the RHS of \eqref{eq:cov} will vanish. Hence,
\begin{align*}
\text{\eqref{eq:cov}} \convP \Sigma(\tau,\tau' ),
\end{align*}
where $\Sigma(\tau,\tau' )$ is defined in Theorem \ref{thm:est}. This concludes the proof.

\section{Technical Lemmas}
\label{sec:lem}
\subsection{A Maximal Inequality with i.n.i.d. Random Variables}
\label{sec:max}
Although \cite{CCK14} derived their Corollary 5.1 for i.i.d. data, the result is still valid when the data are independent but not identically distributed (i.n.i.d.). In this section, we restate their corollary for i.n.i.d. data and provide a brief justification. The proof is due to \cite{CCK14}. We include this section purely for clarification. Let $\{W_i\}_{i=1}^n$ be a sequence of i.n.i.d. random variables taking values in a measurable space $(S,\mathcal{S})$ with distributions $\Pi_{i=1}^n\mathbb{P}^{(i)}$. Let $\mathcal{F}$ be a generic class of measurable functions $S \mapsto \Re$ with envelope $F$. Further denote $\overline{\mathbb{P}}f = \frac{1}{n}\sum_{i=1}^n\mathbb{P}^{(i)} f$, $||f||_{\overline{\mathbb{P}},2} = \sqrt{\overline{\mathbb{P}}f^2}$ and $\mathbb{P}_nf$ is the usual empirical process $\mathbb{P}_nf = \frac{1}{n}\sum_{i=1}^n f(W_i)$, $\sigma^2 = \sup_{f \in \mathcal{F}}\overline{\mathbb{P}}f^2 \leq \overline{\mathbb{P}}F^2$, and $M = \max_{i \in [n]}F(W_i)$.

\begin{lem}
	\label{lem:max_eq}
	Suppose $\overline{\mathbb{P}}F^2<\infty$ and there exist constants $a \geq e$ and $v \geq 1$ such that
	\begin{align}
	\sup_Q N(\mathcal{F},e_Q,\eps||F||_{Q,2}) \leq \left(\frac{a}{\eps}\right)^v, \quad \forall \eps \in (0,1],
	\label{eq:VC}
	\end{align}
	where $e_Q(f,g) = ||f-g||_{Q,2}$ and the supremum is taken over all finitely discrete probability measures on $(S,\mathcal{S})$. Then,
	\begin{align*}
	\mathbb{E}||\sqrt{n}(\mathbb{P}_n - \overline{\mathbb{P}})||_{\mathcal{F}} \lesssim \sqrt{v \sigma^2 \log\left(\frac{a ||F||_{\overline{\mathbb{P}},2}}{\sigma}\right)} + \frac{v||M||_2}{\sqrt{n}}\log\left(\frac{a ||F||_{\overline{\mathbb{P}},2}}{\sigma}\right).
	\end{align*}
\end{lem}
The proof of Lemma \ref{lem:max_eq} is exactly the same as that for \citet[Corollary 5.1]{CCK14} with $\mathbb{P}$ replaced by $\overline{\mathbb{P}}$. For brevity, we just highlight some key steps below.
\begin{proof}
	Let $\{\eps_i\}_{i=1}^n$ be a sequence of Rademacher random variables that is independent of $\{W_i\}_{i=1}^n$, $\sigma_n^2 = \sup_{f \in \mathcal{F}}\mathbb{P}_n f^2$, and $Z= \mathbb{E}\left[\left\Vert \frac{1}{\sqrt{n}}\sum_{i=1}^n \eps_i f(W_i) \right\Vert_{\mathcal{F}} \right]$. Then, by \citet[Lemma 2.3.1]{VW96} or \citet[Lemma 6.3]{LT13},
	$$\mathbb{E}||\sqrt{n}(\mathbb{P}_n - \overline{\mathbb{P}})||_{\mathcal{F}}\leq 2Z.$$
	Note  \citet[Lemma 6.3]{LT13} only requires $\{W_i\}_{i=1}^n$ to be independent. In addition, let the uniform entropy integral be
	\begin{align}
	\label{eq:J}
	J(\delta) \equiv J(\delta,\mathcal{F},F) = \int_{0}^{\delta}\sup_Q\sqrt{1+\log N(\mathcal{F},e_Q,\eps ||F||_{Q,2})}d\eps
	\end{align}
	where $e_Q(f,g) = ||f-g||_{Q,2}$ and the supremum is taken over all finitely discrete probability measures on $(S,\mathcal{S})$. Then, we have
	\begin{equation}
	\begin{aligned}
	Z = & \mathbb{E} \mathbb{E}\biggl[\biggl|\biggl|\frac{1}{\sqrt{n}}\sum_{i=1}^n\eps_if(W_i) \biggr|\biggr|_\mathcal{F} |W_1,\cdots,W_n\biggr] \\
	\lesssim & \mathbb{E} \biggl[||F||_{\mathbb{P}_n,2}J(\sigma_n/||F||_{\mathbb{P}_n,2}) \biggr]  \\
	\lesssim &  ||F||_{\overline{\mathbb{P}},2} J(\sqrt{\mathbb{E}\sigma_n^2}/||F||_{\overline{\mathbb{P}},2}),
	\label{eq:Z}
	\end{aligned}
	\end{equation}
	where the second inequality is due to the Jensen's inequality and the fact that $J(\sqrt{x/y})\sqrt{y}$ is concave in $(x,y)$ as shown by \cite{CCK14}. To see the first inequality, note that by the Hoeffding inequality,
	\begin{align*}
	\mathbb{P}\left(\left|\frac{1}{\sqrt{n}}\sum_{i=1}^n\eps_if(W_i)\right| \geq t \biggl| \{W_i\}_{i=1}^n\right) \lesssim \exp\left( - \frac{t^2/2}{\frac{1}{n}\sum_{i =1}^nf(W_i)^2}\right),
	\end{align*}
	which implies the stochastic process $\frac{1}{\sqrt{n}}\sum_{i=1}^n\eps_if(W_i)$ indexed by $f$ is sub-Gaussian conditionally on $\{W_i\}_{i=1}^n$. Then, the first inequality in \eqref{eq:Z} follows \citet[Corollary 2.2.8]{VW96}, where we let $\delta = \sigma_n/||F||_{\mathbb{P}_n,2}$ and $\sigma_n$ can be viewed as the diameter of the class of functions $\mathcal{F}$. We also note that this is a conditional argument, which is still valid even when $\{W_i\}_{i=1}^n$ is i.n.i.d.
	
	Next, we aim to bound $\mathbb{E}\sigma_n^2$. Recall $\sigma^2 = \sup_{f \in \mathcal{F}}\overline{\mathbb{P}}f^2$. We have, for i.n.i.d. $\{W_i\}_{i=1}^n$,
	\begin{equation}
	\begin{aligned}
	\label{eq:sigma}
	\mathbb{E}\sigma_n^2 \leq & \; \sigma^2 + \mathbb{E}( ||(\mathbb{P}_n - \overline{\mathbb{P}})f^2||_\mathcal{F})\\
	\leq & \;\sigma^2 + 2\mathbb{E}\left[\left\Vert \frac{1}{n}\sum_{i=1}^n \eps_i f^2(W_i)\right\Vert_\mathcal{F}\right]   \\
	\leq &\; \sigma^2 + 8\mathbb{E}\left[M\left\Vert \frac{1}{n}\sum_{i=1}^n \eps_i f(W_i)\right\Vert_\mathcal{F}\right]   \\
	\leq &\; \sigma^2 + 8||M||_{\mathbb{P},2}\{\mathbb{E}[||\mathbb{P}_n \eps_i f(W_i)||^2_\mathcal{F}]\}^{1/2}  \\
	\leq &\; \sigma^2 + C||M||_{\mathbb{P},2}\{\mathbb{E}[||\mathbb{P}_n \eps_i f(W_i)||_\mathcal{F}] + n^{-1}||M||_{\mathbb{P},2}\}  \\
	= & \;\sigma^2 + Cn^{-1/2}||M||_{\mathbb{P},2}Z + Cn^{-1}||M||_{\mathbb{P},2}^2,
	\end{aligned}
	\end{equation}
	where the first inequality is due to the triangle inequality, the second inequality is due to \citet[Lemma 6.3]{LT13}, the third inequality is due to \citet[Theorem 4.12]{LT13}, the fourth inequality is due to the Cauchy-Schwarz inequality, and the fifth inequality is due to \citet[Lemma 6.8]{LT13} with $q =2$.
	
	Given \eqref{eq:sigma}, \cite{CCK14} then proved the results that, for $\delta = \sigma/||F||_{\overline{\mathbb{P}},2}$,
	\begin{align}
	\label{eq:thm51}
	\mathbb{E}[\sqrt{n}||\mathbb{P}_n - \overline{\mathbb{P}}||_{\mathcal{F}}] \lesssim J(\delta,\mathcal{F},F)||F||_{\overline{\mathbb{P}},2} + \frac{||M||_{\mathbb{P},2}J^2(\delta,\mathcal{F},F)}{\delta^2 \sqrt{n}}.
	\end{align}
	In this step, they relied on the facts that $J(\delta) = J(\delta,\mathcal{F},F)$ is concave in $\delta$ and $\delta \mapsto J(\delta)/\delta$ is nonincreasing. The desired result is a quick corollary of \eqref{eq:thm51} by noticing that, under \eqref{eq:VC},
	\begin{align}
	\label{eq:J2}
	J(\delta) \leq \int_0^\delta \sqrt{1+\nu \log\left(\frac{a}{\eps}\right)}d\eps \leq 2\sqrt{2\nu}\delta\sqrt{\log\left(\frac{a}{\delta}\right)}.
	\end{align}
\end{proof}

%
%
%

\subsection{Technical Lemmas Used in the Proof of Theorem \ref{thm:est}}
\begin{lem}
	\label{lem:Q1}
	Recall $H_n(X_i,\tau)$ defined in \eqref{eq:H}. Under the assumptions in Theorem \ref{thm:est},
	\begin{align*}
	\sup_{\tau \in \Upsilon}\left|\sum_{i=1}^{2n} \left(A_i-\frac{1}{2}\right)H_n(X_i,\tau)\right| =o_p(1),
	\end{align*}
	and
	\begin{align*}
	\sup_{\tau \in \Upsilon}&\left|\sum_{i=1}^{2n} A_i \left[\xi^*_i\int_0^{\frac{u_0+ u_1}{\sqrt{n}}}\left(1\{Y_i(1) -  q_1(\tau)\leq v\} - 1\{Y_i(1) -  q_1(\tau)\leq 0\} \right)dv- H_n(X_i,\tau)\right]\right| = o_p(1),
	\end{align*}
	where either $\xi^*_i = 1$, $\xi^*_i = \xi_i$ defined in Assumption \ref{ass:weight}, or $\xi^*_i = \xi_i^p$ defined in Assumption \ref{ass:pair_boot}.
\end{lem}

\begin{proof}
	For the first result, we have
	\begin{align*}
	& \sup_{\tau \in \Upsilon}\left|\sum_{i=1}^{2n} \left(A_i-\frac{1}{2}\right)H_n(X_i,\tau)\right| \notag \\
	&\leq  \frac{1}{2}\sum_{j=1}^{n} \sup_{\tau \in \Upsilon}|H_n(X_{\pi(2j-1)},\tau) - H_n(X_{\pi(2j)},\tau)| \notag \\
	&\leq  \sum_{j=1}^n \frac{1}{2}\int_0^{\frac{|u_0+u_1|}{\sqrt{n}}}\sup_{\tau \in \Upsilon}|f_{1}(q_1(\tau)+\tilde{v}_j|X_{\pi(2j-1)}) - f_{1}(q_1(\tau)+\tilde{v}_j|X_{\pi(2j)})|vdv \notag \\
	&\lesssim  \sum_{j=1}^n \int_0^{\frac{|u_0+u_1|}{\sqrt{n}}}||X_{\pi(2j-1)} - X_{\pi(2j)}||_2vdv \notag \\
	&\lesssim  \frac{(u_0+u_1)^2}{n}\sum_{j=1}^n ||X_{\pi(2j-1)} - X_{\pi(2j)}||_2 \convP 0,
	\end{align*}
	where the first inequality is due to the fact that for the $j$-th pair, $(A_{\pi(2j-1)}-1/2,A_{\pi(2j)}-1/2)$ is either $(1/2,-1/2)$ or $(-1/2,1/2)$, the second inequality is by standard Taylor expansion to the first order where $|\tilde{v}_j| \leq (|u_0+u_1|)/\sqrt{n}$, the third inequality is due to Assumption \ref{ass:reg}, and the last convergence is due to Assumption \ref{ass:assignment1}.
	
	Let $(\tilde{\xi}_j^*,\tilde{Y}_j(1),\tilde{X}_j) = (\xi_{i_j}^*,Y_{i_j}(1),X_{i_j})$ where $i_j$ is the $j$-th smallest index in the set $\{i\in [2n]: A_i = 1 \}$. Then, we have
	\begin{align*}
	\sup_{\tau \in \Upsilon}&\left|\sum_{i=1}^{2n} A_i \left[\xi^*_i\int_0^{\frac{u_0+u_1}{\sqrt{n}}}\left(1\{Y_i(1) -  q_1(\tau)\leq v\} - 1\{Y_i(1) -  q_1(\tau)\leq 0\} \right)dv- H_n(X_i,\tau)\right]\right| \biggl|\{A_i,X_i\}_{i=1}^{2n} \\
	& \stackrel{d}{=}  ||\mathbb{P}_n - \overline{\mathbb{P}}||_{\mathcal{F}_4}|\{\tilde{X}_j\}_{j=1}^n,
	\end{align*}
	where $\mathcal{F}_4 = \{\tilde{\xi}^*\int_0^{(u_0+u_1)/\sqrt{n}}\left(1\{\tilde{Y}(1) \leq q_1(\tau) + v\} -1\{\tilde{Y}(1) \leq q_1(\tau)\}\right) dv:\tau \in \Upsilon\}$, $\mathbb{P}_nf$ is the usual empirical process, $\overline{\mathbb{P}}f = \frac{1}{n}\sum_{j=1}^n \mathbb{P}^{(j)}f$, and $\mathbb{P}^{(j)}$ denotes the probability measure of $(\tilde{\xi}_j^*,\tilde{Y}_j(1))$ given $\tilde{X}_j$. Note we only have treated units ($A_i = 1$) in the summation. No matter the weights $\xi_i^*$ are generated based on Assumption \ref{ass:weight} or \ref{ass:pair_boot}, $\{\tilde{\xi}_j^*,\tilde{Y}_j(1)\}_{j=1}^n$ are independent conditional on $\{\tilde{X}_j\}_{j=1}^n$. In addition, note that $\mathcal{F}_4$ is a VC-class with a fixed VC index, has an envelope $F_j = (|u_0+u_1|\tilde{\xi}^*_j)/\sqrt{n}$, $M = \max_{j \in [n]}F_j =  (|u_0+u_1| \log(n))/\sqrt{n}$, and
	\begin{align*}
	\sigma^2 = \sup_{f \in \mathcal{F}_4}\overline{\mathbb{P}}f^2\leq & \sup_{\tau \in \Upsilon}\frac{1}{n}\sum_{j=1}^n\left[F_1\left(q_1(\tau) + \frac{|u_0+u_1|}{\sqrt{n}}\biggr|\tilde{X}_j\right) - F_1\left(q_1(\tau) - \frac{|u_0+u_1|}{\sqrt{n}}\biggr|\tilde{X}_j\right)\right] \frac{u^2}{n}  \\
	&\leq  \frac{1}{n}\sum_{j=1}^n C(\tilde{X}_j)\frac{(u_0+u_1)^2}{n^{3/2}} \\
	&=  \frac{1}{n}\sum_{i=1}^{2n} A_iC(X_i)\frac{(u_0+u_1)^2}{n^{3/2}} \\
	&\leq \left(\frac{1}{n}\sum_{i=1}^{2n} C(X_i) \right)\frac{(u_0+u_1)^2}{n^{3/2}}.
	\end{align*}
	As $\left(\frac{1}{n}\sum_{i=1}^{2n} C(X_i) \right) \stackrel{a.s.}{\rightarrow} \mathbb{E}2C(X_i)$, we have $\left(\frac{1}{n}\sum_{i=1}^{2n} C(X_i) \right) \leq 3\mathbb{E}C(X_i)$ a.s. Given such a sequence $\{X_i\}_{i \geq 1}$, Lemma \ref{lem:max_eq} implies
	\begin{align*}
	\mathbb{E}\left[||\mathbb{P}_n - \overline{\mathbb{P}}||_{\mathcal{F}_4}|\{\tilde{X}_j\}_{i=1}^n\right] \lesssim  \sqrt{\frac{ 3\mathbb{E}C(X_i) \log(n)}{n^{3/2}}} + \frac{\log^2(n)}{n} = o_{a.s.}(1).
	\end{align*}
	This implies
	\begin{align*}
	\sup_{\tau \in \Upsilon}&\left|\sum_{i=1}^{2n} A_i \left[\xi^*_i\int_0^{\frac{u_0+ u_1}{\sqrt{n}}}\left(1\{Y_i(1) -  q_1(\tau)\leq v\} - 1\{Y_i(1) -  q_1(\tau)\leq 0\} \right)dv- H_n(X_i,\tau)\right]\right| = o_p(1).
	\end{align*}
\end{proof}

\begin{lem}
	\label{lem:R12}
	Under the assumptions in Theorem \ref{thm:est},
	\begin{align*}
	\sup_{\tau \in \Upsilon}\left|\sum_{i=1}^{2n} \frac{(A_i - 1/2)}{\sqrt{n}}m_{1,\tau}(X_i)\right| = o_p(1).
	\end{align*}
\end{lem}
\begin{proof}
	We have
	\begin{align*}
	\sup_{\tau \in \Upsilon}\left|\sum_{i=1}^{2n} \frac{(A_i - 1/2)}{\sqrt{n}}m_{1,\tau}(X_i)\right| = \sup_{\tau \in \Upsilon}\left|\sum_{j=1}^{n}\frac{1}{2\sqrt{n}}(A_{\pi(2j-1)}-A_{\pi(2j)})(F_1(q_1(\tau)|X_{\pi(2j-1)}) - F_1(q_1(\tau)|X_{\pi(2j)}))\right|.
	\end{align*}
	Note that
	\begin{align*}
	\mathcal{F}_5 = \{F_1(q_1(\tau)|X) - F_1(q_1(\tau)|X'):\tau \in \Upsilon\}
	\end{align*}
	is a VC-class with a fixed VC-index and has an envelope $F=2$. This implies \eqref{eq:VC} holds with some constants $a\geq e$ and $v\geq 1$. Then, as discussed in the \eqref{eq:J2}, the uniform entropy integral $J(\delta)$ of $\mathcal{F}_5$ satisfies
	\begin{align*}
	J(\delta) \leq \int_0^\delta \sqrt{1+\nu \log\left(\frac{a}{\eps}\right)}d\eps \leq 2\sqrt{2\nu}\delta\sqrt{\log\left(\frac{a}{\delta}\right)}.
	\end{align*}

	In addition,
	\begin{align*}
	\sigma_n^2 = \sup_{\tau \in \Upsilon} \frac{1}{n}\sum_{j=1}^n (F_1(q_1(\tau)|X_{\pi(2j-1)}) - F_1(q_1(\tau)|X_{\pi(2j)}))^2 \lesssim \frac{1}{n}\sum_{j=1}^n \left \Vert X_{\pi(2j-1)}-X_{\pi(2j)} \right \Vert^2_2 \convP 0.
	\end{align*}
	
	We focus on the set $\mathcal{A}_n = \{\sigma_n^2 \leq \eps\}$ for some arbitrary $\eps>0$ so that $\mathbb{P}(\mathcal{A}_n) \geq 1-\eps$ for $n$ sufficiently large. Note that $\mathcal{A}_n$ belongs to the sigma field generated by $\{X_i\}_{i=1}^{2n}$. In addition, note that conditional on $\{X_i\}_{i=1}^{2n}$, $\{A_{\pi(2j-1)}-A_{\pi(2j)}\}_{j=1}^n$ is a sequence of i.i.d. Rademacher random variables. Then, following the same argument as in \eqref{eq:Z}
	\begin{align*}
	& \mathbb{E}\sup_{\tau \in \Upsilon}\left|\sum_{i=1}^{2n} \frac{(A_i - 1/2)}{\sqrt{n}}m_{1,\tau}(X_i)\right|1\{\mathcal{A}_n\} \\
	&= \mathbb{E}\left\{\mathbb{E}\left[\left\Vert \frac{1}{2\sqrt{n}} \sum_{j=1}^n (A_{\pi(2j-1)}-A_{\pi(2j)})f(X_{\pi(2j-1)},X_{\pi(2j)})\right\Vert_{\mathcal{F}_5}\biggl|\{X_i\}_{i=1}^{2n}\right]1\{\mathcal{A}_n\}\right\} \\
	&\lesssim  \mathbb{E} J(\sigma_n/2) 1\{\mathcal{A}_n\} \\
	&\lesssim  J(\eps/2) \lesssim \sqrt{2\nu}\eps\sqrt{\log\left(\frac{2a}{\eps}\right)},
	\end{align*}
	where the first inequality is due to \citet[Corollary 2.2.8]{VW96} and the fact that, by Hoeffding's inequality, for any $f \in \mathcal{F}_5$,
	\begin{align*}
	\mathbb{P}\biggl(\left|\sum_{j=1}^n(A_{\pi(2j-1)}-A_{\pi(2j)})f(X_{\pi(2j-1)},X_{\pi(2j)})\right| \geq x \biggl| \{X_i\}_{i=1}^{2n}\biggr) \leq 2 \exp\left(-\frac{1}{2}\frac{x^2}{\sum_{j=1}^n f^2(X_{\pi(2j-1)},X_{\pi(2j)})}\right).
	\end{align*}
	As $\sqrt{2\nu}\eps\sqrt{\log\left(\frac{2a}{\eps}\right)} \rightarrow 0$ as $\eps \rightarrow 0$, we derive the desired result by letting $n \rightarrow \infty$ followed by $\eps \rightarrow 0$.

\end{proof}

\subsection{Technical Lemmas Used in the Proof of Theorem \ref{thm:weight}}
\begin{lem}
	\label{lem:w1}
	Suppose the assumptions in Theorem \ref{thm:weight} hold, then
	\begin{align*}
	\frac{1}{n}\sum_{i =1}^{2n}A_i \eta_{1,i}(\tau) m_{j,\tau' }(X_i) \convP 0,
	\end{align*}
	\begin{align*}
	\frac{1}{n}\sum_{i =1}^{2n}A_i \eta_{1,i}(\tau) m_{0,\tau' }(X_i) \convP 0,
	\end{align*}
	\begin{align*}
	\frac{1}{n}\sum_{i =1}^{2n}(1-A_i) \eta_{0,i}(\tau) m_{0,\tau' }(X_i) \convP 0,
	\end{align*}
	and
	\begin{align*}
	\frac{1}{n}\sum_{i =1}^{2n}(1-A_i) \eta_{0,i}(\tau) m_{1,\tau' }(X_i) \convP 0.
	\end{align*}
\end{lem}
\begin{proof}
	We focus on the first statement. The rest can be proved in the same manner. Based on the notation in Section \ref{sec:boot}, we have
	\begin{align*}
	\frac{1}{n}\sum_{i =1}^{2n}A_i \eta_{1,i}(\tau) m_{1,\tau' }(X_i)  = \frac{1}{n}\sum_{j =1}^{n} \eta_{1,(j,1)}(q_1(\tau),\tau) m_{1,\tau' }(X_{(j,1)},q_1(\tau' )).
	\end{align*}
	where $\eta_{1,i}(q,\tau) = (\tau - 1\{Y_i(1) \leq q\}) - m_{1,\tau}(X_{i},q)$. This is a special case of the RHS of \eqref{eq:II1} with $\nu = \nu^\top  = 0$. Therefore, the argument after \eqref{eq:II1} shows
	\begin{align*}
	\frac{1}{n}\sum_{j =1}^{n} \eta_{1,(j,1)}(q_1(\tau),\tau) m_{1,\tau' }(X_{(j,1)},q_1(\tau' )) \convP 0.
	\end{align*}
	
\end{proof}

\subsection{Technical Lemmas Used in the Proof of Theorem \ref{thm:pair}}
\begin{lem}
	\label{lem:Q1m}
	Recall $H_n(X_i,\tau)$ defined in \eqref{eq:H}. Under the assumptions in Theorem \ref{thm:pair},
	\begin{align*}
	\sup_{\tau \in \Upsilon}&\left|\sum_{i=1}^{2n} A_i \left[\xi_i\int_0^{\frac{u_0+ u_1}{\sqrt{n}}}\left(1\{Y_i(1) -  q_1(\tau)\leq v\} - 1\{Y_i(1) -  q_1(\tau)\leq 0\} \right)dv- H_n(X_i,\tau)\right]\right| = o_p(1),
	\end{align*}
	where $\xi_i$ is defined in Assumption \ref{ass:weight}. The same result holds if $\xi_i$ is replaced by $\xi_i^p$ as defined in Assumption \ref{ass:pair_boot}.
\end{lem}

\begin{proof}
	Let $(\tilde{\xi}_j,\tilde{Y}_j(1),\tilde{X}_j) = (\xi_{i_j},Y_{i_j}(1),X_{i_j})$ where $i_j$ is the $j$-th smallest index in the set $\{i\in [2n]: A_i = 1 \}$. Then, similar to \eqref{eq:=dm}, we have
	\begin{align*}
	\sup_{\tau \in \Upsilon}&\left|\sum_{i=1}^{2n} A_i \left[\xi_i\int_0^{\frac{u_0+u_1}{\sqrt{n}}}\left(1\{Y_i(1) -  q_1(\tau)\leq v\} - 1\{Y_i(1) -  q_1(\tau)\leq 0\} \right)dv- H_n(X_i,\tau)\right]\right| \biggl|\{A_i,X_i\}_{i=1}^{2n} \\
	& \stackrel{d}{=}  ||\mathbb{P}_n - \overline{\mathbb{P}}||_{\mathcal{F}_4}|\{\tilde{X}_j\}_{j=1}^n,
	\end{align*}
	where $\mathcal{F}_4 = \{\tilde{\xi}\int_0^{(u_0+u_1)/\sqrt{n}}\left(1\{\tilde{Y}(1) \leq q_1(\tau) + v\} -1\{\tilde{Y}(1) \leq q_1(\tau)\}\right) dv:\tau \in \Upsilon\}$, $\mathbb{P}_nf$ is the usual empirical process, $\overline{\mathbb{P}}f = \frac{1}{n}\sum_{j=1}^n \mathbb{P}^{(j)}f$, and $\mathbb{P}^{(j)}$ denotes the probability measure of $(\tilde{\xi}_j,\tilde{Y}_j(1))$ given $\tilde{X}_j$ such that given $\{\tilde{X}_j\}_{j=1}^n$, $\{\tilde{\xi}_j,\tilde{Y}_j(1)\}_{j=1}^n$ are independent. Note that $\mathcal{F}_4$ is a VC-class with a fixed VC index, has an envelope $F_j = (|u_0+u_1|\tilde{\xi}_j)/\sqrt{n}$, $M = \max_{j \in [n]}F_j =  (|u_0+u_1| \log(n))/\sqrt{n}$, and
	\begin{align*}
	\sigma^2 = \sup_{f \in \mathcal{F}_4}\overline{\mathbb{P}}f^2\leq & \sup_{\tau \in \Upsilon}\frac{1}{n}\sum_{j=1}^n\left[F_1\left(q_1(\tau) + \frac{|u_0+u_1|}{\sqrt{n}}\biggr|\tilde{X}_j\right) - F_1\left(q_1(\tau) - \frac{|u_0+u_1|}{\sqrt{n}}\biggr|\tilde{X}_j\right)\right] \frac{u^2}{n}  \\
	&\leq  \frac{1}{n}\sum_{j=1}^n C(\tilde{X}_j)\frac{(u_0+u_1)^2}{n^{3/2}} \\
	&=  \frac{1}{n}\sum_{i=1}^{2n} A_iC(X_i)\frac{(u_0+u_1)^2}{n^{3/2}} \\
	&\leq \left(\frac{1}{n}\sum_{i=1}^{2n} C(X_i) \right)\frac{(u_0+u_1)^2}{n^{3/2}}.
	\end{align*}
	As $\left(\frac{1}{n}\sum_{i=1}^{2n} C(X_i) \right) \stackrel{a.s.}{\rightarrow} \mathbb{E}2C(X_i)$, we have $\left(\frac{1}{n}\sum_{i=1}^{2n} C(X_i) \right) \leq 3\mathbb{E}C(X_i)$ a.s. Given such a sequence $\{X_i\}_{i \geq 1}$, Lemma \ref{lem:max_eq} implies
	\begin{align*}
	\mathbb{E}\left[||\mathbb{P}_n - \overline{\mathbb{P}}||_{\mathcal{F}_4}|\{\tilde{X}_j\}_{i=1}^n\right] \lesssim  \sqrt{\frac{ 3\mathbb{E}C(X_i) \log(n)}{n^{3/2}}} + \frac{\log^2(n)}{n} = o_{a.s.}(1).
	\end{align*}
	This implies
	\begin{align*}
	\sup_{\tau \in \Upsilon}&\left|\sum_{i=1}^{2n} A_i \left[\xi_i\int_0^{\frac{u_0+ u_1}{\sqrt{n}}}\left(1\{Y_i(1) -  q_1(\tau)\leq v\} - 1\{Y_i(1) -  q_1(\tau)\leq 0\} \right)dv- H_n(X_i,\tau)\right]\right| = o_p(1).
	\end{align*}
	
\end{proof}

\subsection{Technical Lemmas Used in the Proof of Theorem \ref{thm:boot}}
\begin{lem}
	\label{lem:23}
	Recall $II(\tau,\tau' )$ and $III(\tau,\tau' )$ defined in \eqref{eq:sigma111}. Suppose the assumptions in Theorem \ref{thm:est} hold, then
	\begin{align*}
	\sup_{\tau, \tau' \in \Upsilon}|II(\tau,\tau' )| \convP 0 \quad \text{and} \quad \sup_{\tau, \tau' \in \Upsilon}|III(\tau,\tau' )| \convP 0.
	\end{align*}
\end{lem}

\begin{proof}
	We focus on bounding $II(\tau,\tau' )$. The bound for $III(\tau,\tau' )$ can be established similarly. By \eqref{eq:qhat}, we have, with probability greater than $1-\eps$,
	\begin{align}
	\label{eq:II1}
	|II(\tau,\tau' )| \leq \sup_{\tau, \tau' \in \Upsilon, |v|,|v'|\leq L/\sqrt{n}}\left|\frac{1}{n}\sum_{j =1}^n\eta_{1,(j,1)}(q_1(\tau)+v,\tau)m_{1,\tau' }(X_{(j,1)},q_1(\tau' )+v')  \right|.
	\end{align}
	We aim to bound the RHS. Let $\{\eps_j\}_{j=1}^n$ denote a sequence of i.i.d. Rademacher random variables that is independent of the data. Further denote the class of functions
	\begin{align*}
	\mathcal{F}_6 = \{\eta_{1,(j,1)}(q_1(\tau)+v,\tau)m_{1,\tau' }(X_{(j,1)},q_1(\tau' )+v'):\tau,\tau'  \in \Upsilon, |v|,|v'| \leq L/\sqrt{n} \}.
	\end{align*}
	Note $\mathcal{F}_6$ has an envelope $F=1$ and is nested by a VC-class of functions with a fixed VC-index. Then,
	\begin{align}
	\label{eq:II2}
	& \mathbb{E}\left[  \sup_{\tau, \tau' \in \Upsilon, |v|,|v'|\leq L/\sqrt{n}}\left|\frac{1}{n}\sum_{j =1}^n\eta_{1,(j,1)}(q_1(\tau)+v,\tau)m_{1,\tau' }(X_{(j,1)},q_1(\tau' )+v')  \right|    \right] \notag \\
	&= \mathbb{E} \left\{
	\mathbb{E}\left[  \sup_{\tau, \tau' \in \Upsilon, |v|,|v'|\leq L/\sqrt{n}}\left|\frac{1}{n}\sum_{j =1}^n\eta_{1,(j,1)}(q_1(\tau)+v,\tau)m_{1,\tau' }(X_{(j,1)},q_1(\tau' )+v')  \right|    |\{X_i,A_i\}_{i=1}^{2n}\right] \right\} \notag \\
	&\lesssim   \mathbb{E} \left\{
	\mathbb{E}\left[  \sup_{\tau, \tau' \in \Upsilon, |v|,|v'|\leq L/\sqrt{n}}\left|\frac{1}{n}\sum_{j =1}^n\eps_j\eta_{1,(j,1)}(q_1(\tau)+v,\tau)m_{1,\tau' }(X_{(j,1)},q_1(\tau' )+v')  \right|    |\{X_i,A_i\}_{i=1}^{2n}\right] \right\} \notag \\
	&=  \mathbb{E} \left\{
	\mathbb{E}\left[  \sup_{\tau, \tau' \in \Upsilon, |v|,|v'|\leq L/\sqrt{n}}\left|\frac{1}{n}\sum_{j =1}^n\eps_j\eta_{1,(j,1)}(q_1(\tau)+v,\tau)m_{1,\tau' }(X_{(j,1)},q_1(\tau' )+v')  \right|    |\{X_i,A_i,Y_i(1)\}_{i=1}^{2n}\right] \right\} \notag \\
	&\leq \frac{||F||_{\overline{\mathbb{P}},2}J(\sqrt{\mathbb{E}\sigma_n^2}/||F||_{\overline{\mathbb{P}},2})}{\sqrt{n}} \lesssim \frac{1}{\sqrt{n}},
	\end{align}
	where the first equality is due to the law of iterated expectation, the first inequality is due to \citet[Lemma 6.3]{LT13} and the fact that $\{\eta_{1,(j,1)}(q_1(\tau)+v,\tau)\}_{j=1}^n$ is a sequence of independent and centered random variables given $\{X_i,A_i\}_{i=1}^{2n}$, the second inequality follows the same argument in \eqref{eq:Z} with $F = 2$,
	\begin{align*}
	\sigma_n^2 =  \sup_{\tau, \tau' \in \Upsilon, |v|,|v'|\leq L/\sqrt{n}}\frac{1}{n}\sum_{j =1}^n\left[\eta_{1,(j,1)}(q_1(\tau)+v,\tau)m_{1,\tau' }(X_{(j,1)},q_1(\tau' )+v')  \right]^2 \leq 4,
	\end{align*}
	and $J(\cdot)$ being the uniform entropy integral for the class of functions $\mathcal{F}_6$ defined in \eqref{eq:J}, and the last inequality holds because when $\mathcal{F}_6$ is nested by a VC-class, $\eps_i$ is bounded, and thus, has a sub-Gaussian tail, and $\delta = \sqrt{\mathbb{E}\sigma_n^2}/||F||_{\overline{\mathbb{P}},2} \leq 1$, we have
	\begin{align*}
	J(\delta) \lesssim \delta \max(\sqrt{\log(1/\delta)},1)\lesssim 1,
	\end{align*}
	as shown in \eqref{eq:J2}. This implies, uniformly over $\tau,\tau'  \in \Upsilon$,
	\begin{align*}
	II(\tau,\tau' ) \convP 0.
	\end{align*}
	
\end{proof}

\begin{lem}
	\label{lem:4}
	Recall $R_{IV}(\tau,\tau' )$ defined in \eqref{eq:IV}. Suppose assumptions in Theorem \ref{thm:boot} hold, then
	\begin{align*}
	\sup_{\tau, \tau' \in \Upsilon}|R_{IV}(\tau,\tau' )| = o_p(1)   \quad \text{and} \quad \sup_{\tau, \tau' \in \Upsilon}\left|\frac{1}{n}\sum_{i =1}^{2n} \left(A_i-\frac{1}{2}\right) m_{1,\tau}(X_i)m_{1,\tau' }(X_i)\right| =o_p(1).
	\end{align*}
\end{lem}
\begin{proof}
	Note
	\begin{align*}
	R_{IV}(\tau,\tau' ) = \frac{1}{n}\sum_{j =1}^n \left[m_{1,\tau}(X_{(j,1)})m_{1,\tau' }(X_{(j,1)}) - m_{1,\tau}(X_{(j,1)},\hat{q}_1(\tau))m_{1,\tau' }(X_{(j,1)},\hat{q}_1(\tau' ))\right].
	\end{align*}
	By \eqref{eq:qhat} and the fact that $F_1(\cdot|X)$ is Lipschitz continuous, we have
	\begin{align*}
	& \sup_{\tau, \tau' \in \Upsilon}|R_{IV}(\tau,\tau' )| \\
	\leq &   \sup_{\tau, \tau' \in \Upsilon} \frac{1}{n}\sum_{j =1}^n \left|m_{1,\tau}(X_{(j,1)})m_{1,\tau' }(X_{(j,1)}) - m_{1,\tau}(X_{(j,1)},\hat{q}_1(\tau))m_{1,\tau' }(X_{(j,1)},\hat{q}_1(\tau' )) \right| \convP 0.
	\end{align*}
	By the same argument as in the proof of Lemma \ref{lem:R12}, we have
	\begin{align*}
	\sup_{\tau, \tau' \in \Upsilon}\left|\frac{1}{n}\sum_{i =1}^{2n} \left(A_i-\frac{1}{2}\right) m_{1,\tau}(X_i)m_{1,\tau' }(X_i)\right| \convP 0.
	\end{align*}
\end{proof}

\begin{lem}
	\label{lem:tight1}
	Recall $S_{n,1}^*(\tau)$ defined in \eqref{eq:S1}. Suppose assumptions in Theorem \ref{thm:est} hold. Then,
	$\{S_{n,1}^*(\tau):\tau \in \Upsilon\}$ is stochastically equicontinuous and $\sup_{\tau \in \Upsilon}||S_{n,1}^*(\tau)||_2 = O_p(1)$.
\end{lem}

\begin{proof}
	It suffices to show the two marginals of $S_{n,1}^*(\tau)$ are stochastically equicontinuous and uniformly bounded in probability. We focus on the first marginal
	\begin{align*}
	\left\{\frac{1}{\sqrt{n}}\sum_{j=1}^n \eta_j(\tau - 1\{Y_{(j,1)} \leq \hat{q}_1(\tau)\}):\tau \in \Upsilon \right\}.
	\end{align*}
	Based on \eqref{eq:qhat}, it suffices to establish the stochastic equicontinuity of
	\begin{align*}
	\left\{\frac{1}{\sqrt{n}}\sum_{j=1}^n \eta_j(\tau - 1\{Y_{(j,1)} \leq q_1(\tau) + v/\sqrt{n}\}):\tau \in \Upsilon, |v|\leq L \right\}
	\end{align*}
	for any fixed $L$. Let
	\begin{align*}
	\mathcal{F}_7 = \begin{Bmatrix}
	(\tau - 1\{Y_{(j,1)} \leq q_1(\tau) + v/\sqrt{n}\}) - (\tau'  - 1\{Y_{(j,1)} \leq q_1(\tau' ) + v'/\sqrt{n}\}): \\
	\tau,\tau'  \in \Upsilon, |v|,|v'|\leq L, |\tau - \tau' |\leq \eps, |v-v'|\leq \eps
	\end{Bmatrix},
	\end{align*}
	which is nested by a VC-class with envelope $2$. Then, by \eqref{eq:J} and \eqref{eq:J2}, the uniform entropy integral $J(\delta)$ of $\mathcal{F}_7$ satisfies
	\begin{align*}
	J(\delta) \lesssim \delta \max(1,\sqrt{\log(1/\delta)}).
	\end{align*}

	By the calculation of $\tilde{\Sigma}^*_{1,1,1}(\tau,\tau' )$ (with $\hat{q}_1(\tau)$ replaced by $q_1(\tau) + \frac{v}{\sqrt{n}}$) in Section \ref{sec:boot_pf}, we have, uniformly over $\tau,\tau'  \in \Upsilon$, $v, v' \in [-L,L]$,
	\begin{align}
	\label{eq:sigma2n}
	\sigma_n^2(\tau,\tau' ,v,v') = & \frac{1}{n}\sum_{j=1}^{n} \left[(\tau - 1\{Y_{(j,1)} \leq q_1(\tau) + v/\sqrt{n}\}) - (\tau'  - 1\{Y_{(j,1)} \leq q_1(\tau' ) + v'/\sqrt{n}\})\right]^2 \notag \\
	\convP & \tau(1-\tau) + \tau' (1-\tau' ) - 2(\min(\tau,\tau' ) - \tau \tau' ) = |\tau - \tau' | - (\tau - \tau' )^2.
	\end{align}
	Let $\mathcal{A}_n(\eps) = 1\{ \sup_{\tau,\tau'  \in \Upsilon,v, v' \in [-L,L]} |\sigma_n^2(\tau,\tau' ,v,v') - \left( |\tau - \tau' | - (\tau - \tau' )^2 \right)|\leq \eps  \}$,  which will occur with probability approaching one. Also by construction, conditionally on data, $\frac{1}{\sqrt{n}}\sum_{j=1}^n \eta_j(\tau - 1\{Y_{(j,1)} \leq q_1(\tau) + v/\sqrt{n}\})$ is a sub-Gaussian process. Then,
	\small{
		\begin{align*}
		& \mathbb{E} \left[\sup \frac{1}{\sqrt{n}}\sum_{j=1}^n \eta_j(\tau - 1\{Y_{(j,1)} \leq q_1(\tau) + v/\sqrt{n}\}) - (\tau'  - 1\{Y_{(j,1)} \leq q_1(\tau' ) + v'/\sqrt{n}\}) \biggl| Data \right]1\{\mathcal{A}_n(\eps)\}\\
		&\lesssim J(\frac{\sup\sigma_n(\tau,\tau' ,v,v')}{2})1\{\mathcal{A}_n(\eps)\} \\
		&\lesssim J(\sqrt{\eps}) \lesssim \sqrt{\eps}\max(1,\sqrt{\log(1/\eps)}),
		\end{align*}}
	\normalsize
	where the supremum is taken over $ \tau,\tau'  \in \Upsilon, |v|,|v'|\leq L, |\tau - \tau' |\leq \eps, |v-v'|\leq \eps$, the first inequality is due to \cite[Corollary 2.2.8]{VW96}, and the second inequality is due to \eqref{eq:sigma2n} and the definition of $\mathcal{A}_n$. Then, for any $t>0$
	\small
	\begin{align*}
	&\mathbb{P}\left( \sup \frac{1}{\sqrt{n}}\sum_{j=1}^n \eta_j\left[(\tau - 1\{Y_{(j,1)} \leq q_1(\tau) + v/\sqrt{n}\}) - (\tau'  - 1\{Y \leq q_1(\tau' ) + v'/\sqrt{n}\}) \right] \geq t\right) \\
	&\leq \mathbb{P}(\mathcal{A}_n^c(\eps)) + \mathbb{P}\left( \sup \frac{1}{\sqrt{n}}\sum_{j=1}^n \eta_j\left[(\tau - 1\{Y \leq q_1(\tau) + v/\sqrt{n}\}) - (\tau'  - 1\{Y \leq q_1(\tau' ) + v'/\sqrt{n}\}) \right] \geq t,\mathcal{A}_n(\eps)\right) \\
	&\leq  \mathbb{E} \left\{ \frac{ \mathbb{E} \left[\sup \frac{1}{\sqrt{n}}\sum_{j=1}^n \eta_j(\tau - 1\{Y_{(j,1)} \leq q_1(\tau) + v/\sqrt{n}\}) - (\tau'  - 1\{Y_{(j,1)} \leq q_1(\tau' ) + v'/\sqrt{n}\}) \biggl| Data \right]1\{\mathcal{A}_n(\eps)\}}{t}\right\} \\
	& + \mathbb{P}(\mathcal{A}_n^c(\eps)) \\
	&\lesssim  \mathbb{P}(\mathcal{A}_n^c(\eps)) + \frac{\sqrt{\eps}\max(1,\sqrt{\log(1/\eps)})}{t},
	\end{align*}
	\normalsize
	where the supremum is taken over $ \tau,\tau'  \in \Upsilon, |v|,|v'|\leq L, |\tau - \tau' |\leq \eps, |v-v'|\leq \eps$. Let $n\rightarrow \infty$ followed by $\eps \rightarrow 0$, we have
	\small
	\begin{align*}
	\lim_{\eps \rightarrow 0} \limsup_{n} \mathbb{P}\left( \sup \frac{1}{\sqrt{n}}\sum_{j=1}^n \eta_j\left[(\tau - 1\{Y_{(j,1)} \leq q_1(\tau) + v/\sqrt{n}\}) - (\tau'  - 1\{Y_{(j,1)}\leq q_1(\tau' ) + v'/\sqrt{n}\}) \right] \geq t\right) = 0,
	\end{align*}
	\normalsize
	which implies $\left\{\frac{1}{\sqrt{n}}\sum_{j=1}^n \eta_j(\tau - 1\{Y_{(j,1)} \leq \hat{q}_1(\tau)\}):\tau \in \Upsilon \right\}$ is stochastically equicontinuous. In addition, for any fixed $\tau$,
	\begin{align*}
	\frac{1}{\sqrt{n}}\sum_{j=1}^n \eta_j(\tau - 1\{Y_{(j,1)} \leq \hat{q}_1(\tau)\}) = O_p(1).
	\end{align*}
	This implies $\sup_{\tau \in \Upsilon}|\frac{1}{\sqrt{n}}\sum_{j=1}^n \eta_j(\tau - 1\{Y_{(j,1)} \leq \hat{q}_1(\tau)\})| = O_p(1)$.
\end{proof}

\subsection{Technical Lemmas Used in the Proof of Theorem \ref{thm:ipw_w}}
\begin{lem}
	\label{lem:sieve}
	Suppose the assumptions in Theorem \ref{thm:ipw_w} hold, then
	\begin{align*}
	\max_{i \in [2n]}|\hat{A}_i-1/2| = o_p(1)
	\end{align*}
	and
	\begin{align*}
	\frac{1}{n}\sum_{i=1}^{2n}\xi_i(\hat{A}_i-1/2)^2 = o_p(n^{-1/2}).
	\end{align*}
\end{lem}
\begin{proof}
	Let $\theta_0 = (0.5,0,\cdots,0)^\top $ be a $K\times 1$ vector. Then,
	
	\begin{align*}
	||\hat{\theta} - \theta_0||_2 = & \left\Vert \left[\frac{1}{n}\sum_{i=1}^{2n}\xi_i b(X_i)b(X_i)^\top \right]^{-1}\left[\frac{1}{n}\sum_{i=1}^{2n}\xi_i b(X_i)(A_i - \frac{1}{2})\right]\right\Vert_2 \\
	&\lesssim  \left\Vert\frac{1}{n}\sum_{i=1}^{2n}\xi_i b(X_i)(A_i - \frac{1}{2})\right\Vert_2 \\
	&	\lesssim  \sqrt{K}\left\Vert\frac{1}{n}\sum_{i=1}^{2n}\xi_i b(X_i)(A_i - \frac{1}{2})\right\Vert_\infty.
	\end{align*}
	
	Next, we aim to bound $\left\Vert\frac{1}{n}\sum_{i=1}^{2n}\xi_i b(X_i)(A_i - \frac{1}{2})\right\Vert_\infty$. Let $b_k(X)$ be the $k$th component of $b(X)$. Then,
	\begin{align*}
	& \max_{k \in [K]}\frac{1}{n}\sum_{j=1}^{n}(\xi_{\pi(2j-1)} b_k(X_{\pi(2j-1)})-\xi_{\pi(2j)} b_k(X_{\pi(2j)}))^2 \\
	&\lesssim   \max_{k \in [K]}\frac{1}{n}\sum_{i=1}^{2n}\xi_{i}^2 b^2_k(X_{i}) \\
	&\lesssim  \max_{k \in [K]}\mathbb{E}\xi_{i}^2 b^2_k(X_{i}) + \max_{k \in [K]}\left|\frac{1}{n}\sum_{i=1}^{2n}\left[\xi_{i}^2 b^2_k(X_{i}) - \mathbb{E}\xi_{i}^2 b^2_k(X_{i})\right]\right|.
	\end{align*}
	The first term on the RHS of the above display is bounded by $\overline{C}$ based on Assumption \ref{ass:sieve}.
	Let $\{\eps_i\}_{i \in [2n]}$ be a sequence of i.i.d. Rademacher random variables. Then,
	\begin{align*}
	\mathbb{E}\max_{k \in [K]}\left|\frac{1}{n}\sum_{i=1}^{2n}\left[\xi_{i}^2 b^2_k(X_{i}) - \mathbb{E}\xi_{i}^2 b^2_k(X_{i})\right]\right| \leq 2\mathbb{E}\max_{k \in [K]}\left|\frac{1}{n}\sum_{i=1}^{2n}\eps_i\left[\xi_{i}^2 b^2_k(X_{i}) - \mathbb{E}\xi_{i}^2 b^2_k(X_{i})\right]\right|.
	\end{align*}
	By Hoeffding's inequality,
	\begin{align*}
	\mathbb{P}\left(\left|\frac{1}{\sqrt{2n}}\sum_{i=1}^{2n}\eps_i\left[\xi_{i}^2 b^2_k(X_{i}) - \mathbb{E}\xi_{i}^2 b^2_k(X_{i})\right]\right| \geq t \biggl|\{\xi_i,X_i\}_{i \in [2n]}\right) \leq 2\exp(-\frac{t^2}{2\sigma_k^2}),
	\end{align*}
	where $\sigma_k^2 = \frac{1}{2n}\sum_{i=1}^{2n}\left[\xi_{i}^2 b^2_k(X_{i}) - \mathbb{E}\xi_{i}^2 b^2_k(X_{i})\right]^2$. Then, by \citet[Lemmas 2.2.1 and 2.2.2]{VW96},
	\begin{align*}
	\mathbb{E}\left[\max_{k \in [K]}\left|\frac{1}{n}\sum_{i=1}^{2n}\eps_i\xi_{i}^2 b^2_k(X_{i})\right|\biggl|\{\xi_i,X_i\}_{i \in [2n]}\right] \lesssim \sqrt{\frac{\log(K)}{n}}\sqrt{\max_{k \in [K]}\sigma_k^2}.
	\end{align*}
	Applying expectation on both sides and noticing that the square root function is concave, we have
	\begin{align*}
	\mathbb{E}\max_{k \in [K]}\left|\frac{1}{n}\sum_{i=1}^{2n}\eps_i\xi_{i}^2 b^2_k(X_{i})\right| \lesssim & \sqrt{\frac{\log(K)}{n}} \sqrt{\mathbb{E}\max_{k \in [K]}\sigma_k^2} \\
	\lesssim & \sqrt{\frac{\log(K)}{n}} \sqrt{\sum_{k \in [K]}\mathbb{E}\sigma_k^2} \\
	\lesssim & \sqrt{\frac{\log(K)}{n}} \zeta(K) \sqrt{K} = o(1).
	\end{align*}
	Therefore,
	\begin{align*}
	\max_{k \in [K]}\left|\frac{1}{n}\sum_{i=1}^{2n}\left[\xi_{i}^2 b^2_k(X_{i}) - \mathbb{E}\xi_{i}^2 b^2_k(X_{i})\right]\right| = o_p(1)
	\end{align*}
	and with probability approaching one,
	\begin{align*}
	\max_{k \in [K]}\frac{1}{n}\sum_{j=1}^{n}(\xi_{\pi(2j-1)} b_k(X_{\pi(2j-1)})-\xi_{\pi(2j)} b_k(X_{\pi(2j)}))^2 \leq 2\overline{C}.
	\end{align*}
	
	Let $I'_n = \{\max_{k \in [K]}\frac{1}{n}\sum_{j=1}^{n}(\xi_{\pi(2j-1)} b_k(X_{\pi(2j-1)})-\xi_{\pi(2j)} b_k(X_{\pi(2j)}))^2 \leq 2\overline{C}\}$. For $t=\sqrt{\log(n)\overline{C}}$, we have
	\begin{align*}
	& \mathbb{P}\left(\left\Vert\frac{1}{n}\sum_{i=1}^{2n}\xi_i b(X_i)(A_i - \frac{1}{2})\right\Vert_\infty \geq t/\sqrt{n},I'_n\right) \\
	= & \mathbb{E}\mathbb{P}\left(\left\Vert\frac{1}{n}\sum_{i=1}^{2n}\xi_i b(X_i)(A_i - \frac{1}{2})\right\Vert_\infty\geq t/\sqrt{n}\biggl|\{X_i,\xi_i\}_{i \in [2n]}\right) 1\{I'_n\}\\
	= & \mathbb{E}\mathbb{P}\left(\left\Vert \sum_{j=1}^{n} (A_{\pi(2j-1)}-A_{\pi(2j)})(\xi_{\pi(2j-1)} b(X_{\pi(2j-1)})-\xi_{\pi(2j)} b(X_{\pi(2j)}))\right\Vert_\infty\geq 2t\sqrt{n}\biggl|\{X_i,\xi_i\}_{i \in [2n]}\right)1\{I'_n\} \\
	\leq & \sum_{k=1}^K\mathbb{E}\mathbb{P}\left(\left| \sum_{j=1}^{n} (A_{\pi(2j-1)}-A_{\pi(2j)})(\xi_{\pi(2j-1)} b_k(X_{\pi(2j-1)})-\xi_{\pi(2j)} b_k(X_{\pi(2j)}))\right|\geq 2t\sqrt{n}\biggl|\{X_i,\xi_i\}_{i \in [2n]}\right)1\{I'_n\}\\
	\leq & \sum_{k=1}^K 2\mathbb{E}\exp\left( \frac{-2t^2n}{\sum_{j=1}^{n}(\xi_{\pi(2j-1)} b_k(X_{\pi(2j-1)})-\xi_{\pi(2j)} b_k(X_{\pi(2j)}))^2} \right)1\{I'_n\} \\
	\leq &  2\exp\left(\log(K)- \frac{t^2}{\overline{C}}\right) \rightarrow 0,
	\end{align*}
	where the second last inequality is due to Hoeffding's inequality and the fact that given $\{X_i,\xi_i\}_{i \in [2n]}$, $\{A_{\pi(2j-1)}-A_{\pi(2j)}\}_{j \in [n]}$ is an i.i.d. sequence of Rademacher random variables.
	
	This implies,
	\begin{align*}
	||\hat{\theta}-\theta_0 ||_2 = O_p(\sqrt{\frac{K \log(n)}{n}}),
	\end{align*}
	and thus
	\begin{align*}
	\max_{i \in [2n]}|\hat{A}_i-1/2| = \max_i|b^\top(X_i)  (\hat{\theta}-\theta_0)| = O_p\left( \zeta(K)\sqrt{\frac{K \log(n)}{n}}\right) = o_p(1).
	\end{align*}
	
	For the second result, we have
	\begin{align*}
	\frac{1}{n}\sum_{i=1}^{2n}\xi_i(\hat{A}_i-1/2)^2 \leq \lambda_{\max}\left(\frac{1}{n}\sum_{i=1}^{2n}\xi_ib(X_i)b^\top(X_i)  \right)||\hat{\theta}-\theta_0||_2^2 =  O_p\left(\frac{K \log(n)}{n}\right) = o_p(n^{-1/2}),
	\end{align*}
	as $K^2\log^2(n) = o(n)$.
	
\end{proof}

\begin{lem}
	\label{lem:sieve2}
	Suppose assumptions in Theorem \ref{thm:ipw_w} hold, then
	\begin{align*}
	\sup_{\tau \in \Upsilon}\left|\sum_{i=1}^{2n}\frac{\xi_i m_{1,\tau}(X_i) (\hat{A}_i-1/2)}{\sqrt{n}} - \sum_{i=1}^{2n}\frac{\xi_i m_{1,\tau}(X_i)(A_i - 1/2)}{\sqrt{n}}\right| = o_p(1),
	\end{align*}
	\begin{align*}
	\sup_{\tau \in \Upsilon}\left|\sum_{i=1}^{2n}\frac{2\xi_i (A_i - 1/2)m_{1,\tau}(X_i) (\hat{A}_i-1/2)}{\sqrt{n}}\right| = o_p(1),
	\end{align*}
	and
	\begin{align*}
	\sup_{\tau \in \Upsilon}\left|\sum_{i=1}^{2n}\frac{2\xi_iA_i \eta_{1,i}(\tau) (\hat{A}_i-1/2)}{\sqrt{n}}\right| = o_p(1).
	\end{align*}
\end{lem}
\begin{proof}
	For the first result, note $m_{1,\tau}(X_i) = b^\top(X_i)  \gamma_1(\tau)+B_\tau(X_i)$ such that $\sup_{x \in \Supp(X),\tau \in \Upsilon}|B_\tau(x)| = o(1/\sqrt{n})$. Then,
	\begin{align*}
	& \sum_{i=1}^{2n}\frac{\xi_i m_{1,\tau}(X_i) (\hat{A}_i-1/2)}{\sqrt{n}} \\
	&=  \sum_{i=1}^{2n}\frac{\xi_i m_{1,\tau}(X_i) b^\top(X_i)  (\hat{\theta} - \theta_0)}{\sqrt{n}} \\
	&= \gamma_1^\top(\tau)\left[\sum_{i=1}^{2n}\frac{\xi_ib(X_i)b^\top(X_i)  }{\sqrt{n}}\right] (\hat{\theta} - \theta_0) + \sum_{i=1}^{2n}\frac{\xi_i B_\tau(X_i) b^\top(X_i)  (\hat{\theta} - \theta_0)}{\sqrt{n}} \\
	&= \sum_{i=1}^{2n}\frac{\xi_i\gamma_1^\top(\tau)b(X_i)(A_i-1/2)}{\sqrt{n}}+ \sum_{i=1}^{2n}\frac{\xi_i B_\tau(X_i) b^\top(X_i)  (\hat{\theta} - \theta_0)}{\sqrt{n}}  \\
	&= \sum_{i=1}^{2n}\frac{\xi_im_{1,\tau}(X_i)(A_i-1/2)}{\sqrt{n}}-\sum_{i=1}^{2n}\frac{\xi_i B_\tau(X_i)(A_i-1/2)}{\sqrt{n}} + \sum_{i=1}^{2n}\frac{\xi_i B_\tau(X_i) b^\top(X_i)  (\hat{\theta} - \theta_0)}{\sqrt{n}},
	\end{align*}
	where the third equality holds because
	\begin{align*}
	\hat{\theta} - \theta_0 = \left[\sum_{i=1}^{2n}\frac{\xi_ib(X_i)b^\top(X_i)  }{n}\right]^{-1}\left[\sum_{i=1}^{2n}\frac{\xi_ib(X_i)(A_i-1/2)}{n}\right].
	\end{align*}
	Furthermore,
	\begin{align*}
	\sup_{\tau \in \Upsilon}\left|\sum_{i=1}^{2n}\frac{\xi_i B_\tau(X_i)(A_i-1/2)}{\sqrt{n}}\right| \leq o_p(1) \left(\frac{1}{2n}\sum_{i=1}^{2n}\xi_i\right) = o_p(1),
	\end{align*}
	and
	\begin{align*}
	\sup_{\tau \in \Upsilon}\left|\sum_{i=1}^{2n}\frac{\xi_i B_\tau(X_i) b^\top(X_i)  (\hat{\theta} - \theta_0)}{\sqrt{n}}\right| \leq & \sum_{i=1}^{2n}\frac{\xi_i \zeta(K)||\hat{\theta} - \theta_0||_2}{\sqrt{n}}o_p(1/\sqrt{n}) \\
	= & \left(\sum_{i =1}^{2n}\frac{\xi_i}{n}\right) o_p\left(\sqrt{\frac{K \zeta^2(K)\log(n)}{n}}\right) = o_p(1).
	\end{align*}
	This leads to the first result.
	
	For the second result, we have
	\begin{align*}
	\left|\sum_{i=1}^{2n}\frac{2\xi_i (A_i - 1/2)m_{1,\tau}(X_i) (\hat{A}_i-1/2)}{\sqrt{n}}\right| \leq \left\Vert\sum_{i=1}^{2n}\frac{2\xi_i (A_i - 1/2)m_{1,\tau}(X_i)b(X_i)}{\sqrt{n}}\right\Vert_2 ||\hat{\theta} - \theta_0||_2.
	\end{align*}
	In addition,
	\begin{align*}
	& \sup_{\tau \in \Upsilon}\left\Vert\sum_{i=1}^{2n}\frac{2\xi_i (A_i - 1/2)m_{1,\tau}(X_i)b(X_i)}{\sqrt{n}}\right\Vert_2 \\
	= & \sup_{\tau \in \Upsilon,\rho \in \Re^K, ||\rho||_2=1}\sum_{i=1}^{2n}\frac{2\xi_i (A_i - 1/2)m_{1,\tau}(X_i)b^\top(X_i)  \rho}{\sqrt{n}} \\
	= & \sup_{\tau \in \Upsilon,\rho \in \Re^K, ||\rho||_2=1}\sum_{j=1}^{n}\frac{(A_{\pi(2j-1)} - A_{\pi(2j)})( \xi_{\pi(2j-1)} m_{1,\tau}(X_{\pi(2j-1)})b^\top(X_{\pi(2j-1)}) - \xi_{\pi(2j)} m_{1,\tau}(X_{\pi(2j)})b^\top(X_{\pi(2j)}))\rho}{\sqrt{n}}.
	\end{align*}
	Conditional on $\{X_i,\xi_i\}_{i \in [2n]}$, $\{(A_{\pi(2j-1)} - A_{\pi(2j)})\}_{j=1}^n$ is a sequence of i.i.d. Rademacher random variables. In addition, let
	\begin{align*}
	\mathcal{F}_8 = \{( \xi_{\pi(2j-1)} m_{1,\tau}(X_{\pi(2j-1)})b^\top(X_{\pi(2j-1)}) - \xi_{\pi(2j)} m_{1,\tau}(X_{\pi(2j)})b^\top(X_{\pi(2j)}))\rho: \tau \in \Upsilon, \rho \in R^K, ||\rho||_2=1\}
	\end{align*}
	with envelope $F_j = (\xi_{\pi(2j-1)}\zeta(K)+\xi_{\pi(2j)}\zeta(K)) .$ Then, w.p.a.1,
	\begin{align*}
	& \mathbb{E}\frac{1}{n}\sum_{j=1}^n F_j^2 \leq \frac{1}{n}\sum_{i =1}^{2n} \mathbb{E}\xi_i^2 \zeta^2(K) \leq \overline{C} \zeta^2(K).
	\end{align*}
	In addition, for some constant $c>0$,
	\begin{align*}
	\sup_Q N(\mathcal{F}_8,e_Q,\eps||F||_{Q,2}) \leq \left(\frac{a}{\eps}\right)^{cK}, \quad \forall \eps \in (0,1].
	\end{align*}
	Let $\sigma_n^2 = \sup_{f \in \mathcal{F}_8} \mathbb{P}_n f^2$ and $\delta^2 =  \frac{\sigma_n^2}{\frac{1}{n}\sum_{j=1}^n F_j^2} \leq 1$. Then, by \citet[Corollary 2.2.8]{VW96}, \eqref{eq:J} and \eqref{eq:J2}, we have, w.p.a.1,
	\begin{align*}
	&\mathbb{E}\mathbb{E}\left[\sup_{\tau \in \Upsilon}\left\Vert\sum_{i=1}^{2n}\frac{2\xi_i (A_i - 1/2)m_{1,\tau}(X_i)b(X_i)}{\sqrt{n}}\right\Vert_2 \biggl|\{X_i,\xi_i\}_{i \in [2n]}\right] \\
	& \lesssim  \mathbb{E}\int_0^{\sigma_n}\sqrt{1+\log(N(\mathcal{F}_8,e_{\mathbb{P}_n},\eps)) }d\eps \\
	& \lesssim  \mathbb{E}\sqrt{\frac{1}{n}\sum_{j=1}^n F_j^2} \int_0^\delta \sqrt{1+\log \sup_Q N(\mathcal{F}_8,e_Q,\eps||F||_{Q,2})} d\eps \\
	& \leq  \left(\mathbb{E}\sqrt{\frac{1}{n}\sum_{j=1}^n F_j^2}\right)J(1) \\
	& \lesssim  \left(\sqrt{\mathbb{E}\frac{1}{n}\sum_{j=1}^n F_j^2}\right) \sqrt{K} \\
	& \lesssim  \sqrt{K}\zeta(K).
	\end{align*}
	
	This implies
	\begin{align*}
	\sup_{\tau \in \Upsilon}\left\Vert\sum_{i=1}^{2n}\frac{2\xi_i (A_i - 1/2)m_{1,\tau}(X_i)b(X_i)}{\sqrt{n}}\right\Vert_2 = O_p(\sqrt{K}\zeta(K)),
	\end{align*}
	and
	\begin{align*}
	\sup_{\tau \in \Upsilon}\left|\left|\sum_{i=1}^{2n}\frac{2\xi_i (A_i - 1/2)m_{1,\tau}(X_i) (\hat{A}_i-1/2)}{\sqrt{n}}\right| \right|_2 = O_p\left(\sqrt{\frac{K^2 \zeta^2(K)\log(n)}{n}}\right) = o_p(1).
	\end{align*}
	
	Last, for the third result, we have
	\begin{align}
	\label{eq:sieve23}
	\sup_{\tau \in \Upsilon}\left|\sum_{i=1}^{2n}\frac{2\xi_iA_i \eta_{1,i}(\tau) (\hat{A}_i-1/2)}{\sqrt{n}}\right| \leq &  \sup_{\tau \in \Upsilon}\left\Vert\sum_{i=1}^{2n}\frac{2\xi_iA_i \eta_{1,i}(\tau) b(X_i)}{\sqrt{n}} \right\Vert_2 ||\hat{\theta}-\theta_0||_2 \notag \\
	\leq & \sup_{\tau \in \Upsilon,\rho \in \Re^K, ||\rho||_2=1}\left[\sum_{i=1}^{2n}\frac{2\xi_iA_i \eta_{1,i}(\tau) b^\top(X_i)  \rho}{\sqrt{n}} \right] ||\hat{\theta}-\theta_0||_2.
	\end{align}
	Let $\{\tilde{\eps}_j\}_{j \in [n]}$ be a sequences of i.i.d. Rademacher random variables that is independent of the data. By \eqref{eq:=d}, we have
	\begin{align*}
	\sum_{i=1}^{2n}\frac{2\xi_iA_i \eta_{1,i}(\tau) b^\top(X_i)  \rho}{\sqrt{n}} \biggl |\{A_i,X_i\}_{i \in [2n]} \stackrel{d}{=}\sum_{j=1}^{n}\frac{2\tilde{\xi}_j \tilde{\eta}_{1,j}(\tau) b^\top(\tilde{X}_j)\rho}{\sqrt{n}} \biggl |\{\tilde{X}_j\}_{j \in [n]},
	\end{align*}
	and
	\begin{align*}
	\sum_{i=1}^{2n}\frac{2\eps_i\xi_iA_i \eta_{1,i}(\tau) b^\top(X_i)  \rho}{\sqrt{n}} \biggl |\{A_i,X_i\}_{i \in [2n]} \stackrel{d}{=}\sum_{j=1}^{n}\frac{2\tilde{\eps}_j \tilde{\xi}_j \tilde{\eta}_{1,j}(\tau) b^\top(\tilde{X}_j)\rho}{\sqrt{n}} \biggl |\{\tilde{X}_j\}_{j \in [n]},
	\end{align*}
	where conditionally on $\{\tilde{X}_j\}_{j \in [n]}$, $\{\tilde{\xi}_j \tilde{\eta}_{1,j}(\tau)\}_{j \in [n]}$ is a sequence of independent random variables.  Then, by the same argument as in \eqref{eq:II2}, we have
	\begin{align}
	& \mathbb{E} \sup_{\tau \in \Upsilon,\rho \in \Re^K, ||\rho||_2=1}\left[\sum_{i=1}^{2n}\frac{2 \xi_iA_i \eta_{1,i}(\tau) b^\top(X_i)  \rho}{\sqrt{n}}\biggl|\{X_i,A_i\}_{i \in [2n]} \right]   \notag \\
	&= \mathbb{E} \sup_{\tau \in \Upsilon,\rho \in \Re^K, ||\rho||_2=1}\left[\sum_{j=1}^{n}\frac{2\tilde{\xi}_j \tilde{\eta}_{1,j}(\tau) b^\top(\tilde{X}_j)\rho}{\sqrt{n}} \biggl |\{\tilde{X}_j\}_{j \in [n]} \right]  \notag \\
	&\lesssim \mathbb{E} \sup_{\tau \in \Upsilon,\rho \in \Re^K, ||\rho||_2=1}\left[\sum_{j=1}^{n}\frac{2\tilde{\eps}_j\tilde{\xi}_j \tilde{\eta}_{1,j}(\tau) b^\top(\tilde{X}_j)\rho}{\sqrt{n}} \biggl |\{\tilde{X}_j\}_{j \in [n]} \right]  \notag \\
	& = \mathbb{E}\left\{  \mathbb{E} \sup_{\tau \in \Upsilon,\rho \in \Re^K, ||\rho||_2=1}\left[\sum_{j=1}^{n}\frac{2\tilde{\eps}_j\tilde{\xi}_j \tilde{\eta}_{1,j}(\tau) b^\top(\tilde{X}_j)\rho}{\sqrt{n}} \biggl |\{\tilde{Y}_{j}(1),\tilde{\xi}_j, \tilde{X}_j\}_{j \in [n]} \right] \biggl| \{\tilde{X}_j\}_{j \in [n]}  \right\}.
	\label{eq:sieve1}
	\end{align}
	
	Let
	\begin{align*}
	\mathcal{F}_9 = \{2\tilde{\xi}_j (\tau - 1\{\tilde{Y}_j(1) \leq q_1(\tau) \} - m_{1,\tau}(\tilde{X}_j)) b^\top(\tilde{X}_j)\rho: \tau \in \Upsilon,\rho \in \Re^K, ||\rho||_2=1\},
	\end{align*}
	with envelope $F_j =2\tilde{\xi}_j \zeta(K)$. In addition, for some constant $c>0$,
	\begin{align*}
	\sup_Q N(\mathcal{F}_9,e_Q,\eps||F||_{Q,2}) \leq \left(\frac{a}{\eps}\right)^{cK}, \quad \forall \eps \in (0,1].
	\end{align*}
	Then, following \eqref{eq:Z} and \eqref{eq:J2}, we have
	\begin{align}
	& \mathbb{E} \sup_{\tau \in \Upsilon,\rho \in \Re^K, ||\rho||_2=1}\left[\sum_{j=1}^{n}\frac{2\tilde{\eps}_j\tilde{\xi}_j \tilde{\eta}_{1,j}(\tau) b^\top(\tilde{X}_j)\rho}{\sqrt{n}} \biggl |\{\tilde{Y}_{j}(1),\tilde{\xi}_j, \tilde{X}_j\}_{j \in [n]} \right] \notag \\
	\lesssim & ||F||_{\mathbb{P}_n,2} J(1) \lesssim \sqrt{K}\zeta(K) \left(\frac{1}{n}\sum_{j \in [n]} \tilde{\xi}_j^2 \right)^{1/2}.
	\label{eq:sieve2}
	\end{align}
	Combining \eqref{eq:sieve1} and \eqref{eq:sieve2}, we have
	\begin{align*}
	& \mathbb{E} \sup_{\tau \in \Upsilon,\rho \in \Re^K, ||\rho||_2=1}\left[\sum_{i=1}^{2n}\frac{2 \xi_iA_i \eta_{1,i}(\tau) b^\top(X_i)  \rho}{\sqrt{n}}\biggl|\{X_i,A_i\}_{i \in [2n]} \right]\\
	\lesssim & \mathbb{E}\left\{  \mathbb{E} \sup_{\tau \in \Upsilon,\rho \in \Re^K, ||\rho||_2=1}\left[\sum_{j=1}^{n}\frac{2\tilde{\eps}_j\tilde{\xi}_j \tilde{\eta}_{1,j}(\tau) b^\top(\tilde{X}_j)\rho}{\sqrt{n}} \biggl |\{\tilde{Y}_{j}(1),\tilde{\xi}_j, \tilde{X}_j\}_{j \in [n]} \right] \biggl| \{\tilde{X}_j\}_{j \in [n]}  \right\} \\
	\lesssim & \mathbb{E}\left\{\sqrt{K}\zeta(K)\left(\frac{1}{n}\sum_{j \in [n]} \tilde{\xi}_j^2 \right)^{1/2}\biggl| \{\tilde{X}_j\}_{j \in [n]}\right\} \\
	= & \sqrt{K}\zeta(K)\mathbb{E}\left\{\left(\frac{1}{n}\sum_{i \in [2n]}A_i \xi_i^2 \right)^{1/2}\biggl| \{A_i,X_i\}_{i \in [2n]}\right\} \\
	\lesssim & \sqrt{K}\zeta(K)\mathbb{E}\left\{\left(\frac{1}{n}\sum_{i \in [2n]}\xi_i^2 \right)^{1/2}\biggl| \{A_i,X_i\}_{i \in [2n]}\right\} \\
	= & \sqrt{K}\zeta(K)\mathbb{E} \left(\frac{1}{n}\sum_{i \in [2n]}\xi_i^2 \right)^{1/2} \\
	\lesssim  & \sqrt{K}\zeta(K),
	\end{align*}
	and thus,
	\begin{align*}
	\sup_{\tau \in \Upsilon,\rho \in \Re^K, ||\rho||_2=1}\sum_{i=1}^{2n}\frac{2\xi_iA_i \eta_{1,i}(\tau) b^\top(X_i)  \rho}{\sqrt{n}}  = O_p(\sqrt{K}\zeta(K)).
	\end{align*}
	Then, by \eqref{eq:sieve23} and Lemma \ref{lem:sieve}, we have
	\begin{align*}
	\sup_{\tau \in \Upsilon}\left|\sum_{i=1}^{2n}\frac{2\xi_iA_i \eta_{1,i}(\tau) (\hat{A}_i-1/2)}{\sqrt{n}}\right| = O_p\left(\sqrt{\frac{K^2\zeta^2(K)\log(n)}{n}}\right) = o_p(1).
	\end{align*}
	
\end{proof}

\section{IPW Multiplier Bootstrap with Cross-Validation}
\label{sec:CV}

We implement sieve estimation in the IPW multiplier bootstrap. A practical issue in implementation is the choice of the sieve basis functions. In this section, we first propose a leave-one-out cross-validation (CV) method to choose the sieve basis functions. We then provide computational details and report the performance of the CV method in both the simulated and real datasets used in the paper.

Our leave-one-out CV method uses a procedure similar to that described in \cite{hansen2014}. Although the sieve regression is implemented to estimate the propensity score, scrutiny of the proof reveals that the sieve basis functions are actually used to project $\mathbb{P}(Y(a) \leq q_a(\tau)|X=x)$. For this reason, we propose to choose the sieve basis functions by cross-validating the linear regression of $1\{Y_i \leq \hat{q}_a(\tau)\}$ on the basis functions for the subsample of $\{i \in [2n]: A_i = a\}$, where $\hat{q}_1(\tau)$ and $\hat{q}_0(\tau)$ are the $\tau$th percentiles of the outcomes in the treated and control groups as defined in Section
\ref{sec:est}, respectively.\footnote{In theory, cross-validating the regression of $A_i$ on the basis functions may not work. This is because we know that $\mathbb{E}(A_i|X_i) = 0.5$. As long as the basis functions contain an intercept term the linear regression of $A_i$ on the basis functions is correctly specified with the intercept being $0.5$ and the rest of the linear coefficients being zero. Therefore, the CV will favor the most parsimonious model which only contains the intercept term, i.e., $K=1$. This choice of $K$ clearly violates Assumption 6(v) in the paper.}

\subsection{Computation}
\label{subsec:CV_procedure}

The procedure to compute the IPW multiplier bootstrap estimator with CV begins with selecting the sieve terms from a class of basis functions. After that, we estimate the propensity score $\hat{A}_{i}$ based on the selected sieve basis functions and obtain the IPW multiplier bootstrap estimators of the QTE and ATE as in Section \ref{subsec:IPW computation}. Suppose the sets of sieve basis functions are indexed by $m \in \mathcal{M}$, where $\mathcal{M}$ denotes the class of models to be selected.

\begin{enumerate}
	\item Denote the  basis functions for model $m \in \mathcal{M}$ as $b_{m}\left(X_{i}\right)=\left(b_{i,1},\dots,b_{i,K}\right)^{\prime}$ for every $X_{i}$, we compute the CV criterion for both the treatment and control groups.
	
	\begin{enumerate}
		\item For $a=0,1$, compute $\hat{q}_{a}\left(\tau\right)$ by
		\[
		\hat{q}_{a}\left(\tau\right)=\underset{q}{\arg \min } \sum_{j=1}^{n} \rho_{\tau}\left(Y_{(j,a)}-q\right),~ \text{for
			a=0,1},
		\]
		where $\left\{ Y_{(j,1)}\right\} _{j \in [n]}$ and $\left\{ Y_{(j,0)}\right\} _{j \in [n]}$ are the observed outcomes for the treatment and control groups, respectively.
		\item Construct the dependent variable for CV,
		\[
		D_{(j,a)}=1\left\{ Y_{(j,a)}\leq\hat{q}_{a}\left(\tau\right)\right\} ,\text{ for }j \in [n],a=0,1.
		\]

		\item Obtain the CV criterion for the given basis functions
		\[
		CV^{a}_{n}\left(m\right)=\frac{1}{n}\sum_{j \in [n]}\frac{\hat{e}_{mj}^{a}}{\left(1-h_{mj}^{a}\right)^{2}}.
		\]
		where $\hat{e}_{mj}^{a}=D_{(j,a)}-\hat{D}_{(j,a)}$, $h_{mj}^{a}=b_{m}\left(X_{(j,a)}\right)^{\top}\left(\sum_{j \in [n]}b_{m}\left(X_{(j,a)}\right)b_{m}\left(X_{(j,a)}\right)^{\top}\right)^{-1}b_{m}\left(X_{(j,a)}\right)$, $\{X_{(j,1)}\}_{j \in [n]}$ and $\{X_{(j,0)}\}_{j \in [n]}$ are  the covariates for the treatment and control groups, respectively, $D_{(j,a)}$ is the variable constructed in Step (b), and
		$\hat{D}_{(j,a)}=b_{m}\left(X_{(j,a)}\right)^{\top}\hat{\beta}^{a,CV}$ with  $$\hat{\beta}^{a,CV}=\left(\sum_{j \in [n]}b_{m}\left(X_{(j,a)}\right)b_{m}\left(X_{(j,a)}\right)^{\top}\right)^{-1}\sum_{j \in [n]}b_{m}\left(X_{(j,a)}\right)D_{(j,a)}.$$
	\end{enumerate}

	\item For both the treatment and control groups, choose the basis functions that give the smallest value of $CV^{a}_{n}\left(m\right)$ as
	\[
	\hat{m}^{a} = \argmin_{m \in \mathcal{M}}{CV_{n}^{a}(m)},~a=0,1.
	\]

	\item Using the basis functions selected for the treatment and control groups, we estimate the propensity score $\hat{A}_{i}$. Specifically, for $i \in [2n]$, let
	
	\[
	\hat{A}_{i}=A_i b_{\hat{m}^1}\left(X_{i}\right)^{\top}\hat{\theta}^{1} + (1-A_i) b_{\hat{m}^0}\left(X_{i}\right)^{\top}\hat{\theta}^{0},
	\]
	where for $a=0,1$,  $\hat{\theta}^{a}=\arg\min_{\theta}\sum_{i=1}^{2n}\xi_{i}\left(A_{i}-b_{\hat{m}^a}\left(X_{i}\right)^{\top}\theta\right)^{2}$.\footnote{Note that $\hat{\theta}^{a}$ is estimated by using all the data in the treatment and control groups.}
	
	\item The IPW multiplier bootstrap estimator $\hat{q}^w_{ipw}(\tau)$ with CV is obtained as in Section \ref{subsec:IPW computation} of the paper by plugging $\{ \hat{A}_i\}_{i \in [2n]}$ obtained in Step 3.
	
	\item We repeat the above procedure for $\tau \in \mathcal{G}$.
\end{enumerate}

To calculate the IPW multiplier bootstrap estimator of the ATE with CV, we use a similar procedure with one modification specific to the ATE: in Step 1, we can directly calculate the CV criterion function with
\begin{align*}
& \hat{e}_{mj}^{a}=Y_{(j,a)}-\hat{Y}_{(j,a)} \\
& h_{mj}^{a}=b_{m}\left(X_{(j,a)}\right)^{\top}\left(\sum_{j \in [n]}b_{m}\left(X_{(j,a)}\right)b_{m}\left(X_{(j,a)}\right)^{\top}\right)^{-1}b_{m}\left(X_{(j,a)}\right) \\
& \hat{Y}_{(j,a)}=b_{m}\left(X_{(j,a)}\right)^{\top}\hat{\beta}^{a,CV}, \quad \text{and}, \\
& \hat{\beta}^{a,CV}=\left(\sum_{j \in [n]}b_{m}\left(X_{(j,a)}\right)b_{m}\left(X_{(j,a)}\right)^{\top}\right)^{-1}\sum_{j \in [n]}b_{m}\left(X_{(j,a)}\right)Y_{(j,a)}.
\end{align*}

\subsection{Simulations}
\label{subsec:CV_simul}

We assess the finite sample performance of the CV method discussed in Section \ref{subsec:CV_procedure} using the same set of DGPs considered in Section \ref{sec:sim} of the paper.\footnote{Let $qx_{i,\alpha}$ be the $\alpha$-th quantile of $X_i$. When $d_{x}=1$ , the sets of basis functions under selection by the CV procedure are
	\begin{flalign*}
	& \left\{ 1,X,\max\left\{X-qx_{0.5},0\right\} \right\}, &&\\
	& \left\{ 1,X,\max\left\{X-qx_{0.3},0\right\} ,\max\left\{ X-qx_{0.7},0\right\} \right\}, &&\\
	& \left\{ 1,X,X^2,\max\left\{X-qx_{0.5},0\right\}^{2} \right\}, &&\\
	& \left\{ 1,X,X^{2},\max\left\{X-qx_{0.3},0\right\} ^{2},\max\left\{X-qx_{0.7},0\right\} ^{2}\right\}.&&
	\end{flalign*}
	When $d_{x}=2$, the sets of basis functions under selection by the CV procedure are
	\begin{flalign*}
	& \left\{ 1,X_{1},X_{2},X_1X_2,\max\left\{ X_1 -qx_{1,0.5},0\right\} ,\max\left\{ X_2 -qx_{2,0.5},0\right\} \right\}, &&\\
	& \left\{ 1,X_{1},X_{2},X_1X_2,\max\left\{ X_1 -qx_{1,0.3},0\right\} , \max\left\{ X_1 -qx_{1,0.7},0\right\},\max\left\{ X_2 -qx_{2,0.5},0\right\},\max\left\{ X_2 -qx_{2,0.7},0\right\} \right\} , &&\\
	& \left\{ 1,X_{1},X_{2},X_1X_2,X_1^2,X_2^2,\max\left\{ X_1 -qx_{1,0.5},0\right\}^2 ,\max\left\{ X_2 -qx_{2,0.5},0\right\}^2 \right\}, &&\\
	& \left\{ 1,X_{1},X_{2},X_1X_2,X_1^2,X_2^2,\max\left\{ X_1 -qx_{1,0.3},0\right\}^2 , \max\left\{ X_1 -qx_{1,0.7},0\right\}^2,\max\left\{ X_2 -qx_{2,0.5},0\right\}^2,\max\left\{ X_2 -qx_{2,0.7},0\right\}^2 \right\}.&&
	\end{flalign*}
}

Tables \ref{tab:sim_ate_IPWCV}--\ref{tab:sim_uniform_IPWCV} report the results corresponding to those in Tables \ref{tab:sim_ate}--\ref{tab:sim_uniform} in the paper,  respectively. For a direct comparison, Tables \ref{tab:sim_ate_IPWCV}--\ref{tab:sim_uniform_IPWCV} only contain the results for the IPW multiplier bootstrap with and without the CV. In general, the simulation shows that the performance of the two bootstrap methods is similar.  For the empirical size and power of the uniform confidence bands, Table \ref{tab:sim_uniform_IPWCV} shows that the performance of the bootstrap method with CV is similar to that without for models 1 and 2. For models 3 and 4, the bootstrap method with CV performs poorly when the number of pairs is as small as 50. When the number of pairs increases to 100, its performance improves and is close to that of the IPW multiplier bootstrap method used in the paper, in which the sieve basis functions are selected in an ad-hoc manner. We conjecture that the poor performance for $n=50$ is due to the model selection bias. As we do not change the size of $\mathcal{M}$, the model selection bias becomes asymptotically negligible as the sample size increases. Therefore, the performance of the uniform confidence band improves when $n$ increases to 100.

\newcolumntype{L}{>{\raggedright\arraybackslash}X}
\newcolumntype{C}{>{\centering\arraybackslash}X}

\begin{table}[ht]
	\caption{The Empirical Size and Power of Tests for ATEs}
	\vspace{1ex}
	\begin{adjustbox}{max width=\textwidth}
		\centering{}%
		\begin{tabular}{l ccccc ccccc ccccc ccccc}
			\hline
			\hline
			\multirow{3}{*}{Model} & \multicolumn{4}{c}{$\mathcal{H}_{0}$: $\Delta=0$} & \multicolumn{4}{c}{$\mathcal{H}_{1}$: $\Delta=1/2$}\\
			& \multicolumn{2}{c}{$n=50$} & \multicolumn{2}{c}{$n=100$} & \multicolumn{2}{c}{$n=50$} & \multicolumn{2}{c}{$n=100$}\\
			\cmidrule(lr){2-3} \cmidrule(lr){4-5} \cmidrule(lr){6-7}  \cmidrule(lr){8-9}
			&  IPW  & IPW with CV & IPW  & IPW with CV & IPW  & IPW with CV & IPW  & IPW with CV\\
			\hline
			1 & 5.44 & 5.65 & 6.00 & 5.92 & 29.44 & 29.73 & 50.46 & 50.35 \\
			2 & 5.59 & 5.81 & 5.89 & 5.76 & 24.24 & 24.25 & 41.68 & 41.46 \\
			3 & 4.92 & 4.50 & 5.16 & 5.16 & 5.91 & 5.41 & 8.74 & 8.47 \\
			4 & 6.01 & 5.57 & 4.74 & 4.73 & 4.35 & 4.07 & 4.89 & 4.85 \\
			\hline
		\end{tabular}
	\end{adjustbox}
	\vspace{-1ex}
	\justify
	Notes: The table presents the rejection probabilities for tests of ATEs. The columns  `IPW'  and `IPW with CV'  correspond to the $t$-test using the standard errors estimated by the IPW multiplier bootstrap reported in the main part of the paper and the one with CV, respectively.
	\label{tab:sim_ate_IPWCV}
\end{table}

\begin{landscape}
	\begin{table}[htbp]
		\caption{The Empirical Size and Power of Tests for QTEs }
		\vspace{1ex}

		\centering{}
		\begin{threeparttable}
			\begin{tabular}{lcccccccccccccccc}
				\hline
				\hline
				\multicolumn{1}{l}{} & \multicolumn{8}{c}{$\mathcal{H}_{0}$: $\Delta=0$} & \multicolumn{8}{c}{$\mathcal{H}_{1}$: $\Delta=1/2$}\\
				\multicolumn{1}{l}{} & \multicolumn{4}{c}{$n=50$} & \multicolumn{4}{c}{$n=100$} & \multicolumn{4}{c}{$n=50$} & \multicolumn{4}{c}{$n=100$}\\
				\cmidrule(lr){2-5} \cmidrule(lr){6-9} \cmidrule(lr){10-13}  \cmidrule(lr){14-17}
				\multicolumn{1}{l}{} & 0.25 & 0.50 & 0.75 & \multicolumn{1}{c}{Dif} & 0.25 & 0.50 & 0.75 & \multicolumn{1}{c}{Dif} & 0.25 & 0.50 & 0.75 & \multicolumn{1}{c}{Dif} & 0.25 & 0.50 & 0.75 & Dif\\
				\hline
				\emph{Model 1} &  &  &  &  &  &  &  &  &  &  &  &  &  &  &  & \\
				IPW with CV & 5.47 & 5.37 & 6.16 & 4.19 & 5.23 & 5.89 & 5.67 & 4.11 & 24.51 & 13.76 & 12.20 & 8.49 & 43.93 & 23.25 & 17.19 & 14.09 \\
				IPW & 5.47 & 5.31 & 6.17 & 4.24 & 5.26 & 5.83 & 5.65 & 3.95 & 24.81 & 13.48 & 12.12 & 8.40 & 43.93 & 23.33 & 17.21 & 13.91 \\
				&  &  &  &  &  &  &  &  &  &  &  &  &  &  &  & \\
				\emph{Model 2} &  &  &  &  &  &  &  &  &  &  &  &  &  &  &  & \\
				IPW with CV & 4.99 & 5.11 & 5.79 & 4.47 & 5.15 & 5.77 & 5.95 & 4.16 & 20.43 & 13.11 & 10.51 & 7.50 & 36.24 & 21.62 & 15.24 & 12.46 \\
				IPW & 4.93 & 5.12 & 5.78 & 4.45 & 5.17 & 5.73 & 5.88 & 4.00 & 20.29 & 12.90 & 10.40 & 7.35 & 36.38 & 21.53 & 15.14 & 12.53 \\
				&  &  &  &  &  &  &  &  &  &  &  &  &  &  &  & \\
				\emph{Model 3} &  &  &  &  &  &  &  &  &  &  &  &  &  &  &  & \\
				IPW with CV & 4.55 & 3.04 & 5.17 & 2.18 & 4.76 & 3.61 & 4.79 & 2.91 & 8.44 & 7.30 & 4.94 & 2.70 & 13.14 & 15.49 & 6.05 & 4.14 \\
				IPW & 4.76 & 3.19 & 5.61 & 2.60 & 4.77 & 3.71 & 4.95 & 3.02 & 8.75 & 7.81 & 5.35 & 3.09 & 13.04 & 15.42 & 6.06 & 4.21 \\
				&  &  &  &  &  &  &  &  &  &  &  &  &  &  &  & \\
				\emph{Model 4} &  &  &  &  &  &  &  &  &  &  &  &  &  &  &  & \\
				IPW with CV & 3.57 & 3.34 & 4.43 & 2.96 & 4.29 & 4.60 & 5.11 & 3.60 & 7.44 & 4.64 & 4.08 & 2.63 & 13.91 & 8.57 & 5.38 & 3.78 \\
				IPW & 3.97 & 3.97 & 4.91 & 3.68 & 4.23 & 4.51 & 5.01 & 3.48 & 8.08 & 5.37 & 4.79 & 3.26 & 13.50 & 8.33 & 5.17 & 3.51 \\
				\hline
			\end{tabular}
			
			\begin{tablenotes}[para,flushleft]
				Note:	The table presents the rejection probabilities for tests of QTEs. The columns `0.25', `0.50', and `0.75' correspond to tests with quantiles at 25\%, 50\%, and 75\%, respectively; the column `Dif' corresponds to the test with the null hypothesis specified in \eqref{eq:diff} in the paper. The rows `IPW'  and `IPW with CV' correspond to the results of the IPW multiplier bootstrap reported in the main part of the paper and the one with CV, respectively.
			\end{tablenotes}
			
			\label{tab:sim_qte_IPWCV}
		\end{threeparttable}
		
	\end{table}
	
\end{landscape}

\begin{table}[ht]
	\caption{The Empirical Size and Power of Uniform Inferences for QTEs}
	\vspace{1ex}
	
	\centering{}%
	\begin{tabularx}{1\textwidth}{LCCCC}
		\hline
		\hline
		& \multicolumn{2}{c}{$\mathcal{H}_{0}$: $\Delta=0$} & \multicolumn{2}{c}{$\mathcal{H}_{1}$: $\Delta=1/2$}\\
		\cmidrule(lr){2-3} \cmidrule(lr){4-5}
		& \multicolumn{1}{c}{$n=50$} & \multicolumn{1}{c}{$n=100$} & \multicolumn{1}{c}{$n=50$} & \multicolumn{1}{c}{$n=100$}\\
		\hline
		\emph{Model 1} &  &  &  & \\
		IPW with CV & 4.03 & 4.94 & 15.30 & 32.09 \\
		IPW & 4.49 & 4.94 & 16.30 & 32.40 \\
		&  &  &  & \\
		\emph{Model 2} &  &  &  & \\
		IPW with CV & 3.98 & 4.79 & 13.34 & 27.29 \\
		IPW & 4.25 & 4.91 & 14.27 & 27.47 \\
		&  &  &  & \\
		\emph{Model 3} &  &  &  & \\
		IPW with CV & 1.22 & 2.55 & 2.33 & 10.25 \\
		IPW & 2.19 & 2.99 & 4.25 & 11.34 \\
		&  &  &  & \\
		\emph{Model 4} &  &  &  & \\
		IPW with CV & 0.99 & 2.95 & 0.91 & 6.20 \\
		IPW & 2.78 & 3.36 & 3.18 & 6.98 \\
		\hline
	\end{tabularx}
	
	\vspace{-1ex}
	\justify
	Notes: The table presents the rejection probabilities of the uniform confidence bands for the hypothesis specified in \eqref{eq:uni} in the paper. The rows `IPW'  and `IPW with CV' correspond to the results of the IPW multiplier bootstrap reported in the main part of the paper and the one with CV, respectively.
	
	\label{tab:sim_uniform_IPWCV}
\end{table}

\subsection{Empirical Application}
\label{subsec:CV_app}

We revisit the empirical application in Section \ref{sec:app} with the CV method discussed in Section \ref{subsec:CV_procedure}.\footnote{Let $qx_{i,\alpha}$ be the $\alpha$-th quantile of $X_i$ and use the same set of the variables $X$s and DV in the paper. The sets of basis functions under selection by the CV procedure are
	\begin{align*}
	& \left\{ 1,X_1X_2,X_2X_3,X_1X_3,\max\left\{ X_1 -qx_{1,0.5},0\right\},\max\left\{ X_2 -qx_{2,0.5},0\right\},\max\left\{X_3 -qx_{3,0.5},0\right\}, DV\right\}, \\
	& \left\{ 1,X_1X_2,X_2X_3,X_1X_3,\max\left\{ X_1 -qx_{1,0.5},0\right\}^2,\max\left\{ X_2 -qx_{2,0.5},0\right\}^2,\max\left\{ X_3 -qx_{3,0.5},0\right\}^2,DV \right\}, \\
	& \left\{ 1,X_1,X_2,X_3,X_1X_2,X_2X_3,X_1X_3,\max\left\{ X_1 -qx_{1,0.5},0\right\},\max\left\{ X_2 -qx_{2,0.5},0\right\},\max\left\{X_3 -qx_{3,0.5},0\right\}, DV\right\}, \\
	& \left\{ 1,X_1,X_2,X_3,X_1X_2,X_2X_3,X_1X_3,\max\left\{ X_1 -qx_{1,0.5},0\right\}^2,\max\left\{ X_2 -qx_{2,0.5},0\right\}^2,\max\left\{ X_3 -qx_{3,0.5},0\right\}^2,DV \right\}, \\
	& \left\{ 1,X_1X_2,X_2X_3,X_1X_3,\max\left\{ X_1 -qx_{1,0.3},0\right\},\max\left\{ X_1 -qx_{1,0.5},0\right\},\max\left\{ X_2 -qx_{2,0.3},0\right\},\max\left\{ X_2 -qx_{2,0.5},0\right\},DV \right\},\\
	& \begin{Bmatrix}
	1,X_1X_2,X_2X_3,X_1X_3,\max\left\{ X_1 -qx_{1,0.3},0\right\}^2,\max\left\{ X_1 -qx_{1,0.5},0\right\}^2,\max\left\{ X_2 -qx_{2,0.3},0\right\}^2,\max\left\{ X_2 -qx_{2,0.5},0\right\}^2, \\
	\max\left\{ X_3 -qx_{3,0.3},0\right\}^2,\max\left\{ X_3 -qx_{3,0.5},0\right\}^2,DV
	\end{Bmatrix},\\
	& \begin{Bmatrix}
	1,X_1,X_2,X_3,X_1X_2,X_2X_3,X_1X_3,\max\left\{ X_1 -qx_{1,0.3},0\right\},\max\left\{ X_1 -qx_{1,0.5},0\right\},\max\left\{ X_2 -qx_{2,0.3},0\right\},\max\left\{ X_2 -qx_{2,0.5},0\right\}, \\ \max\left\{X_3 -qx_{3,0.3},0\right\},\max\left\{X_3 -qx_{3,0.5},0\right\}, DV
	\end{Bmatrix},\\
	& \begin{Bmatrix}
	1,X_1,X_2,X_3,X_1X_2,X_2X_3,X_1X_3,\max\left\{ X_1 -qx_{1,0.3},0\right\}^2,\max\left\{ X_1 -qx_{1,0.5},0\right\}^2,\max\left\{ X_2 -qx_{2,0.3},0\right\}^2,\max\left\{ X_2 -qx_{2,0.5},0\right\}^2, \\
	\max\left\{ X_3 -qx_{3,0.3},0\right\}^2,\max\left\{ X_3 -qx_{3,0.5},0\right\}^2,DV
	\end{Bmatrix}.
	\end{align*}
} Tables \ref{tab:app_ate_IPWCV}--\ref{tab:app_qte_dif_IPWCV} report the results corresponding to those in Tables \ref{tab:emp_ate}--\ref{tab:emp_qte_dif} in the paper,  respectively. For a direct comparison, Tables \ref{tab:app_ate_IPWCV}--\ref{tab:app_qte_dif_IPWCV} only contain the results for the IPW multiplier bootstrap method with and without the CV.

In general, the standard errors estimated by the bootstrap method with CV are slightly smaller than those without. All the empirical implications from Tables \ref{tab:app_ate_IPWCV}--\ref{tab:app_qte_dif_IPWCV} in the paper remain valid if we use the IPW multiplier bootstrap with CV instead.

\begin{table}[ht]
	\caption{ATEs of Macroinsurance on Consumption and Profits}
	\vspace{1ex}
	\centering{}%
	\begin{tabularx}{1\textwidth}{LCC}
		\hline
		\hline
		& IPW & IPW with CV \\
		\hline
		Consumption & -1.59(29.45) & -1.59(29.28) \\
		Profit & -86.68(45.38) & -86.68(44.33) \\
		\hline
	\end{tabularx} \\
	
	\vspace{-1ex}
	\justify
	Notes: The table presents the ATE estimates of the effect of macroinsurance on the monthly consumption and profits. Standard errors are in parentheses.  The columns  `IPW'  and `IPW with CV' correspond to the $t$-test using the standard errors estimated by the IPW multiplier bootstrap  ATE estimator reported in the main part of the paper and the one with CV, respectively.
	
	\label{tab:app_ate_IPWCV}
\end{table}

\begin{table}[ht]
	\caption{QTEs of Macroinsurance on Consumption and Profits}
	\vspace{1ex}
	\centering{}%
	\begin{tabularx}{1\textwidth}{LCC}
		\hline
		\hline
		& IPW & IPW with CV  \\
		\hline
		\textit{Panel A. Consumption} & & \\
		25\% & -14.33(25.89) & -14.33(25.26) \\
		50\% & -4.50(31.80) & -4.50(31.29) \\
		75\% & -22.17(56.08) & -22.17(54.98) \\
		& & \\
		\textit{Panel B. Profit } & &\\
		25\% & -33.33(25.94) & -33.33(25.51) \\
		50\% & -66.67(51.02) & -66.67(51.02) \\
		75\% & -200.00(85.04) & -200.00(85.04) \\
		\hline
	\end{tabularx} \\
	
	\vspace{-1ex}
	\justify
	Notes: The table presents the QTE estimates of the effect of macroinsurance on the monthly consumption and profits at quantiles 25\%, 50\%, and 75\%. Standard errors are in parentheses. The rows `IPW'  and `IPW with CV' correspond to the results of the IPW multiplier bootstrap reported in the main part of the paper and the one with CV, respectively.
	\label{tab:app_qte_IPWCV}
\end{table}

\begin{table}[ht]
	\centering
	\caption{Tests for the Difference between Two QTEs of Macroinsurance}
	\vspace{1ex}
	\begin{tabularx}{1\textwidth}{LCC}
		\hline
		\hline
		& IPW & IPW with CV  \\
		\hline
		\textit{Panel A. Consumption} & & \\
		50\%-25\% & 9.83(29.87) & 9.83(30.12) \\
		75\%-50\% & -17.67(48.58) & -17.67(48.47) \\
		75\%-25\% & -7.83(55.87) & -7.83(55.44) \\
		& & \\
		\textit{Panel B. Profit } & &\\
		50\%-25\% & -33.33(49.32) & -33.33(51.02) \\
		75\%-50\% & -133.33(85.04) & -133.33(85.04) \\
		75\%-25\% & -166.67(89.29) & -166.67(85.04) \\
		\hline
	\end{tabularx} \\
	\vspace{-1ex}
	\justify
	Notes: The table presents tests for the difference between two QTEs of macroinsurance on the monthly consumption and profits. Standard errors are in parentheses. The rows `IPW'  and `IPW with CV' correspond to the results of the IPW multiplier bootstrap  QTE estimator reported in the main part of the paper and the one with CV, respectively.
	\label{tab:app_qte_dif_IPWCV}
\end{table}

\subsection{Conclusion}
\label{subsec:CV_conclusion}

This section proposed a leave-one-out CV method to choose the sieve basis functions for the IPW multiplier bootstrap method. We provided details of its computation and compared the performance of the IPW multiplier bootstrap methods with and without CV in both the simulated and real datasets considered in the main paper. The findings show that the IPW multiplier bootstrap methods perform equally well with and without CV.

From a theory perspective using CV to select the basis functions is likely to induce some model selection bias. It is possible to reduce such a bias by using the orthogonal moment to estimate the QTE. The same idea is used in the literature of causal inference with high-dimensional data such as \cite{BCFH13}. In this case, we need to estimate other nuisance parameters such as $\mathbb{P}(Y \leq q_1(\tau)|A=1,X)$ and $\mathbb{P}(Y \leq q_0(\tau)|A=0,X)$ in addition to the propensity score. It is of interest to estimate the QTE via the orthogonal moment with a built-in model selection step, and derive the limit distribution and bootstrap validity of the moment-based estimator under MPD. A comprehensive study of this approach deserves an independent paper and is left for  future research.

\section{Additional Simulations with dim(X)=5}
\label{sec:add_sim}

We have conducted additional simulations with dim(X) = 5. In all cases, potential outcomes for $a\in\{0,1\}$ and $1\leq i\leq2n$ are generated according to
\begin{equation}
\ensuremath{Y_{i}(a)=\mu_{a}+m_{a}\left(X_{i}\right)+\sigma_{a}\left(X_{i}\right)\epsilon_{a,i}},~a=0,1,\label{eq:simulpart01}
\end{equation}
where $\mu_{a},m_{a}\left(X_{i}\right),\sigma_{a}\left(X_{i}\right)$,
and $\epsilon_{a,i}$ are specified as follows. In each of the specifications below, $n\in\{50,100\}$ and ($X_{i},\epsilon_{0,i},\epsilon_{1,i}$) are i.i.d. The number of replications is 10,000. For bootstrap replications we set $B=5,000$. To better illustrate the power of the tests, we add the  alternative hypothesis $\mathcal{H}_1: \Delta = 2$ besides $\mathcal{H}_1: \Delta = 1/2$.

\begin{itemize}
	\item 	\emph{Model 5}: Same as Model 4 in the paper, but with $d_{x}=5$. Specifically, $X_{i}=\left(\Phi\left(V_{i1}\right),\Phi\left(V_{i2}\right),\dots,\Phi\left(V_{id_{x}}\right)\right)^{\prime}$,
	where $\Phi(\cdot)$ is the standard normal c.d.f. and
	\[
	V_{i}=\left(V_{i1,}\dots,V_{id_{x}}\right)\sim N\left(\left(\begin{array}{c}
	0\\
	\vdots\\
	0
	\end{array}\right),\left(\begin{array}{llccc}
	1 & \rho & \rho & \cdots & \rho\\
	\rho & 1 & \rho & \ddots & \vdots\\
	\rho & \rho & 1 & \ddots & \rho\\
	\vdots & \ddots & \ddots & \ddots & \rho\\
	\rho & \cdots & \rho& \rho & 1
	\end{array}\right)\right)
	\]
	$m_{0}\left(X_{i}\right)= \gamma^{\prime}X_{i}-1, m_{1}\left(X_{i}\right)=m_{0}\left(X_{i}\right) + 10(\Phi^{-1}(X_{i1})\Phi^{-1}(X_{i2})-\rho); \epsilon_{d,i}\sim N(0,1)$
	for $d=0,1; \sigma_{0}\left(X_{i}\right)=\sigma_{0}=1$ and $\sigma_{1}\left(X_{i}\right)=\sigma_{1}=2$; $\rho=0.7$. We set $\gamma=\left(1,4,1,1,1\right)^{\prime}$.
	
	\item  \emph{Model 6}:
	Same as Model 5, but with $\rho = 0.2$.

\end{itemize}

Table \ref{tab:sim_ate'} collects the simulation results for ATEs. We make the following observations. First, both the two-sample $t$ test and the naive pair $t$-test have rejection probabilities under $\mathcal{H}_{0}$ far below the nominal level. Second, the adjusted $t$-test has rejection probability under $\mathcal{H}_{0}$ close to the nominal level and is not conservative. Third, the IPW $t$-test proposed in this paper has performance close to the adjusted $t$-test when $n=100$.\footnote{For both ATE and QTE estimations in this additional simulation exercise, we choose the following basis function when estimating the propensity score in the IPW multiplier bootstrap: $\{1,X_{1},X_{2},X_{3},X_{4},X_{5},\max(X_{1}-qx_{1,0.5},0), \max(X_{2}-qx_{2,0.5},0), \max(X_{3}-qx_{3,0.5},0), \max(X_{4}-qx_{4,0.5},0), \max(X_{5}-qx_{5,0.5},0)\}$. } When $n=50$, the IPW $t$-test performs poorly due to the curse of dimensionality in estimating the propensity score. However, its performance improves as we increase the sample size to $n=100$: under $\mathcal{H}_{0}$, the IPW $t$-test has rejection probability close to that of the  adjusted $t$-test; under $\mathcal{H}_{1}$, it is more powerful than the `naive' and `naive pair' methods and has power similar to the adjusted $t$-test.

Tables \ref{tab:sim_qte_size'}-\ref{tab:sim_qte_power'} report the empirical size and power of the tests for a single null hypothesis involving one or two quantile indexes. Consistent with the results in the main part of the paper, the test with standard errors estimated by two naive methods performs poorly in all cases. It is conservative under $\mathcal{H}_{0}$ and lacks power under $\mathcal{H}_{1}$. In contrast, the test using standard errors estimated by the gradient bootstrap or the IPW multiplier bootstrap has a rejection probability under $\mathcal{H}_{0}$ that is close to the nominal level when $n=100$.

Table \ref{tab:sim_uniform'} reports the empirical size and power of the uniform confidence bands. The test using standard errors estimated by the two naive methods has rejection probabilities under $\mathcal{H}_{0}$ far below the nominal level in both DGPs. In both models. the tests constructed based on the two proposed bootstrap methods are conservative when $n=50$. But their performance improves as the number of pairs increases to 100. In addition, under $\mathcal{H}_{1}$, the tests based on both the gradient and the IPW multiplier bootstraps are more powerful than those based on the naive methods when $n=100$.

\begin{table}[htbp]
	\caption{The Empirical Size and Power of Tests for ATEs }
	\vspace{1ex}

	\centering{}
	\begin{tabularx}{1\textwidth}{LCCCCCC}
		\hline
		\hline
		\multicolumn{1}{l}{} & \multicolumn{2}{c}{$\mathcal{H}_{0}$: $\Delta=0$} & \multicolumn{2}{c}{$\mathcal{H}_{1}$: $\Delta=1/2$} &
		\multicolumn{2}{c}{$\mathcal{H}_{1}$: $\Delta=2$}\\
		\cmidrule(lr){2-3} \cmidrule(lr){4-5} \cmidrule(lr){6-7}
		\multicolumn{1}{l}{} & \multicolumn{1}{c}{$n=50$} & \multicolumn{1}{c}{$n=100$} & \multicolumn{1}{c}{$n=50$} & \multicolumn{1}{c}{$n=100$} & \multicolumn{1}{c}{$n=50$}& $n=100$ \\
		\hline
		\emph{Model 5} &  &  &  &  &  & \\
		Naive & 2.52 & 1.71 & 1.35 & 0.93 & 7.45 & 26.28 \\
		Naive Pair & 2.89 & 1.88 & 1.62 & 1.12 & 9.67 & 27.88  \\
		Adj & 5.75 & 5.22 & 4.18 & 5.38 & 22.47 & 51.68 \\
		IPW & 1.69 & 4.40 & 1.05 & 4.44 & 7.28 & 47.98 \\
		&  &  &  &  &  &  \\
		\emph{Model 6} &  &  &  &  & \\
		Naive & 2.19 & 1.63 & 2.04 & 2.50 & 22.18 & 46.63 \\
		Naive Pair & 2.35 & 1.76 & 2.44 & 2.72 & 23.70 & 47.32  \\
		Adj & 4.01 & 4.06 & 4.52 & 6.14 & 32.25 & 60.98 \\
		IPW & 0.86 & 2.61 & 0.94 & 4.50 & 12.16 & 51.67 \\

		\hline
	\end{tabularx}
	
	\vspace{-1ex}
	\justify
	Notes: The table presents the rejection probabilities for tests of ATEs.
	The columns `Naive' and `Adj' correspond to the two-sample $t$-test and the adjusted $t$-test in \cite{BRS19}, respectively; the `Naive pair' column corresponds to the $t$-test using the standard errors estimated by the naive multiplier bootstrap of the pairs; the column `IPW' corresponds to the $t$-test using the standard errors estimated by the IPW multiplier bootstrap ATE estimator.
	\label{tab:sim_ate'}
\end{table}

\begin{table}[htbp]
	\caption{The Empirical Size of Tests for QTEs }
	\vspace{1ex}

	\centering{}
	\begin{tabularx}{1\textwidth}{LCCCCCCCC}
		\hline
		\hline
		\multicolumn{1}{l}{} & \multicolumn{8}{c}{$\mathcal{H}_{0}$: $\Delta=0$}\\
		\multicolumn{1}{l}{} & \multicolumn{4}{c}{$n=50$} & \multicolumn{4}{c}{$n=100$}\\
		\cmidrule(lr){2-5} \cmidrule(lr){6-9}
		\multicolumn{1}{l}{} & 0.25 & 0.50 & 0.75 & \multicolumn{1}{c}{Dif} & 0.25 & 0.50 & 0.75 & Dif\\
		\hline
		\emph{Model 5} &  &  &  &  &  &  &  & \\
		Naive & 2.68 & 1.48 & 2.30 & 1.60 & 2.68 & 1.45 & 1.59 & 1.70 \\
		Naive Pair & 2.87 & 1.58 & 2.40 & 1.89 & 2.94 & 1.60 & 1.75 & 1.78 \\
		Gradient & 4.44 & 3.21 & 3.18 & 2.41 & 4.13 & 3.87 & 4.41 & 3.64 \\
		IPW & 1.10 & 0.63 & 1.47 & 0.93 & 3.86 & 3.58 & 4.84 & 3.79 \\
		&  &  &  &  &  &  &  & \\
		\emph{Model 6} &  &  &  &  &  &  &  &  \\
		Naive & 3.31 & 1.88 & 2.96 & 1.82 & 2.55 & 2.14 & 2.63 & 1.85 \\
		Naive Pair & 3.44 & 2.20 & 3.08 & 2.08 & 2.72 & 2.29 & 2.68 & 2.07 \\
		Gradient & 4.46 & 3.34 & 2.96 & 1.68 & 4.37 & 4.25 & 4.25 & 2.60 \\
		IPW & 1.75 & 0.58 & 1.02 & 0.76 & 3.48 & 3.25 & 4.31 & 4.03 \\
		\hline
	\end{tabularx}
	
	\vspace{-1ex}
	\justify
	Note:	The table presents the empirical size of the tests for QTEs. The columns `0.25', `0.50', and `0.75' correspond to tests with quantiles at 25\%, 50\%, and 75\%, respectively; the column `Dif' corresponds to the test with the null hypothesis specified in \eqref{eq:diff}. The rows `Naive', `Naive pair', `Gradient', and `IPW' correspond to the results of the naive multiplier bootstrap, the naive multiplier bootstrap of the pairs, the gradient bootstrap, and IPW multiplier bootstrap, respectively.
	
	\label{tab:sim_qte_size'}
\end{table}

\begin{landscape}
	\begin{table}[htbp]
		\caption{The Empirical Power of Tests for QTEs }
		\vspace{1ex}

		\centering{}
		\begin{threeparttable}
			\begin{tabular}{lcccccccccccccccc}
				\hline
				\hline
				\multicolumn{1}{l}{} & \multicolumn{8}{c}{$\mathcal{H}_{1}$: $\Delta=1/2$} & \multicolumn{8}{c}{$\mathcal{H}_{1}$: $\Delta=2$}\\
				\multicolumn{1}{l}{} & \multicolumn{4}{c}{$n=50$} & \multicolumn{4}{c}{$n=100$} & \multicolumn{4}{c}{$n=50$} & \multicolumn{4}{c}{$n=100$}\\
				\cmidrule(lr){2-5} \cmidrule(lr){6-9} \cmidrule(lr){10-13}  \cmidrule(lr){14-17}
				\multicolumn{1}{l}{} & 0.25 & 0.50 & 0.75 & \multicolumn{1}{c}{Dif} & 0.25 & 0.50 & 0.75 & \multicolumn{1}{c}{Dif} & 0.25 & 0.50 & 0.75 & \multicolumn{1}{c}{Dif} & 0.25 & 0.50 & 0.75 & Dif\\
				\hline
				\emph{Model 5} &  &  &  &  &  &  &  &  &  &  &  &  &  &  &  & \\
				Naive & 5.76 & 1.93 & 2.15 & 1.54 & 8.93 & 2.62 & 1.51 & 1.56 & 53.24 & 24.46 & 4.66 & 4.60 & 85.86 & 57.75 & 8.18 & 8.96 \\
				Naive Pair & 5.86 & 2.13 & 2.35 & 1.68 & 9.25 & 2.80 & 1.57 & 1.56 & 53.15 & 25.54 & 4.94 & 4.86 & 86.05 & 58.52 & 8.40 & 9.31 \\
				Gradient & 7.99 & 4.33 & 3.38 & 2.71 & 11.84 & 7.28 & 4.26 & 3.67 & 59.77 & 36.90 & 7.78 & 7.32 & 89.24 & 74.33 & 17.62 & 17.56 \\
				IPW & 2.53 & 1.32 & 1.13 & 0.70 & 11.28 & 6.65 & 4.75 & 4.09 & 36.02 & 19.39 & 2.29 & 1.87 & 88.15 & 73.61 & 18.92 & 18.54 \\
				&  &  &  &  &  &  &  &  &  &  &  &  &  &  &  & \\
				\emph{Model 6} &  &  &  &  &  &  &  &  &  &  &  &  &  &  &  & \\
				Naive & 5.74 & 3.59 & 2.99 & 1.82 & 6.84 & 5.78 & 2.61 & 2.57 & 32.87 & 42.53 & 11.81 & 8.68 & 56.27 & 78.14 & 22.91 & 21.33 \\
				Naive Pair & 6.24 & 3.75 & 3.02 & 2.05 & 7.15 & 6.09 & 2.66 & 2.65 & 33.24 & 43.69 & 12.06 & 9.36 & 56.51 & 78.53 & 23.39 & 21.91 \\
				Gradient & 7.30 & 5.94 & 2.96 & 2.02 & 9.55 & 10.28 & 4.52 & 3.49 & 37.06 & 50.97 & 12.07 & 8.85 & 63.31 & 85.29 & 30.89 & 25.76 \\
				IPW & 3.63 & 0.94 & 0.97 & 0.99 & 8.35 & 8.20 & 5.14 & 5.55 & 24.52 & 23.90 & 5.67 & 6.44 & 60.06 & 81.93 & 33.75 & 34.67 \\
				\hline
			\end{tabular}

			\begin{tablenotes}[para,flushleft]
				Note:	The table presents the empirical power of the tests for QTEs. The columns `0.25', `0.50', and `0.75' correspond to tests with quantiles at 25\%, 50\%, and 75\%, respectively; the column `Dif' corresponds to the test with the null hypothesis specified in \eqref{eq:diff}. The rows `Naive', `Naive pair', `Gradient', and `IPW' correspond to the results of the naive multiplier bootstrap, the naive multiplier bootstrap of the pairs, the gradient bootstrap, and IPW multiplier bootstrap, respectively.
			\end{tablenotes}

		\end{threeparttable}
		\label{tab:sim_qte_power'}
	\end{table}
	
\end{landscape}

\newcolumntype{L}{>{\hsize=1.3\hsize\raggedright\arraybackslash}X}
\newcolumntype{C}{>{\hsize=0.95\hsize\centering\arraybackslash}X}

\begin{table}[H]
	\caption{The Empirical Size and Power of Uniform Inferences for QTEs}
	\vspace{1ex}
	
	\centering{}%
	\begin{tabularx}{1\textwidth}{LCCCCCC}
		\hline
		\hline
		& \multicolumn{2}{c}{$\mathcal{H}_{0}$: $\Delta=0$} & \multicolumn{2}{c}{$\mathcal{H}_{1}$: $\Delta=1/2$} & \multicolumn{2}{c}{$\mathcal{H}_{1}$: $\Delta=2$}\\
		\cmidrule(lr){2-3} \cmidrule(lr){4-5} \cmidrule(lr){6-7}
		& \multicolumn{1}{c}{$n=50$} & \multicolumn{1}{c}{$n=100$} & \multicolumn{1}{c}{$n=50$} & \multicolumn{1}{c}{$n=100$} & \multicolumn{1}{c}{$n=50$} & \multicolumn{1}{c}{$n=100$} \\
		\hline
		\emph{Model 5} &  &  &  &  &  & \\
		Naive & 1.53 & 1.26 & 1.99 & 3.19 & 32.12 & 75.44\\
		Naive Pair & 1.59 & 1.51 & 2.14 & 3.43 & 30.52 & 76.10 \\
		Gradient & 2.37 & 3.13 & 3.25 & 6.09 & 40.45 & 83.37 \\
		IPW & 0.25 & 3.67 & 0.41 & 6.24 & 11.04 & 82.91 \\
		&  &  &  &  &  & \\
		\emph{Model 6} &  &  &  &  &  & \\
		Naive & 1.11 & 1.34 & 2.40 & 3.81 & 32.54 & 73.05 \\
		Naive Pair & 1.23 & 1.32 & 2.56 & 4.12 & 33.42 & 73.59 \\
		Gradient & 1.62 & 2.43 & 3.57 & 7.21 & 40.37 & 82.48 \\
		IPW & 0.16 & 2.13 & 0.46 & 5.91 & 11.44 & 78.92 \\
		
		\hline
	\end{tabularx}
	
	\vspace{-1ex}
	\justify
	Notes: The table presents the rejection probabilities of the uniform confidence bands for the hypothesis specified in \eqref{eq:uni}. The rows `Naive', `Naive pair', `Gradient', and `IPW' correspond to the results of the naive multiplier bootstrap, the naive multiplier bootstrap of the pairs, the gradient bootstrap, and the IPW multiplier bootstrap, respectively.
	
	\label{tab:sim_uniform'}
\end{table}

\bibliographystyle{chicago}
\bibliography{BCAR}

\end{document}